\documentclass[10pt]{article}

\usepackage{amssymb}
\usepackage{latexsym}
\usepackage{amsmath}
\usepackage{euscript}
\usepackage{graphics}
\usepackage{color}
\usepackage{pdfcolmk}

\newcommand{\green}{\color{black}}

\newcommand\gH{{\mathfrak{H}}}
\newcommand\gotH{{\mathfrak{H}}}

\newcommand\gotm{{\mathfrak{m}}}

\newcommand\gotN{{\mathfrak{N}}}

\newcommand\gotS{{\mathfrak{S}}}
\newcommand\gt{{\mathfrak{t}}}

\newcommand{\gd}{{\delta}}
\newcommand{\gD}{{\Delta}}
\newcommand{\gga}{{\gamma}}
\newcommand{\gG}{{\Gamma}}

\newcommand{\gl}{{\lambda}}

\newcommand{\gs}{{\sigma}}
\newcommand{\gS}{{\Sigma}}
\newcommand{\gT}{{\Theta}}
\newcommand{\gO}{{\Omega}}
\newcommand{\gP}{{\Pi}}

\newcommand\R{{\mathbb{R}}}
\newcommand\C{{\mathbb{C}}}

\newcommand\N{{\mathbb{N}}}

\newcommand\cA{{\mathcal{A}}}
\newcommand\cB{{\mathcal{B}}}
\newcommand\cC{{\mathcal{C}}}
\newcommand\cD{{\mathcal{D}}}
\newcommand\cE{{\mathcal{E}}}
\newcommand\cF{{\mathcal{F}}}

\newcommand\cH{{\mathcal{H}}}

\newcommand\cL{{\mathcal{L}}}

\newcommand\cN{{\mathcal{N}}}
\newcommand\cO{{\mathcal{O}}}

\def\h#1{{{\hat #1} }}
\def\wt#1{{{\widetilde #1} }}
\def\wh#1{{{\,\widehat #1\,} }}

\def\bm\chi{\mbox{\boldmath$\chi$}}

\def\IM{{\rm Im\,}}
\def\real{{\rm Re\,}}
\def\imag{{\rm Im\,}}

\def\Ext{{\rm Ext\,}}

\def\ess{{\rm ess\,}}

\def\ran{{\rm ran\,}}
\def\dom{{\rm dom\,}}
\def\mul{{\rm mul\,}}

\def\dim{{\rm dim\,}}

\def\rank{{\rm rank\,}}

\let\xker=\ker \def\ker{{\xker\,}}
\def\span{{\rm span\,}}

\def\supp{{\rm supp\,}}

\def\mod{{\rm mod}}
\def\dim{{\rm dim}}
\def\cl{{\rm cl}}

\def\sign{{\rm sign}}

\DeclareMathOperator\re{Re}
\DeclareMathOperator\im{Im}

\newtheorem{theorem}{Theorem}[section]
\newtheorem{proposition}[theorem]{Proposition}
\newtheorem{corollary}[theorem]{Corollary}
\newtheorem{lemma}[theorem]{Lemma}
\newtheorem{definition}[theorem]{Definition}

\newtheorem{remark}[theorem]{Remark}

\numberwithin{equation}{section}
\newcommand{\ba}{\begin{array}}
\newcommand{\ea}{\end{array}}
\newcommand{\bea}{\begin{eqnarray}}
\newcommand{\eea}{\end{eqnarray}}
\newcommand{\bead}{\begin{eqnarray*}}
\newcommand{\eead}{\end{eqnarray*}}
\newcommand{\be}{\begin{equation}}
\newcommand{\ee}{\end{equation}}
\newcommand{\bed}{\begin{displaymath}}
\newcommand{\eed}{\end{displaymath}}
\newcommand{\bl}{\begin{lemma}}
\newcommand{\el}{\end{lemma}}
\newcommand{\bp}{\begin{proposition}}
\newcommand{\ep}{\end{propostion}}
\newcommand{\bt}{\begin{theorem}}
\newcommand{\et}{\end{theorem}}
\newcommand{\Label}{\label}
\newcommand{\bc}{\begin{corollary}}
\newcommand{\ec}{\end{corollary}}
\newcommand{\la}{\Label}

\newcommand{\br}{\begin{remark}}
\newcommand{\er}{\end{remark}}
\newcommand{\bd}{\begin{definition}}
\newcommand{\ed}{\end{definition}}

\newcommand{\slim}{\,\mbox{\rm s-}\hspace{-2pt} \lim}
\newcommand{\wlim}{\,\mbox{\rm w-}\hspace{-2pt} \lim}
\newcommand{\olim}{\,\mbox{\rm o-}\hspace{-2pt} \lim}

\newenvironment{proof}%
{\begin{sloppypar}\noindent{\bf Proof.}}%
{\hspace*{\fill}$\square$\end{sloppypar}}

\title{
On the unitary equivalence of
absolutely continuous parts of self-adjoint extensions}

\author{Mark~Malamud\\
Institute of Applied Mathematics and Mechanics \\
Universitetskaya str. 74 \\
83114 Donetsk \\
Ukraine\\
E-mail: mmm@telenet.dn.ua \and
Hagen~Neidhardt\\
Institut f\"ur Angewandte Analysis und Stochastik\\
Mohrenstr. 39\\
D-10117 Berlin\\
Germany\\
E-mail: neidhard@wias-berlin.de
}

\begin{document}

\maketitle

\vspace{-7mm}

\vspace{-5mm}

\noindent {\bf Abstract:} The  classical Weyl-von~Neumann theorem
states that for any self-adjoint operator $A$ in a separable
Hilbert space $\gotH$ there exists  a (non-unique) Hilbert-Schmidt
operator $C = C^*$ such that the perturbed
operator $A+C$ has purely
point spectrum. We are interesting whether this result remains
valid for non-additive perturbations by considering self-adjoint
extensions of a given densely defined symmetric operator $A$ in
$\mathfrak H$ and fixing an extension $A_0 = A_0^*$. We show that for
a wide class of symmetric operators the absolutely continuous
parts of extensions $\widetilde A = {\widetilde A}^*$ and $A_0$
are unitarily equivalent provided that their  resolvent difference
is a compact operator. Namely, we show that this is
true whenever the Weyl function $M(\cdot)$  of a pair $\{A,A_0\}$
admits bounded limits $M(t) := \wlim_{y\to+0}M(t+iy)$ for a.e. $t \in
\mathbb{R}$. This result is applied to direct sums of
symmetric operators and Sturm-Liouville operators with operator potentials.


\section{Introduction} \label{intro}

Let  $A_0$ be a self-adjoint operator in a separable Hilbert space
$\gH$ and let $C=C^*$ be a trace class operator in $\gH$,
$C\in\mathfrak S_1(\gH)$. Recall, that according to the
Kato-Rosenblum theorem, cf. \cite{Kat57,Ros57} the absolutely
continuous parts  $A^{ac}_0$ and ${\wt A}^{ac}$, { in short the
$ac$-parts}, of the operators $A_0$ and ${\wt A}=A_0+C$ {are
unitarily equivalent. In other words, the { absolutely continuous
spectrum}, { in short $ac$-spectrum}, of $A_0$ and its spectral
multiplicity { are} stable under { \emph{additive}} trace class
perturbations. At the same time, the Weyl-von~Neumann-Kuroda
theorem \cite[Theorem 94.2]{AG81}, \cite{Neu35}, \cite{Kur58} 
shows that the condition
$C\in\mathfrak S_1(\gH)$ cannot be replaced by $C\in\mathfrak
S_p(\gH)$ with $p\in (1,\infty]$ (where $\mathfrak S_p(\gH)$
denotes the Neumann-Schatten operator ideals).
\bt[{\cite[Theorem 10.2.1 and Theorem 10.2.3]{Ka76}}]\label{WNth}
For any operator $A_0=A^*_0$ in $\gH$ and any  $p\in (1,\infty]$
there exists an operator $C= C^*\in \mathfrak S_p(\gH)$  such that
the perturbed operator ${\wt A} = A_0 + C$ has purely point
spectrum. In particular, $\sigma_{ac}(A_0 + C)=\emptyset.$
\et

The Kato-Rosenblum theorem was generalized by Birman
\cite{Bir63a} and Birman and Krein \cite{BirKrei62}
to the case of \emph{non-additive} perturbations. Namely, it was
shown that $A^{ac}_0$ and ${\wt A}^{ac}$ still remain unitary
equivalent whenever
\bed
({\wt A}-i)^{-1}-(A_0-i)^{-1}\in\mathfrak S_1(\gH).
\eed
In particular, this is true if $A_0=A^*_0$ and ${\wt A}={\wt A}^*$
are self-adjoint extensions of a
symmetric operator $A$ (in short $A_0,  {\wt A}\in \Ext_A$).
This rises the following Weyl-von~Neumann problem for
extensions: Given  $p \in (1,\infty]$ and a self-adjoint
  extension   $A_0$ of $A.$  Does  there exist
 a self-adjoint extension  $\wt A$ of $A$ such that $\wt A$ has
purely point spectrum and the difference $(\wt A -
i)^{-1} - (A_0 - i)^{-1}$ belongs to $\mathfrak S_p(\gotH)$? To
the best of our knowledge this problem was not investigated.

In the present paper we show that the Weyl-von~Neumann theorem for
extensions becomes false in general. We show that under an
additional assumption on the symmetric operator $A$ the $ac$-part of a
certain extension $A_0=A^*_0$  is unitarily equivalent to the $ac$-part of any
extension ${\wt A}={\wt A}^*$ of $A$ provided that their
resolvent difference is compact, that is,
\begin{equation}\label{0.2}
K_{\wt A} := ({\wt A}-i)^{-1}-(A_0-i)^{-1}\in{\mathfrak
S}_{\infty}(\gH).
\end{equation}
The additional assumption on the pair $\{A,A_0\}$ is formulated in
terms of the Weyl function of the pair $\{A,A_0\}$.  The
latter is the main object in the boundary triplet approach to the
extension theory developed in the last three decades, see
\cite{DM91,DM95,GG91} and references therein.

The core of this approach is the following abstract version of
Green's formula
\be\label{0.3}
 (A^*f,g) - (f,A^*g) = (\gG_1f,\gG_0g)_{\cH} -
(\gG_0f,\gG_1g)_{\cH}, \qquad f,g\in\dom(A^*),
      \ee
where $\cH$  is an auxiliary Hilbert space and
$\Gamma_0,\Gamma_1 : \dom(A^*)\rightarrow \cH$  are linear mappings.  A triplet
$\Pi=\{\cH,\Gamma_0,\Gamma_1\}$ is called a boundary triplet  for
the operator $A^*$ if \eqref{0.3} holds and the mapping
$\Gamma:=\{\Gamma_0,\Gamma_1\}: \dom(A^*)\rightarrow \cH\oplus
\cH$ is surjective.

With a boundary triplet $\gP$ for $A^*$ one associates in a
natural way the Weyl function $M(\cdot)= M_{\Pi}(\cdot)$ (see
Definition \ref{Weylfunc}), which is the key object of this approach. It
is an operator-valued Nevanlinna function with values in $[\cH]$
(i.e. $R_{\cH}$-function) and its role in the extension theory is
similar to that of the classical Weyl function in the spectral
theory of Sturm-Liouville operators. In particular, if $A$ is
simple, then $M(\cdot)$ determines the pair $\{A,A_0\}$, where
$A_0 :=A^*\upharpoonright \ker\gG_0$, uniquely, up to
unitary equivalence. Moreover, $M(\cdot)$ is regular
(holomorphic) precisely on the resolvent set $\varrho(A_0)$ of
$A_0$  and the spectral properties of $A_0$ are described in terms
of the limits $M(t+i0)$ at the real line (see \cite{BMN02}).

One of our main results (Theorem \ref{V.5}) reads now as
follows.
\bt\label{th0.1}
Let $\Pi=\{\cH,\Gamma_0,\Gamma_1\}$ be a boundary triplet
for $A^*$  { such that} the
corresponding Weyl function $M(\cdot)$ has  weak limits
\begin{equation}\label{0.2A}
M(t+i0) :=\wlim_{y\downarrow 0}M(t+iy)\quad \text{for a.e. }\ t\in{\mathbb R}.
 \end{equation}
If a self-adjoint extension $\wt A$ of $A$ satisfies condition
\eqref{0.2}, then the { $ac$-parts} ${\wt A}^{ac}$ and
$A^{ac}_0$  of  ${\wt A}$ and { $A_0 \,(=
A^*\upharpoonright\ker(\gG_0)$)} are unitarily
equivalent.
\et

We apply this result to direct sums $A :=\oplus^{\infty}_{n=1}
S_n$ of symmetric operators $S_n$ { with equal and finite
deficiency indices $n_\pm(S_n)$}. { Let $S_{0n}$ be a
self-adjoint extension of $S_n$ for each $n \in \N$. We show  that
the $ac$-part of $A_0 :=\oplus^{\infty}_{n=1} S_{0n}$ is unitarily
equivalent to the $ac$-part of any other extension $\wt
A=\wt A^*\in \Ext_A$  provided that condition \eqref{0.2} is
satisfied and the symmetric operators $S_n$ are unitarily
equivalent to $S_1$ for any $n \in \N$.}

The second part of the paper is concerned with  a spectral
extremal property of certain self-adjoint extensions of $A$
described by the following definition.
\begin{definition}\la{0.3AA}
{\em
(i)
Let $T_j=T_j^*\in \cC(\gotH_j)$, $j = 1,2$. We
say that $T_1$ is a part of  $T_2$  if there is an isometry $V$
from $\gotH_1$ into $\gotH_2$ such that $VT_1V^* \subseteq T_2$.

(ii) { Let $A_0=A^*_0$ be an extension of $A$.
We say that $A_0$ is $ac$-minimal if $A^{ac}_0$ is a part of any
self-adjoint extension ${\wt A}$ of $A$.}

(iii) Let $\gs_0 :=\gs_{ac}(A_0)$. We say that $A_0$ is strictly $ac$-minimal if for
any extension $\wt A=\wt A^*$ of $A$ the parts $A^{ac}_0$ and $\wt
A^{ac}E_{\wt A}(\sigma_0)$ are unitarily equivalent.
}
\end{definition}

In particular, if $A_0$ is $ac$-minimal, then
$\sigma_{ac}({\wt A})\supseteq \sigma_{ac}( A_0)$. Note
that an $ac$-minimal extension of $A$ is not unique.
For any two  $ac$-minimal extensions their $ac$-parts are
unitarily equivalent.

We show (cf. Theorem \ref{VI.6}) that if $n_{\pm}(S_n)<\infty,$
then the $ac$-part $A^{ac}_0$ of any direct sum extension
$A_0 =\oplus^{\infty}_{n=1} S_{0n}$ of $A :=\oplus^{\infty}_{n=1}
S_{n}$ is $ac$-minimal. In particular,
$\sigma_{ac}({\wt A})\supseteq \sigma_{ac}( A_0)$ for any ${\wt
A}={\wt A}^*\in \Ext_A.$
This result looks surprising with respect to Theorem \ref{WNth}.
Indeed, in this case $A^{ac}_0$ is still  a part of $\wt
A^{ac}$ for any ${\wt A}\in \Ext_A$ though the resolvent
difference $K_{\wt A}$ (see \eqref{0.2}) is not even compact.
In other words, in this case the $ac$-spectrum of $A_0$
(but not its spectral multiplicity) remains stable under
(non-additive) compact perturbations $K_{\wt A}$ though
both $\sigma_{ac}( A_0)$ and its multiplicity  can only increase,
whenever} $K_{\wt A}\notin\gotS_{\infty}.$

Moreover, we apply our technique to  minimal symmetric
non-negative  Sturm-Liouville operator $A$ with an unbounded
operator potential
\begin{eqnarray}\la{0.4}
(Af)(x)   =  -f''(x) + Tf(x).
\end{eqnarray}
We show that the Friedrichs extension $A^F$ is $ac$-minimal
and under a simple additional assumption is even strictly
$ac$-minimal.

The paper is organized as follows. In Section 2 we give a short
introduction into the theory of ordinary and generalized boundary
triplets and  the corresponding Weyl functions. In Section 3 we
express the spectral multiplicity function of the { $ac$-part
$\wt A^{ac}$ of  $\wt A= \wt A^*(\in \Ext_A)$} by means of the
corresponding Weyl function. In Section 4 we apply this technique
to prove Theorem \ref{th0.1} as well as to give a simple proof of
the Kato-Rosenblum theorem.

In Section 5 direct sums of boundary triplets $\gP_n =
\{\cH_n,\gG_{0n},\gG_{1n}\}$  for operators  $S_n^*$ adjoint to
symmetric operators $S_n$ are investigated.  We show that though,
in general, $\gP = \oplus^\infty_{n=1} \gP_n$ is not a boundary
triplet for the direct sum $A^* :=\oplus^{\infty}_{n=1} S^*_n,$ it
is always possible to modify the triplets $\gP_n$ in such a way
that a new sequence $\wt \gP_n= \{\cH_n,\wt\gG_{0n},\wt\gG_{1n}\}$
of boundary triplets for $S_n^*$ satisfies the following
properties: $\wt \gP = \oplus^\infty_{n=1} \wt \gP_n$ forms a
boundary triplet for $A^*$ { such that $S_{0n} :=
S^*_n\upharpoonright \ker(\gG_{0n}) = S^*_n\upharpoonright
\ker(\wt \gG_{0n}) =: {\wt S}_{0n},\  n \in \N$. In particular,
the corresponding Weyl function  $M(\cdot)$ is {
block}-diagonal (see  Theorem \ref{VI.3}).  { Our spectral
applications to direct sums are substantially based on this
result. In particular, it is used in proving of Theorem \ref{VI.6}
mentioned above}.

Finally, in Section 6 we apply the  technique (and abstract
results) to  operators \eqref{0.4}  with { bounded and}
unbounded operator potentials. { In particular, we
investigate the $ac$-spectrum of self-adjoint realizations of
Schr\"odinger operator
\bed \cL = - \left(\frac{\partial^2}{\partial t^2} +
\sum^n_{j=1}\frac{\partial^2}{\partial x^2_j}\right) + q(x). \quad
(t,x) \in \R_+ \times \R^n, \quad q \in L^\infty(\R^n),
      \eed
in $L^2(\R_+\times\R^n)$, $n \ge 1$. For instance, we show that if
$q(\cdot)\ge 0$ and}
    \begin{equation}\label{0.33}
\lim_{|x|\to\infty}\int_{|x-y|\le 1}|q(y)|dy = 0,
    \end{equation}
{ then the Dirichlet realization   $L^D$ is absolutely
continuous,  strictly $ac$-minimal and  $\gs(L^D) = \gs_{ac}(L^D)=
[0,\infty).$}

{\bf Notations}  In the following we consider only separable
Hilbert spaces which are denoted by $\mathfrak H$, $\cH$ etc.
{ The symbols $\cC(\cH_1, \cH_2)$ and $[\mathfrak H_1,
\mathfrak H_2]$ stand for the set of closed densely defined {
linear} operators and the set of bounded linear operators from
$\mathfrak H_1$ to $\mathfrak H_2,$ respectively}.
{ We set} $\cC(\cH):= \cC(\cH, \cH)$ and $[\mathfrak H] :=
[\mathfrak H, \mathfrak H]$.  { The symbols $\dom(\cdot),$
$\ran(\cdot),$ $\varrho(T)$ and $\sigma(T)$ stand for the domain,
the range,  the resolvent set and the spectrum of an operator
$T\in \cC(\cH),$ respectively}; $T^{ac}$ and
$\sigma_{ac}(T)$ stand for the $ac$-part and the $ac$
spectrum of an operator $T=T^*\in \cC(\cH).$

$\mathfrak S_p(\gH),$  $p\in [1,\infty],$ stand for the
Schatten-von~Neumann ideals in $\gH$.
{ Denote by $\cB({\R})$ the Borel $\sigma$-algebra of the
line $\R$  and  by $\cB_b({\R})$ the algebra of bounded subsets
in} $\cB_b(\R)$. The Lebesgue measure of a set $\gd \in \cB({\R})$
is denoted by $|\gd|$.

\section{Preliminaries}

\subsection{Operator measures}
\bd
{\em
Let  $\cH$ be a separable Hilbert space. A  mapping
$\gS(\cdot):\  \cB_b({\R})\to[\cH]$ is called an operator
(operator-valued) measure if

\item[\;\;{\rm (i)}] $\gS(\cdot)$ is $\gd$-additive in the strong
    sense and

\item[\;\;{\rm (ii)}] $\gS(\gd)=\gS(\gd)^*\ge 0$ for $\gd\in
\cB_b({\R})$.

The operator measure $\gS({\cdot})$ is called bounded if it
extends to the Borel algebra $\cB({\R})$ of $\R$, i.e.
$\gS({\R})\in[\cH]$. Otherwise,  it
is called unbounded. A bounded operator measure
$\gS(\cdot)=E(\cdot)$ is called orthogonal if, in addition the conditions

\item[\;\;{\rm (iii)}] $E(\gd_1)E(\gd_2)=E(\gd_1\cap\gd_2)$ for
$\gd_1,\gd_2\in \cB({\R})$ and
$E({\R})=I_{\cH}$

are satisfied.
}
\ed

Setting in (iii) $\gd_1=\gd_2$, one gets that an orthogonal
measure $E(\cdot)$ takes its values in the set of orthogonal
projections on $\cH$. Every orthogonal measure $E(\cdot)$
defines an operator $T= T^*= \int_{{\R}}\gl  d E(\gl)$ in
$\cH$ with $E(\cdot)$ being its spectral measure. Conversely, by
the spectral theorem, every operator $T=T^*$ in $\cH$ admits the
above representation with the orthogonal spectral measure
$E=:E_T$.

By $\gS^{ac}$, $\gS^s$, $\gS^{sc}$ and $\gS^{pp}$  we denote
absolutely continuous, singular, singular continuous and pure
point parts of the measure $\gS$, respectively. The Lebesgue
decomposition of $\gS$ is given by
$\gS = \gS^{ac} + \gS^s = \gS^{ac} + \gS^{sc} + \gS^{pp}$.

The operator measure $\gS_1$ is called subordinated to
the operator measure $\gS_2$, in short $\gS_1 \prec \gS_2$, if
$\gS_2(\gd) = 0$ yields $\gS_1(\gd) = 0$ for $\gd \in \cB_b(\R)$. If
the measures $\gS_1$ and $\gS_2$ are mutually subordinated, then
they are called equivalent, in short $\gS_1 \thicksim \gS_2$.
Note, that  there are always exists a scalar measure $\rho$
defined on $\cB_b(\R)$ such that $\gS \thicksim \rho$, see
\cite[Remark 2.2]{MM03}. In particular, there is always a scalar
measure such that $\gS \prec \rho$.

Usually, with the operator-valued measure $\gS(\cdot)$ one
associates a distribution operator-valued function $\gS(\cdot)$
defined by
\be\la{a1}
\gS(t) =
\begin{cases}
\gS([0,t)) & t > 0 \\
0            & t = 0 \\
-\gS([t,0)) & t < 0
\end{cases}
\ee
which is called the spectral function of $\gS$.
Clearly, $\gS(\cdot)$ is strongly left continuous, $\gS(t-0)=
\gS(t),$ and satisfies
$\gS(t) = \gS(t)^*, \quad \gS(s) \le \gS(t)$, $s \le t$.
\begin{definition}[{\cite[Definition 4.5]{MM03}}]
{\em
Let $\gS$ be an operator measure in
$\cH$ and let $\rho$ be a scalar measure on $\cB(\R)$ such that
$\gS \prec \rho$. Further, let $e=\{e_j\}^{\infty}_{j=1}$
be an orthonormal basis in $\cH$. Let
\bead
\gS_{ij}(t):=\bigl(\gS(t)e_i,e_j\bigr), \qquad
\Psi_{ij}(t):=d\gS_{ij}(t)/d\rho,\\
\Psi^e_n(t):=\bigl(\Psi_{ij}(t)\bigr)^n_{i,j=1}, \qquad
\Psi^e(t):=\bigl(\Psi_{ij}(t)\bigr)^{\infty}_{i,j=1}.
\eead
We call
\be\label{2.1A}
N^e_{\gS}(t):=\rank(\Psi^e(t)) := \sup_{n\ge 1} \rank(\Psi^e_n(t))\;(\mod(\rho))
\ee
and
\bed
N_\gS(t) := \ess\!\sup_e N^e_{\gS}(t)\;(\mod(\rho))
\eed
the multiplicity function and the total multiplicity of $\gS$, respectively.
}
\end{definition}

By  \cite[Proposition 4.6]{MM03} $N^e_{\gS}(\cdot)$ does
not depend on the orthogonal basis $e$. Therefore one always has
$N_{\gS}(t) := N^e_{\gS}(t)$  and one can omit the index $e$ in
\eqref{2.1A}.

When  applying  this definition to the absolutely
continuous part $\gS^{ac}$ of $\gS$  the scalar measure
$\rho^{ac}$  can be chosen  to be  the Lebesgue measure $|\cdot|$
on $\cB(\R)$.

The concept of the multiplicity function allows one to introduce
the following definitions.
\begin{definition}\label{II.3}
{\em
Let $\gS_1$ and $\gS_2$ be two operator measures.

\item[\;\;\rm (i)]
The operator measure $\gS_1$ is called
spectrally subordinate to the operator measure $\gS_2$, in short
$\gS_1\prec\!\!\prec\gS_2$, if $\gS_1\prec\gS_2$
and $N_{\gS_1}(t)\le N_{\gS_2}(t)\;(\mbox{\mod}(\gS_2))$.
\item[\;\;\rm (ii)]
The operator measures $\gS_1$ and $\gS_2$ are called spectrally
equivalent, in short $\gS_1 \thickapprox \gS_2$, if $\gS_1 \thicksim \gS_2$ and
$N_{\gS_1}(t)=N_{\gS_2}(t)(\mbox{\mod}(\gS_2))$.
}
\end{definition}

{ Crucial for us in the sequel is the following theorem}.
\bt\la{II.4a}
{ Let $T_j$ be self-adjoint operators acting in $\gotH_j$
with corresponding spectral measures $E_{T_j}(\cdot)$, $j =1,2$.
Let $\cD \in \cB(\R)$.

\item[\;\;\rm (i)]
$T_1E_{T_1}(\cD)$ is a part of $T_2E_{T_2}(\cD)$ if and only if
$E_{T_1,\cD} \prec\!\!\prec  E_{T_2,\cD}$, where $E_{T_j,\cD}(\gd) :=
E_{T_j}(\gd\cap\cD), \ j = 1,2$.

\item[\;\;\rm (ii)]
The parts $T_1E_{T_1}(\cD)$ and $T_2E_{T_2}(\cD)$ are unitarily equivalent
if and only if $E_{T_1,\cD}  \thickapprox E_{T_2,\cD}$.}
\et

The proof is immediate from \cite[Theorem 7.5.1]{BS87}.
For  $\cD = \R$ Theorem \ref{II.4a} gives conditions for
$T_1$ to be unitarily equivalent either to a part of $T_2$  or to
$T_2$ itself.

\subsection{$R$-Functions}

Let ${\cH}$ be a separable Hilbert space. We recall that an operator-valued
function $F(\cdot)$ with values in $[{\cH}]$ is called to be a
Herglotz, Nevanlinna or $R$-function
\cite{AG81,Ber68,GG91,KN77}, if it is holomorphic in $\C_+$ and
its imaginary part is non-negative, i.e.
$\im(F(z)) := (2i)^{-1}\bigl({F(z) - F(z)^*}\bigr) \ge 0, \  z\in
\C_+$. In what follows we prefer the notion of R-function. The class of
R-functions with values in $[\cH]$ will be denoted by $(R_\cH)$.
Any $(R_\cH)$-function $F(\cdot)$ admits an integral
representation
\begin{equation}\label{1.2}
F(z) = C_0 + C_1 z +
\int_{-\infty}^{\infty} \left(\frac{1}{t-z}- \frac{t}{1 + t^2}\right) d\gS_F,\quad z\in \C_+,
\end{equation}
(see, for instance, \cite{AG81, Ber68, KN77}),  where $C_0 = C_0^*, \ C_1\ge 0$ and $\gS_F$ is an
operator-valued Borel measure on $\R$ satisfying $\int_\R
(1+t^2)^{-1}d\gS_F \in [\cH]$. The integral is understood in the
strong sense.

In contrast to spectral measures of self-adjoint operators the
measure $\gS_F$ is not necessarily orthogonal. However, the
operator-valued measure $\gS_F$ is uniquely determined by the
$R$-function $F(\cdot)$.  It is called the spectral
measure of $F(\cdot)$. The associated spectral function is denoted
by $\gS_F(t)$, $t \in \R$, cf. { \eqref{a1}}.

Let us  calculate $N_{\gS^{ac}_F}(t)$, $t \in \R$. For any
Hilbert-Schmidt operator ${ D} \in {\mathfrak S}_2(\cH)$
satisfying $\ker({D}) = \ker({D}^*) = \{0\}$ let us consider the
modified $R_\cH$-function
\bed
(F^{ D})(z) := { D}^*F(z){ D}, \quad z \in \C_+.
\eed
For $F^D(\cdot)$ the strong
limit $F^D(t) := F^D(t+i0) :=\slim_{y\to+\infty}F^D(t+iy)$
exists for a.e. $t \in \R$. We set
\be\la{2.13}
 d_{F^D}(t):= \dim(\ran(\im(F^D)(t))), \quad \mbox{for a.e.}  \quad t \in \R.
\ee
\begin{proposition}\la{II.4}
Let $F(\cdot)\in (R_\cH),$ \ $D \in {\mathfrak S}_2(\cH)$
and $\ker(D) = \ker(D^*) = \{0\}.$ Then
$N_{\gS^{ac}_F}(t) = d_{F^{ D}}(t)$ for a.e. $t \in \R$.
\end{proposition}
\begin{proof}
It follows from \eqref{1.2} that
\be\label{2.12}
 \im(F(\gl + iy)) = y C_1 +\int^\infty_{-\infty}\frac{y}{(t-\gl)^2 +y^2}d\gS_F,
\qquad \gl \in \R.
\ee
By Berezanskii-Gel'fand-Kostyuchenko theorem \cite{Ber68,BS87} the
derivative $\Psi_{D^*\gS_FD}(t) := \frac{d}{dt}D^*\gS_F(t)D$ exists for a.e.
$t \in \R$ and the representation
\bed
 D^*\gS^{ac}_F(\gd)D = \int_\gd \Psi_{D^*\gS_FD}(t) dt,
\qquad \gd \in \cB_b(\R)
\eed
 holds. Applying the  Fatou theorem (see \cite{KN77}) to \eqref{2.12} and using
\eqref{2.13} we obtain
\be\la{2.16}
\im((F^D)(\gl)) = \pi\Psi_{D^*\gS_FD}(\gl) \quad \text{for
a.e.}\quad \gl \in \R.
\ee
By  \cite[Corollary 4.7]{MM03} $N_{\gS^{ac}_F}(\gl) =
\rank(\Psi_{D^*\gS_FD}(\gl)) = \dim(\overline{\ran(\Psi_{D^*\gS_FD}(\gl))})$ for a.e.
$\gl \in \R$. Finally, using \eqref{2.16} we get $N_{\gS^{ac}_F}(\gl) =
d_{F^D}(\gl)$ for a.e. $\gl \in \R$.
\end{proof}

{ Notice that Proposition \ref{II.4} implies that $D_{F^D}(t)$ does not
dependent on $D$}. Assuming
the existence of the  limit $F(t) := \slim_{y\to+0}F(t + iy)$
for a.e. $t \in \R$,  we set
\bed
d_F(t) := \rank(\im(F(t)) = \dim(\ran(\im(F(t))))
\eed
for a.e. $t\in \R$. In this case Proposition \ref{II.4}
can be modified as follows.
\bc\la{II.5}
Let $F(\cdot)\in (R_\cH)$. If
the limit $F(t) := \slim_{y\to+0}F(t + iy)$
exists for a.e. $t \in \R$, then $N_{\gS^{ac}_F}(t) = d_F(t)$ for
a.e. $t \in \R$.
\ec

\subsection{Boundary triplets and self-adjoint extensions}

In this section we briefly
recall the basic facts on  boundary triplets and the corresponding
Weyl functions, cf. \cite{DM87,DM91,DM95,GG91}.

Let $A$ be a densely defined closed symmetric operator in the
separable Hilbert space $\gH$ with equal deficiency indices
$n_\pm(A)=\dim(\ker(A^*\mp i)) \leq \infty$.
\begin{definition}[\cite{GG91}]\la{t1.9}
{\em
A triplet $\Pi = \{\cH, \Gamma_0,
\Gamma_1\}$, where $\cH$ is an auxiliary Hilbert space and
$\Gamma_0,\Gamma_1:\  \dom(A^*)\rightarrow \cH$ are linear
mappings,  is called an (ordinary) boundary triplet for $A^*$ if
the "abstract Green's identity"
\begin{equation}\la{2.10A}
(A^*f,g) - (f,A^*g) = (\gG_1f,\gG_0g)_{\cH} -
(\gG_0f,\gG_1g)_{\cH}, \qquad f,g\in\dom(A^*),
\end{equation}
holds and the mapping $\gG:=(\Gamma_0,\Gamma_1)^\top:  \dom(A^*)
\rightarrow \cH \oplus \cH$ is surjective.
}
\end{definition}
\begin{definition}[\cite{GG91}]\la{def2.2}
{\em
A closed extension $A'$ of $A$ is called a proper extension,
in short $A'\in \Ext_A,$ if $A\subset A' \subset A^*;$

Two proper extensions $A', A''$ are called disjoint if
$\dom( A')\cap \dom( A'') = \dom( A)$ and  transversal if in
addition $\dom( A') + \dom( A'') = \dom( A^*).$
}
\end{definition}

Clearly, any  self-adjoint extension $\wt A= {\wt A}^*$ is proper,
 $\wt A\in \Ext_A.$
A boundary triplet $\Pi=\{\cH,\gG_0,\gG_1\}$ for $A^*$ exists
whenever  $n_+(A) = n_-(A)$. Moreover, the relations $n_\pm(A) =
\dim(\cH)$ and $\ker(\Gamma_0) \cap \ker(\Gamma_1)=\dom(A)$ are
valid. Besides, $\Gamma_0, \Gamma_1\in [{\mathfrak H}_+, \cH]$,
where ${\mathfrak H}_+$ denotes the Hilbert space { obtained
by equipping} $\dom(A^*)$ with the graph norm of $A^*$.

With any boundary triplet $\Pi$ one associates  two extensions
$A_j:=A^*\!\upharpoonright\ker(\gG_j), \ j\in \{0,1\}$, which are
self-adjoint in view of Proposition \ref{prop2.1} below.
Conversely, for any extension $A_0=A_0^*\in \Ext_A$ there exists
a (non-unique) boundary triplet $\Pi=\{\cH,\gG_0,\gG_1\}$ for $A^*$
such that $A_0:=A^*\!\upharpoonright\ker(\gG_0)$.

Using the concept of { boundary triplets} one can parameterize all
proper, in particular, self-adjoint extensions of $A$. For this
purpose denote by $\widetilde\cC(\cH)$ the set  of closed linear
relations in $\cH$, that is, the set of (closed) linear subspaces
of $\cH\oplus\cH$. The adjoint relation
$\Theta^*\in\widetilde\cC(\cH)$ of a linear relation $\Theta$ in
$\cH$ is defined by
\begin{equation*}
\Theta^*= \left\{
\begin{pmatrix} k\\k^\prime
\end{pmatrix}: (h^\prime,k)=(h,k^\prime)\,\,\text{for all}\,
\begin{pmatrix} h\\h^\prime\end{pmatrix}
\in\Theta\right\}.
\end{equation*}
A linear relation $\Theta$ is called {\it symmetric} if
$\Theta\subset\Theta^*$ and self-adjoint if $\Theta=\Theta^*$.

The multivalued part $\mul(\Theta)$ of { $\Theta\in
\wt\cC(\cH)$ is} $\mul(\Theta) = \{h \in \cH: \{0,h\}\in
\Theta\}.$ Setting $\cH_{\infty} :=\mul(\Theta)$ and $\cH_{\rm op}
:=\cH^{\perp}_{\infty}$ we get $\cH=\cH_{\rm
op}\oplus\cH_{\infty}$. This decomposition yields an orthogonal
{ decomposition}
$\Theta=\Theta_{\rm op}\oplus\Theta_{\infty}$  where
$\Theta_{\infty}:= \{0\}\oplus\mul(\Theta)$ and $\Theta_{\rm op}:=
\{\{f,g\}\in\Theta:\ f\in \dom(\Theta), g\perp\mul(\Theta)\}$. For
the definition of the inverse and the resolvent set of a linear
relation $\Theta$ we refer to \cite{DS87}.
\begin{proposition}\label{prop2.1}
Let  $\Pi=\{\cH,\gG_0,\gG_1\}$  be a
boundary triplet for  $A^*.$  Then the mapping
\begin{equation}\label{bij}
(\Ext_A\ni)\ \widetilde A \to  \Gamma \dom(\widetilde A)
=\{\{\Gamma_0 f,\Gamma_1f \} : \  f\in \dom(\widetilde A) \} =:
\Theta \in \widetilde\cC(\cH)
\end{equation}
establishes  a bijective correspondence between the sets $\Ext_A$
and  $\widetilde\cC(\cH)$. We put $A_\Theta :=\widetilde A$ where
$\Theta$ is defined by \eqref{bij}. { Moreover, the} following holds:

\item[\;\;\rm (i)] $A_\Theta= A_\Theta^*$ if and only if $\Theta=
\Theta^*;$

\item[\;\;\rm (ii)] The extensions $A_\Theta$ and $A_0$ are
disjoint if and only if $\Theta\in \cC(\cH).$ In this case
\eqref{bij} becomes
\bed
A_\Theta = A^*\!\upharpoonright\ker(\gG_1-
\Theta\gG_0);
\eed

\item[\;\;\rm (iii)] The extensions $A_\Theta$ and $A_0$ are
transversal if and only if $\Theta = \Theta^*\in [\cH].$
\end{proposition}

In particular,
$A_j:=A^*\!\upharpoonright\ker(\gG_j) = A_{\Theta_j},\ j\in
\{0,1\}$ where $\gT_0:= \{0\} \times \cH$ and $\gT_1 := \cH \times
\{0\}$. Hence $A_j= A_j^*$ since $\gT_j
= \gT^*_j.$ In the sequel the extension $A_0$
is usually regarded as a reference
self-adjoint extension.

\subsection{Weyl functions and $\gamma$-fields}\label{weylsec}

It is well known that Weyl functions give an important tool in the
direct and inverse spectral theory of singular Sturm-Liouville
operators. In \cite{DM87,DM91,DM95} the concept of Weyl function
was generalized to the case of an arbitrary symmetric operator $A$
with  $n_+(A) = n_-(A).$ Following \cite{DM87,DM91,DM95} we recall
basic facts on Weyl functions and $\gamma$-fields associated with
a boundary triplet $\Pi$.
\bd[{\cite{DM87,DM91}}]\label{Weylfunc}
{\em
 Let $\Pi=\{\cH,\gG_0,\gG_1\}$ be a boundary triplet  for $A^*.$
The functions $\gamma(\cdot):
\varrho(A_0)\rightarrow [\cH,\gH]$ and  $M(\cdot):
\varrho(A_0)\rightarrow  [\cH]$ defined by
\begin{equation}\label{2.3A}
\gamma(z):=\bigl(\Gamma_0\!\upharpoonright\mathfrak N_z\bigr)^{-1}
\qquad\text{and}\qquad M(z):=\Gamma_1\gamma(z), \qquad
z\in\varrho(A_0),
      \end{equation}
are called the {\em $\gamma$-field} and the {\em Weyl function},
respectively, corresponding to  $\Pi.$
}
\ed

It follows from the identity  $\dom(A^*)=\ker(\Gamma_0)\dot +
{\mathfrak N}_z$, $z\in\varrho(A_0)$, where  $A_0 =
A^*\!\upharpoonright\ker(\gG_0)$, and ${\mathfrak N}_z :=\ker(A^* -
z)$, that the $\gamma$-field $\gamma(\cdot)$ is well defined and
takes values in $[\cH,\gH]$. Since $\gG_1 \in [\gH_+, \cH]$, it
follows from \eqref{2.3A} that  $M(\cdot)$ is well defined too and
takes values in   $[\cH]$. Moreover, both $\gamma(\cdot)$ and
$M(\cdot)$ are holomorphic on $\varrho(A_0)$ and satisfy the
following relations (see \cite{DM91})
\begin{equation}\label{2.15}
\gamma(z)=\bigl(I+(z-\zeta)(A_0 - z)^{-1}\bigr)\gamma(\zeta),
\qquad z, \zeta\in\varrho(A_0),
\end{equation}
and
\begin{equation}\label{mlambda}
M(z) - M(\zeta)^*= (z - \bar\zeta)\gamma(\zeta)^*\gamma(z), \qquad z,
\zeta\in\varrho(A_0).
\end{equation}
The last identity yields that $M(\cdot)$ is a $R_{\cH}$-function,
that is, $M(\cdot)$ is a $[\cH]$-valued holomorphic function on
$\C\backslash \R$  satisfying
\bed
M(z)=M(\overline z)^*\qquad\text{and}\qquad
\frac{\IM(M(z))}{\IM(z)}\geq 0, \qquad z\in\C\backslash\R.
 \eed
Moreover, it follows from \eqref{mlambda} that $M(\cdot)$ satisfies
$0\in \varrho(\IM(M(z))),\  z\in\C\backslash\R$.

If $A$ is a simple symmetric operator, then  the Weyl function
$M(\cdot)$ determines the pair $\{A,A_0\}$ uniquely up to unitary
equivalence (see \cite{DM95,KL71}). Therefore $M(\cdot)$ contains
(implicitly)  full information on spectral properties of $A_0$.
We recall that a symmetric operator is
said to be {\it simple} if there is no non-trivial subspace which
reduces it to a self-adjoint operator.

For a fixed $A_0 = A_0^*$ a boundary triplet $\Pi =\{\cH,
\gG_0,\gG_1\}$ satisfying  $\dom(A_0)=\ker (\Gamma_0)$ is not
unique. Let $\Pi_j =\{\cH_j, \gG^j_0,\gG^j_1\}, \ j\in \{1,2\},$
be two such triplets. Then the corresponding Weyl functions
$M_1(\cdot)$ and $M_2(\cdot)$ are related by
\begin{equation}\label{2.11}
M_2(z) = R^* M_1(z) R + R_0,
\end{equation}
where $R_0= R^*_0 \in [{\cH}_2]$ and $R \in
[{\cH}_2,{\cH}_1]$ is boundedly invertible.

According to Proposition \ref{prop2.1} the extensions
$A_{\gT}$ and $A_0$ are not disjoint whenever $\mul(\Theta)\not =
\{0\}.$ Considering $A_{\gT}$ and $A_0$ as extensions of an
intermediate extension  $S := A_0\upharpoonright(\dom(A_0) \cap
\dom(A_\gT))$ we can avoid this inconvenience.
\bl\la{III.4}
Let $\gP = \{\cH,\gG_0,\gG_1\}$ be a boundary triplet for $A^*$,
$M(\cdot)$ the corresponding  Weyl function,   $\gT= \gT^* \in
\wt\cC(\cH)$ and $\gT = \Theta_{op}\oplus \Theta_\infty
$ its orthogonal decomposition. Further let  $S := A_0\upharpoonright(\dom(A_0) \cap
\dom(A_\gT))$. Then the triplet $\wh\gP =
\{\wh\cH,\wh\gG_0,\wh\gG_1\}$, defined by
\bed
\wh\cH := \cH_{\rm op} = \overline{\dom(\Theta)},\quad
\wh\gG_0 := \gG_0\upharpoonright\dom( S^*),\quad \wh\gG_1 :=
\pi_{\rm op}\gG_1\upharpoonright\dom(S^*),
\eed
is a boundary triplet for $S^*$, where $\pi_{\rm op}$ is the
orthogonal projection from $\cH$ onto $\cH_{\rm op}$, $A_0 =
S^*\upharpoonright\ker(\wh \gG_0)$ and $A_\gT = S_{\gT_{\rm op}}$.
The corresponding  Weyl function
is
\be\label{2.13a}
\wh M(z) := \pi_{\rm op}M(z)\upharpoonright\cH_{\rm op}, \quad z
\in \C_\pm. \ee
\el
The proof can be found in \cite{DHMS00}. Hence without loss of
generality we can very often assume  that the ``coordinate''
$\gT:=\gG\wt A$ of an extension $\wt A = A_\gT=A_\gT^* \in
\Ext_A$ corresponds to the  graph of a self-adjoint operator.

In what follows, without loss of generality, we always
assume that the closed symmetric $A$ is simple and, due to  Lemma
\ref{III.4}, the ``coordinate'' $\gT$ of {\green the  extension
$A_\gT = A_\gT^*\in {\Ext}_A$ is the graph of a self-adjoint
operator.

\subsection{Krein type formula for resolvents and comparability}

With  any  boundary triplet   $\Pi=\{\cH,\gG_0,\gG_1\}$  for $A^*$
and any proper (not necessarily  self-adjoint)  extension $A_{\Theta}\in
\Ext_A$ it is naturally  associated  the following (unique) Krein
type formula (cf. \cite{DM87,DM91,DM95})
\be\label{2.30}
(A_\Theta - z)^{-1} - (A_0 - z)^{-1} = \gamma(z) (\Theta -
M(z))^{-1} \gamma({\overline z})^*, \quad z\in \varrho(A_0)\cap
\varrho(A_\Theta).
\ee
Formula \eqref{2.30} is a generalization of the known Krein
formula for resolvents.
We note also, that all objects in
\eqref{2.30} are expressed in terms of the boundary triplet $\Pi$
(cf. \cite{DM87,DM91,DM95}). In other words, \eqref{2.30} gives a relation
between Krein-type  formula for canonical resolvents  and the theory of abstract
boundary value problems (framework of boundary triplets).

The following result is deduced from  formula \eqref{2.30} (cf.
\cite[Theorem 2]{DM91}).
\bp\label{prop2.9}
Let $\Pi=\{\cH,\gG_0,\gG_1\}$  be a boundary triplet for $A^*$,
$\Theta_i = \Theta_i^* \in \wt\cC(\cH), \ i\in \{1,2\}$.  Then for
any Schatten-von~Neumann ideal ${\mathfrak S}_p$, $p \in (0,\infty]$, and any $z \in
\C\setminus\R$ the following equivalence holds
\bed
(A_{\Theta_1}-z)^{-1} - (A_{\Theta_2}-z)^{-1}\in{\mathfrak
S}_p(\gH)
\Longleftrightarrow  \bigl(\Theta_1 - z\bigr)^{-1}-
\bigl(\Theta_2 - z \bigr)^{-1}\in{\mathfrak S}_p(\cH)
\eed
In particular,  $(A_{{\Theta}_1} - z)^{-1} - (A_0 - z)^{-1} \in
{\mathfrak S}_p(\gH) \Longleftrightarrow \bigl(\Theta_1 - i\bigr)^{-1} \in {\mathfrak
S}_p(\cH)$.

If in addition $\Theta_1, \Theta_2\in[\cH]$, then for any $p \in
(0, \infty]$ the equivalence holds
\bed
(A_{\Theta_1}-z)^{-1} -
(A_{\Theta_2}-z)^{-1}\in{\mathfrak S}_p(\gH) \Longleftrightarrow
\Theta_1 - \Theta_2 \in{\mathfrak S}_p(\cH).
\eed
\end{proposition}

\subsection{Generalized boundary triplets and proper extensions}

In applications the concept of boundary triplets is too
restrictive. Here we recall some facts on generalized boundary
triplets following \cite{DM95}.
\begin{definition}[{\cite[Definition 6.1]{DM95}}]\label{def3.1}
{\em
{ A triplet $\Pi=\{\cH,\gG_0,\gG_1\}$  is called a
\emph{generalized boundary triplet for $A^*$} if $\cH$ is an
auxiliary Hilbert space and { $\gG_j: \dom(\Gamma_j) \to\cH$}, $j
= 0,1$ are linear
mappings such that { $\dom(\Gamma) := \dom(\gG_0) \cap \dom(\gG_1)$
is a core for $A^*$},
$\gG_0$ is surjective, $A_0:= A^*\upharpoonright\ker(\gG_0)$ is
self-adjoint and the following Green's formula  holds}
\begin{equation}\label{3.1}
(A_*f, g) - (f, A_*g) = (\gG_1f, \gG_0 g)_{\cH} - (\gG_0 f, \gG_1
g)_{\cH}, \qquad f,g\in\dom(A_*),
\end{equation}
where  $A_*:= A^*\upharpoonright\dom(\Gamma).$
}
\end{definition}

{ By definition, { $A_* := A^*\upharpoonright\dom(\Gamma)$} and $A_*
\subseteq A^* = \overline{A_*}$ and $(A_*)^*=A.$}  Clearly,  every
ordinary boundary triplet is a generalized boundary triplet.
\bl[{\cite[Proposition 6.1]{DM95}}]\label{lem3.2}
Let $A$ be a densely defined closed symmetric operator and let
$\Pi=\{\cH,\gG_0,\gG_1\}$ be  a generalized boundary triplet for
$A^*.$ Then the following assertions are true:

\item[\;\;{\rm (i)}] ${\mathfrak N}^*_{z}:=\dom(A_*)\cap{\N}_{z}$ is
dense in $\mathfrak N_{z}$ and $\dom(A_*) =\dom(A_0) + \mathfrak N^*_{z}$;

\item[\;\;{\rm (ii)}] $\overline{\gG_1\dom(A_0)}=\cH$;

\item[\;\;{\rm (iii)}] $\ker(\gG) = \dom(A)$ and
${\overline{\ran(\gG)}} = \cH \oplus \cH$.
\end{lemma}
\begin{lemma}\label{lem3.3}
Let $A$ be a densely defined closed symmetric operator and let
$\Pi = \{\cH, \Gamma_0, \Gamma_1\}$ be a generalized boundary
triplet for $A^*$. Then the mapping $\gG = \{\gG_0,\gG_1\}^\top$
is closable and $\overline{\gG} \in \cC(\gH_+, \cH)$.
\end{lemma}
\begin{proof}
The Green's formula can be rewritten as $(A_*f,g) - (f,A_*g) =
(J\gG f,\gG g)$ where  $\gG:=(\Gamma_0,\Gamma_1)^\top$ and  $J :=
\begin{pmatrix}
0 & I\\
-I & 0
\end{pmatrix}$.
Let $f_n\in\dom(\Gamma_0)\cap\dom(\Gamma_1)=\dom(A_*),$\
$\|f_n\|_{\gH_+}\to 0$ and $\Gamma f_n=\{\Gamma_0 f_n,\Gamma_1
f_n\}\to\{\varphi, \psi\}$ as $n\to\infty$. Hence
\bed
0 = \lim_{n\to\infty}[(A_*f_n, g) - (f_n, A_* g)] =
\left(Jf_\infty,\gG g\right), \quad \mbox{where} \quad f_\infty :=
\{\varphi,  \psi \}^\top.
\eed
Since $\ran(\gG)$ is dense in $\cH \oplus \cH$ one has $Jf_\infty = 0$.
Thus, $\varphi=\psi=0$ and $\Gamma$ is closable.
\end{proof}

For any generalized boundary triplet $\Pi=\{\cH,\gG_0,\gG_1\}$ we
set $A_j:=A^*\upharpoonright\ker(\gG_j)$, $j \in \{0,1\}$. The
extensions $A_0$ and $A_1$ are disjoint but not necessarily
transversal. The latter holds if and only if $\Pi$ is an ordinary
boundary triplet. In general, the extension $A_1$ is only essentially
self-adjoint.

Starting with Definition \ref{def3.1},  one easily  extends the
definitions of $\gga$-field and Weyl function
to the case of a generalized boundary triplet $\Pi$ by
analogy with Definition \ref{Weylfunc} (cf. {\cite[Definition
6.2]{DM95}}).
\bd \label{def3.2}
Let $\Pi=\{\cH,\gG_0,\gG_1\}$ be  a generalized boundary triplet
for $A^*.$   Then the operator valued functions $\gamma(\cdot)$
and $M(\cdot)$  defined by
 \begin{equation}\label{3.3A}
\gamma(z):=\bigl(\Gamma_0\!\upharpoonright\mathfrak N^*_z\bigr)^{-1}:\
\cH\to \mathfrak N_{z} \;\;\text{and}\;\;M(z):=\Gamma_1\gamma(z),
\;\;z\in\varrho(A_0),
\end{equation}
are called the (generalized) {\em $\gamma$-field} and the {\em
Weyl function} associated with the generalized
boundary triplet $\Pi$, { respectively}.
\ed

It follows from  Lemma \ref{lem3.2}(i) that $\gamma(\cdot)$ takes
values in $[\cH, \gH]$, $\ran(\gga(z))=\mathfrak N^*_{z} := \dom(A_*)
\cap \mathfrak N_z$ and it satisfies the identity similar to that of
\eqref{2.15} which shows that $\gga(z)$
is a holomorphic operator valued function on
$\varrho(A_0)$.

Further, one has $\dom (M(z))= \cH$ since $\ran\gamma(z)\subset
\dom(\gG_1),\  z\in\varrho(A_0).$ By \eqref{3.3A} $M(z)$
is closable  since $\gamma(z)$ is bounded and $\gG_1$ is closable,
by Lemma \ref{lem3.3}.  Hence, by the closed graph theorem $M(\cdot)$
takes values in $[\cH]$. Moreover, it is holomorphic on
$\varrho(A_0)$, because so is $\gamma(\cdot)$, and satisfies the
relation \eqref{mlambda}.
It follows that $\ker(\im M(z)) =\{0\}, \  z\in \C_+,$ though the
stronger condition $0 \in \varrho(\im M(i)) (\Longleftrightarrow
\ran(\gga(i))=\mathfrak N_{i})$ is satisfied if and only if $\Pi$
is an ordinary boundary triplet (in the sense of Definition
\ref{t1.9}).

In the sequel we need the following simple but useful statement.
     \begin{proposition}\label{prop3.1}
Let $\Pi=\{\cH,\gG_0,\gG_1\}$ be an ordinary boundary triplet for
$A^*$, $M(\cdot)$ the corresponding Weyl function,   $B=B^* \in
\cC(\cH)$ and $A_B = A^*\upharpoonright \ker(\gG_1 - B\gG_0)$. Let
$\gG^B_1 :=\gG_0$ and $\gG^B_0 := B\gG_0-\gG_1$. Then

\item[\;\;\rm (i)] $\Pi_B=\{\cH,\gG^B_0,\gG^B_1\}$  is a
generalized boundary triplet for $A^*$ such that it holds
$\dom(A_*) := \dom(\gG) := \dom(A_0) + \dom(A_B) \subseteq
\dom(A^*)$, $A^*_* = A$;

\item[\;\;\rm (ii)] the corresponding  (generalized) Weyl function
$M_B(\cdot)$ is
\bed
M_B(z)=\left(B-M(z)\right)^{-1}, \qquad z\in\C_\pm;
\eed

\item[\;\;\rm (iii)]   $\Pi_B$ is an (ordinary) boundary triplet if
and only if $B = B^* \in[\cH]$. In this case $M_B(\cdot)$ is an ordinary
Weyl function in the sense  of Definition \ref{t1.9}.
\end{proposition}

Note, an analogon  of Proposition \ref{prop2.1} does not hold for
generalized boundary triplets.
Nevertheless, since the corresponding Weyl function determines the
pair $\{A,A_0\}$ uniquely, up to unitary equivalence,  it is
possible to describe the spectral properties of $A_0$
in terms of the (generalized) Weyl function $M(\cdot).$

\section{Weyl function and spectral multiplicity}

Throughout of this section $A$ is a densely defined
simple closed symmetric operator in $\gH$ with $n_+(A) = n_-(A).$
Let $\gP = \{\cH,\gG_0,\gG_1\}$ be a generalized} boundary triplet for $A^*$,
and  { let $M(\cdot)$ be} the corresponding generalized Weyl function.
Since $M(\cdot)\in (R_{\cH})$ it admits representation
\eqref{1.2}. Since
 $A$ is densely defined (see \cite{DM95, Ma92b}), one gets $C_1 = 0$, i.e.
\bed
M(z) = C_0  + \int^\infty_{-\infty}
\left(\frac{1}{t-z} -
  \frac{t}{1+t^2}\right)d\gS_M.
\eed
\bp\la{IV.6}
Let $A$ be a densely defined, simple closed symmetric operator and
let $\gP = \{\cH,\gG_0,\gG_1\}$ be a generalized boundary triplet
for $A_* (\subseteq A^*)$, $A^*_* = A$,   and { let}
$M(\cdot)$ { be} the corresponding  Weyl function. If $E_{A_0}$ is the spectral
measure of $A_0 := A^*\upharpoonright\ker(\gG_0)$, then $\gS_M
\thickapprox E_{A_0}$ and $\gS^{ac}_M \thickapprox E^{ac}_{A_0}$.
\end{proposition}
\begin{proof}
Alongside $\gS_M(\cdot)$ we introduce the  bounded  operator measure
$\gS^0_M(\cdot),$
\bed
\gS^0_M(\gd) = \int_\gd \frac{1}{1+t^2} \;d\gS_M, \qquad
\gd \in \cB_b(\R).
\eed
Clearly, $\gS^0_M(\cdot) \thickapprox  \gS_M(\cdot)$.
  According to \cite[formula (2.16)]{ABMN05} one has
\be\label{2.15A} \gS^0_M(\gd) = \gga(i)^*E_{A_0}(\gd)\gga(i),
\qquad \gd \in \cB(\R), \ee
where $\gga(\cdot)$ is the generalized $\gga$-field of
$\gP$. Note, that though formula \eqref{2.15A}
is proved in \cite{ABMN05} for ordinary boundary triplets, the proof remains valid
for generalized boundary triplets. Due to the simplicity of $A$
one has
\bed
\span\left\{(A_0 - z)^{-1}\ran(\gga(i)): \quad z \in \C_+ \cup
\C_-\right\} = \gH.
\eed
Hence the subspace $\mathfrak{N}_i :=
\overline{{\mathfrak N}^*_i},$ where ${\mathfrak N}^*_i :=
\ran(\gga(i))$ is cyclic for $A_0$. Next, let $P_i$ be the
orthogonal projection from $\gH$ onto $\mathfrak{N}_i$.
We set $\widetilde{\gS}^0_M(\cdot) :=
P_iE_{A_0}(\cdot)\upharpoonright{\mathfrak N}_i$.

Clearly, $\widetilde{\gS}^0_M(\cdot)$ is an operator measure.
Since the linear manifold ${\mathfrak N}^*_i$ is cyclic
for $A_0$, one gets from \cite[Theorem 4.15]{MM03} that the
measures $\widetilde{\gS}^0_M$ and $E_{A_0}$ are spectrally
equivalent.

Note that $\gS^0_M(\cdot) =
\gga(i)^*\widetilde{\gS}^0_M(\cdot)\gga(i)$. Since
$\ran(\gga(i))$ is dense in ${\mathfrak N}_i$, the latter
 yields $\gS^0_M \thicksim \wt\gS^0_M$. Let $D\in
{\mathfrak S}_2(\cH)$ and  $\ker(D) = \ker(D^*) =
\{0\}$. We set
\bed
\Psi_{D^*\gS^0_MD}(t) := \frac{dD^*\gS^0_M(t)D}{d\rho(t)}
\quad \mbox{and} \quad
\Psi_{\wt D^*\wt \gS^0_M\wt D}(t) :=
\frac{d\widetilde{D}^*\widetilde{\gS}^0_M(t)\widetilde{D}}{d\rho(t)}
\eed
where $\rho$ is a scalar measure such
that $\widetilde{\gS}^0_M  \thicksim \rho$ and
$\widetilde{D} := \gga(i)D: \cH \longrightarrow
{\mathfrak N}_i$. We note that $\ker(\widetilde{D}) =
\ker(\widetilde{D}^*) = \{0\}$. By  \cite[Corollary
4.7]{MM03} we have
\bed
N_{\gS^0_M}(t) = \rank(\Psi_{D^*\gS^0_MD}(t))
\quad \mbox{and} \quad
N_{\wt \gS^0_M}(t) = \rank(\Psi_{\wt D^*\wt \gS^0_M\wt D}(t))
\eed
for a.e. $t \in \R$ ($\mod(\rho)$). Since $\Psi_{D^*\gS^0_MD}(t) =
\Psi_{\wt D^*\wt \gS^0_M\wt D}(t)$ for a.e. $t \in \R$
($\mod(\rho)$) we get $N_{\gS^0_M}(t) =
N_{\widetilde{\gS}^0_M}(t)$ for a.e. $t \in \R$ ($\mod(\rho)$).
Hence $\widetilde{\gS}^0_M$ and $\gS^0_M$ are spectrally
equivalent. Since $\wt\gS^0_M$ and $E_{A_0}$ are spectrally
equivalent the measures ${\gS}^0_M$ and $E_{A_0}$ are spectrally
equivalent. This  proves  the first statement.

The  second statement follows from the equality
$\gS^{0,ac}_M(\gd) = \gga(i)^*E^{ac}_{A_0}(\gd)\gga(i),$\  $\gd
\in \mathcal \cB(\R)$ where $\gS^{0,ac}_M$ is the absolutely
continuous part of $\gS^0_M$.
\end{proof}

The proof of Proposition \ref{IV.6} leads to the following
computing procedure for $N_{\gS^{ac}_M}(t)$: choosing
$D \in {\mathfrak S}_2(\cH) $ { such that}
$\ker(D) = \ker(D^*) = \{0\}$ we introduce
the sandwiched Weyl function $M^D(\cdot)$,
\bed
(M^D)(z) := D^*M(z)D, \quad z \in \C_+.
\eed
It turns out that the limit $(M^D)(t) :=
\slim_{y\to+0}M^D(t+iy)$ exists for a.e. $t \in \R$. We define
in accordance with \eqref{2.13a} the function
$d_{M^D}(\cdot): \R \to \N\cup\{\infty\}$,
\bed
d_{M^D}(t) := \rank(\im(M^D(t))) =
\dim(\ran(\im(M^D(t))))
\eed
which is well-defined for a.e. $t \in \R$.

For a measurable non-negative function $\xi: \R \longrightarrow
\R_+$ defined for a.e. $t \in \R$ we introduce its support
$\supp(\xi) := \{t \in \R: \  \xi(t)>0\}$. By $\cl_{ac}(\cdot)$
  we denote the absolutely continuous closure of a Borel set of
  $\R$., cf. Appendix.
\bp\la{III.8}
Let $A$ be as in Proposition \ref{IV.6}, { let} $\gP =
\{\cH,\gG_0,\gG_1\}$  { be} a generalized boundary triplet for
$A_* (\subseteq A^*)$, $A^*_* = A$,  and { let} $M(\cdot)$
{ be} the corresponding Weyl function. Further, let
$E_{A_0}(\cdot)$ be the spectral measure of $A_0 =
A_*\upharpoonright\ker(\gG_0)= A_0^*.$ If $D \in \gotS_2(\cH)$
{ and} satisfies $\ker( D) = \ker(D^*) = \{0\}$, then
$N_{E^{ac}_{A_0}}(t) = d_{M^D}(t)$ for a.e. $t \in \R$ and
$\gs_{ac}(A_0) = \cl_{ac}(\supp(d_{M^D}))$.

If, in addition, the limit $M(t) := \slim_{y\to+0}M(t+iy)$ exists
for a.e. $t \in \R$, then $N_{E^{ac}_{A_0}}(t) = d_M(t)$ for a.e.
$t \in \R$ and $\gs_{ac}(A_0) = \cl_{ac}(\supp(d_M))$.
\end{proposition}
\begin{proof}
The relation $N_{E^{ac}_{A_0}}(t) = d_{M^D}(t)$ follows
from Theorem \ref{II.4} and Theorem \ref{IV.6}. Further, let
$\{g_k\}^N_{k=1}$, $1 \le N \le \infty$, be a total set in $\cH$.
We set $h_k := Dg_k$. One easily verifies that
$\{h_n\}^N_{n=1}$ is a total set. We set $M_{h_n}(z) :=
(M(z)h_n,h_n)$, $z \in \C_+$. Clearly, $M_{h_n}(z)$ is
$R$-function for {\green every} $n \in \{ 1,2,\ldots,N\}$  and
\bed
M_{h_n}(t) := \lim_{y\to+0}M_{h_n}(t+iy) = (M(t)h_n,h_n)
\eed
exists for a.e. $t \in \R$. Set
\bed
\gO_{ac}(M_{h_n}) := \{t \in \R: \; 0 < \im(M_{h_n}(t)) <
\infty\}.
\eed
Combining  \cite[Proposition 4.1]{BMN02} with Lemma \ref{III.7} we
obtain
\be\la{3.23}
\gs_{ac}(A_0) =
\overline{\bigcup^N_{k=1}\cl_{ac}(\gO_{ac}(M_{h_n}))} =
\cl_{ac}\left(\bigcup^N_{k=1}\gO_{ac}(M_{h_n})\right).
 \ee
If $t \in \supp(d_{M^D})$, then $\im((M^D)(t)) \not= 0$. Hence
$t \in \gO_{ac}(M_{h_n})$ for some  $n \in \{ 1,2,\dots,N\}.$
Therefore $\supp(d_{M^D}) \subseteq
\bigcup^N_{k=1}\gO_{ac}(M_{h_n})$ which yields
\be\la{3.24}
\cl_{ac}(\supp(d_{M^D})) \subseteq
\cl_{ac}\left(\bigcup^N_{k=1}\gO_{ac}(M_{h_n})\right).
\ee
Conversely, if $t \in \gO_{ac}(M_{h_n}) \cap \cE_{M^D}$,
where $\cE_{M^D} := \{t \in \R: \;
\exists\;(M^D)(t)\}$, for some $n$, then $0 < d_{M^
D}(t)$. Hence $\gO_{ac}(M_{h_n}) \cap \cE_{M^D} \subseteq
\supp(d_{M^D})$ which yields
$\bigcup^N_{k=1}\gO_{ac}(M_{h_n}) \cap \;\cE_{M^D}
\subseteq \supp(d_{M^D})$. Hence
\begin{equation*}
\cl_{ac}\left(\bigcup^N_{k=1}\gO_{ac}(M_{h_n}) \cap \;
\cE_M\right) =
\cl_{ac}\left(\bigcup^N_{k=1}\gO_{ac}(M_{h_n})\right) \subseteq
\cl_{ac}(\supp(d_{M^D}))
     \end{equation*}
Combining this equality with  \eqref{3.23} and \eqref{3.24} we
obtain $\gs_{ac}(A_0) = \cl_{ac}(\supp(d_{M^D}))$.
\end{proof}
\bc\la{IV.9}
Let $A$ be as in Proposition \ref{III.8},  { let
}$\Pi=\{\cH,\gG_0,\gG_1\}$ {  be} an ordinary boundary triplet
for $A^*$ and { let} $M(\cdot)$ { be} the corresponding
Weyl function. Further, let $B= B^* \in \cC(\cH)$, $A_B =
A^*\upharpoonright \ker(\gG_1 - B\gG_0)$ and $E_{A_B}(\cdot)$ the
spectral measure of $A_B$.  If $D\in \gotS_2(\cH)$ { and
satisfies } $\ker(D) = \ker( D^*) = \{0\}$, then
$N_{E^{ac}_{A_B}}(t) = d_{M_B^D}(t)$ for a.e. $t\in \R$ and
$\gs_{ac}(A_B) = \cl_{ac}(\supp(d_{M_B^D}))$.

If, in addition, the limit $M_B(t) := \slim_{y\to+0}M_B(t
+ iy)$ exists for a.e. $t \in \R$, then
$N_{E^{ac}_{A_B}}(t) = d_{M_B}(t)$ for a.e. $t \in \R$ and
$\gs_{ac}(A_B) = \cl_{ac}(\supp(d_{M_B}))$.
\ec
\begin{proof}
By Proposition \ref{prop3.1}  $\gP_B = \{\cH,\gG^B_0,\gG^B_1\}$ is
a generalized boundary triplet for $A_*  :=
A^*\upharpoonright\dom(A_*)$, $\dom(A_*)= \dom(A_0) + \dom(A_B)$,
and  $M_B(z) = (B - M(z))^{-1}$, $z \in \C_+$, the corresponding
generalized Weyl function. Clearly, $A_B =
A_*\upharpoonright\ker(\gG^B_0)$. It remains to apply
Proposition \ref{III.8}.
\end{proof}

This leads to the following theorem.
\bt\la{IV.10}
{ Let $A$ be a densely defined closed symmetric operator, let
$\gP = \{\cH,\gG_0,\gG_1\}$ be an ordinary boundary triplet for
$A^*$ and let $M(\cdot)$ be the corresponding Weyl function.
Further, let $A_B := A^*\upharpoonright\ker(\gG_1-B\gG_0)$, $B =
B^* \in \cC(\cH)$, and $E_{A_B}(\cdot)$ the spectral measure of
$A_B$. Let $D \in \gotS_2(\cH)$} { and } $\ker(D) = \ker(D^*)
= \{0\}$. Then

\item[{\rm\;\;(i)}] { $A_0E^{ac}_{A_0}(\cD)$ is
a part of $A_BE^{ac}_{A_B}(\cD)$} if and only if
   $d_{M^{ D}}(t) \le d_{M_B^{ D}}(t)$ for
a.e. $t \in \cD.$

\item[{\rm\;\;(ii)}] $A_0E^{ac}_{A_0}(\cD)$
and $A_BE^{ac}_{A_B}(\cD)$  are unitarily equivalent if and only if
$d_{M^D}(t) = d_{M_B^D}(t)$ for a.e. $t \in \cD.$
\et
\begin{proof}
Without loss of generality we assume that $A$ is simple
since the self-adjoint part of $A$ is contained as a direct
summand in any self-adjoint extension of $A$. We to show
that $\gS^{ac}_M(\gd) = 0$ for some $\gd \in \cB_b(\R)$ if and
only if $d_{M^D}(t) = 0$ for a.e $t \in \gd$. By the
Berezanskii-Gel'fand-Kostyuchenko theorem \cite{Ber68,BS87}
the derivative $\Psi_{{ D}^*\gS_M{ D}}(t) :=\frac{d}{dt}{
D}^*\gS(t){D}$ exists and the relation
\bed {D}^*\gS^{ac}_M(\gd\cap\cD){ D} = \int_{\gd\cap\cD} \Psi_{{
D}^*\gS_M{D}}(t) dt, \quad \gd \in \cB_b,
 \eed
holds. One has $\gS^{ac}_M(\gd) = 0$ if and only if $\Psi_{
D^*\gS_MD}(t) = 0$ for a.e. $t \in \gd$. Since $d_{M^{
D}}(t) = \dim(\ran(\Psi_{{D}^*\gS_M{ D}}(t)))$ for a.e. $t \in \R$
we find that $\gS^{ac}_M(\gd\cap\cD) = 0$ if and only if $d_{M^{
D}}(t) = 0$\  for a.e. $t \in \gd\cap\cD$. Similarly we prove that
$\gS^{ac}_{M_B}(\gd\cap\cD) = 0$ if and only if $d_{{ D}^*{M_B}{
D}}(t) = 0$ for a.e. $t \in \gd\cap\cD$.

(i) Since by assumption $d_{M^{D}}(t) \le d_{M_B^{ D}}(t)$ for
a.e. $t \in \cD,$  one gets by the considerations above that
$\gS^{ac}_M(\gd\cap\cD) \prec \gS^{ac}_{M_B}(\gd\cap\cD)$. By
Theorem \ref{II.4} we have $N_{\gS^{ac}_M}(t) = d_{M^{D}}(t)$ and
$N_{\gS^{ac}_{M_B}}(t) = d_{M_B^{D}}(t)$ for a.e $t \in \R$. Hence
$N_{\gS^{ac}_M}(t) \le N_{\gS^{ac}_{M_B}}(t)$ for a.e. $t \in \cD$
which proves that the restricted measures
$\gS^{ac}_M(\cdot\cap\cD)$ is spectrally subordinated to
$\gS^{ac}_{M_B}(\cdot\cap\cD)$, cf. Definition \ref{II.3}(i).
Since $\gS^{ac}_M \thickapprox E^{ac}_{A_0}$ and $\gS^{ac}_{M_B}
\thickapprox E^{ac}_{A_B},$ by Theorem \ref{IV.6}, we get that
$E^{ac}_{A_0}(\cdot\cap\cD)$ is spectrally subordinated to
$E^{ac}_{A_B}(\cdot\cap\cD)$. Applying Theorem \ref{II.4a}(i) we
complete the proof.

(ii) If $d_{M^{ D}}(t) = d_{{ D}^*{M_B}{ D}}(t)$ for a.e. $t \in
\cD,$ then $\gS^{ac}_M(\cdot\cap\cD) \thicksim
\gS^{ac}_{M_B}(\cdot\cap\cD)$. By Theorem \ref{II.4},
$N_{\gS^{ac}_M}(t) = d_{M^D}(t)$ and $N_{\gS^{ac}_{M_B}}(t)
= d_{M^D_B}(t)$ for a.e $t \in \R$ which implies that
the operator measures $\gS^{ac}_M(\cdot\cap\cD)$ and
$\gS^{ac}_{M_B}(\cdot\cap\cD)$ are spectrally equivalent, cf.
Definition \ref{II.3}(ii). By Theorem \ref{IV.6},
$E^{ac}_{A_0}(\cdot\cap\cD)$ and $E^{ac}_{A_B}(\cdot\cap\cD)$ are
spectrally equivalent.  Applying Theorem \ref{II.4a}(ii)
we prove that the absolutely continuous parts
$A_0E^{ac}_{A_0}(\cD)$ and $A_BE^{ac}_{A_B}(\cD)$ are unitarily
equivalent.
\end{proof}

Theorem \ref{IV.10} reduces the problem of unitary
equivalence of $ac$-parts  of certain self-adjoint extensions of
$A$ to investigation of the functions $d_{M^{D}}(\cdot)$ and
$d_{M_B^{D}}(\cdot).$
\bc
Let $A$ be as in Theorem \ref{IV.10}. { If the self-adjoint
extensions { $\wt A$ and $\wt A'$ of $A$ are $ac$-minimal,
then their $ac$-parts are unitarily equivalent.}}
\ec

\section{Unitary equivalence}

\subsection{Preliminaries}

In what  follows we assume that $A$ is a densely defined
simple closed symmetric operator in $\gotH$. By $A_0$ we denote a
 self-adjoint extension of $A$ which is fixed.
Alongside $A_0$ we consider  $\wt A = \wt A^*\in \Ext_A$.
Usually we assume that
\be\la{5.0}
(\wt A - i)^{-1} - (A_0 - i)^{-1} \in \gotS_\infty(\gotH).
\ee
It is known (see \cite{DM91} that there exists a boundary
triplet $\gP := \{\cH,\gG_0,\gG_1\}$ for $A^*$ such that $A_0 :=
A^*\upharpoonright\ker(\gG_0)$. Of course, the boundary triplet
$\gP$ is not uniquely determined by the assumption $A_0 :=
A^*\upharpoonright\ker(\gG_0)$. { If
$\Pi_1$ and $\Pi_2$ are two such boundary triplets of $A^*$, then
their Weyl functions $M_1(\cdot)$ and $M_2(\cdot)$ are related by \eqref{2.11} (cf.
\cite{DM91})}.

{ Fix  a boundary triplet $\gP := \{\cH,\gG_0,\gG_1\}$ for
$A^*$ such that $A_0 = A^*\ker(\gG_0)$.  By Proposition \ref{prop2.1} there is a linear relation
$\gT= \gT^* \in \wt\cC(\cH)$ such that $\wt A = A_\gT$.
In general,  $\gT$ is not the graph of an operator,
$\gT \not \in \cC(\cH)$. However, let us assume that $\gT$ is
the graph an operator $B$.
By Proposition} \ref{prop2.9} { we get} that $(B - i)^{-1} \in
\gotS_\infty(\gotH)$, that means, that $B$ is a self-adjoint
operator with discrete spectrum. Hence, $\varrho(B)\cap \R
\not = \emptyset.$  In what follows we assume without loss of
generality that $0 \in \varrho(B).$ According to the polar
decomposition { we have} $B^{-1} = DJD$ where
\be\label{4.2A} D := |B|^{-1/2} = D^* \in \gotS_\infty(\gotH)
\quad \mbox{and} \quad J := \sign(B) = J^* = J^{-1}. \ee
Clearly, $D \in \gotS_\infty(\cH)$, $\ker(D) = \{0\}$,
and   $D$ commutes with $J$. We set
\be\label{4.2}
G(z) := J - M^D(z), \quad z \in \C_+,
\ee
{ $M^D(z) := DM(z)D$,  $z \in \C_+$, as usually}. Obviously,
${ M^D(z)}$ and $-G(z)$ are $R$-functions. We have $\ker(G(z))
= \{0\}$ for every $z \in \C_+$. Indeed, if $G(z)f = 0$, then $Jf
= DM(z)Df.$ Hence, $\im(M(z)Df,Df) = \im(Jf,f) = 0$ which yields
$Df = 0$ or $f = 0$. Since $J$ is a Fredholm operator satisfying
$\ker(J) = \ker(J^*) = \{0\}$ we find by \cite[Theorem 5.26]{Ka76}
that $G(z)$ is boundedly invertible for $z \in \C_+$. We set $T(z)
:= G(z)^{-1}$, $z \in \cC_+$ { and note that $T(\cdot)$ is  a
Nevanlinna function because so is  $M^D(\cdot).$} Moreover,  $T(z)
- J = T(z)M^D(z)J\in \gotS_\infty(\gotH)$ for $z \in \C_+$.

\subsection{Trace class perturbations: Rosenblum-Kato theorem}

Here we apply the Weyl function technique in order to
obtain a simple and quite different proof of the classical
Rosenblum-Kato theorem.
In fact, we prove a generalization of the Rosenblum-Kato theorem
due to Birman and Krein \cite{BirKrei62} which includes
non-additive (trace class) perturbations. Our proof demonstrates
the main idea of the proof of more general { results} contained in the
next subsection.
\bt\la{V.1}
Let  $A_0$ and $\wt A$  be  self-adjoint operators in $\gH$
satisfying
\be\label{4.3}
(\wt A - i)^{-1} - (A_0 - i)^{-1} \in \gotS_1(\gotH).
\ee
Then the absolutely continuous parts $\wt A^{ac}$ and $A^{ac}_0$
of $\wt A$ and $A_0$, respectively,  are unitarily equivalent.
\et
\begin{proof}
To include the operators $\wt A^{ac}$ and $A^{ac}_0$ in the
framework of extension theory we set
\bed
A:= { A_0}\upharpoonright \dom(A),\quad  \dom(A)=\{f\in \dom({\wt A})\cap
\dom(A_0):\  A_0f = {\wt A}f\}.
\eed
{ Obviously, we have $A := \wt A\upharpoonright\dom(A)$}.
Clearly, $A$ is a closed symmetric operator in $\gH$ with
equal deficiency indices and  $A_0,\wt A\in \Ext_A$.

{  First we assume that $A$ is densely defined}. Let
$\Pi=\{\cH,\gG_0,\gG_1\}$ be a (ordinary) boundary triplet for
$A^*,$  such that $A_0 := A^*\upharpoonright\ker(\gG_0)$, and
{  $M(\cdot)$ } the corresponding  Weyl function. { By
definition $\wt A = \wt A^*\in \Ext_A$
and $\wt A$ and $A_0$ are disjoint, that is, $\dom(A) = \dom(A_0)
\cap \dom(\wt A)$. Hence, by Proposition \ref{prop2.1}(ii), there
exists  an operator $B = B^*\in \cC(\cH)$ such that $\wt A = A_B.$}

It follows from  \eqref{2.30} and \eqref{4.3} that $M_B(z) := (B -
M(z))^{-1} \in \gotS_1(\cH)$ for $z \in \C_+$.  In accordance with
\cite[Lemma 2.4]{BirEnt67}, see also \cite{Yaf92},  the limits
$M_B(t) := \lim_{y\to+0}M_B(t + iy)$ exist in $\gotS_2(\cH)$,  for
a.e $t \in \R$. By Theorem \ref{IV.10} it is { it suffices}
to calculate the multiplicity function $d_{M_B}(t) :=
\rank(M_B(t)) = \dim(\ran(\im(M_B(t))))$.

{ It follows from  \eqref{4.2A} and \eqref{4.2} that}
\begin{eqnarray}
\lefteqn{
T(z) = G(z)^{-1} = \bigl(J - M^D(z)\bigr)^{-1} = \bigl(J - D M(z)D \bigr)^{-1} }\\
& & = D^{-1}\bigl(D^{-1}JD^{-1}-M(z)\bigr)^{-1}D^{-1}
=|B|^{1/2}\bigl(B-M(z)\bigr)^{-1} |B|^{1/2},  \;\;
  z \in \C_+.
\nonumber
\end{eqnarray}
{ Combining this relation with  \eqref{4.2A} yields}
  \bed
M_B(z) := (B - M(z))^{-1} = DT(z)D, \qquad
  z \in \C_+.
  \eed
In turn,  this equality implies
      \be\label{4.5}
\im(M_B(z)) = DT(z)^*\im({ M^D(z)})T(z)D, \qquad z \in
\C_+.
      \ee
{ Moreover, since ${ M^D(z)} \in \gotS_1(\cH)$ and $T(z)
- J \in \gotS_1$ for $z \in \C_+$, by \cite[Lemma 2.4]{BirEnt67}
(see also \cite{Yaf92})}, {  for a.e $t \in \R$ and $y\to 0$
there exist the limits $M^D(t)$ and $T(t)$ in $\gotS_2(\cH)$-norm
of the  Nevanlinna operator functions ${ M^D((t+iy))}$ and
$T(t+iy)$, respectively.}
Therefore passing to the limit in \eqref{4.5} as $y\to 0$ we get
\be\la{5.3}
\im(M_B(t)) = {D}T(t)^*\im({ M^D(t)})T(t){D}\quad \text{for
a.e.}\quad t \in \R.
\ee
 Therefore we find
\bea\la{5.3a}
\lefteqn{
d_{M_B}(t) = \dim(\ran(\im(M_B(t))))  }\\
& &  = \dim(\ran(\sqrt{\im(M_B(t))}\,)) =
\dim(\ran(\sqrt{\im({ M^D(t)})}T(t)D)). \nonumber \eea
 Since $(J -
{ M^D(t)})T(t) = T(t)(J - { M^D(t)}) = I$ for a.e. $t \in \R,$ we find
$\ran(T(t)) = \cH$ for a.e. $t \in \R$. Combining this relation
with  $\overline {\ran(D)} = \cH$ and  \eqref{5.3a}  we obtain
\be\la{5.3b}
d_{M_B}(t) = \dim(\ran(\sqrt{\im({ M^D(t)})}\,)) =
\dim(\ran(\im({ M^D(t)}))) = d_{M^D}(t)
\ee
for a.e. $t \in \R$. { Applying} Theorem \ref{IV.10}(ii) we
complete this part of the proof.

{ If $A$ is not densely defined} { one can repeat the
above reasonings  applying only}  the boundary triplet technique
for non-densely defined symmetric operators developed in
\cite{DM95, Ma92b}. It turns out that the proof above can  easily
be carried over to this case.
\end{proof}

In the following  corollary we show that in  proving of
unitary equivalence of  $A_0$ and $\wt A\in \Ext_A$ it suffices to
{ restrict the consideration to disjoint extensions}.
\bc\la{V.1A}
Let $A$ be a densely defined closed symmetric operator in $\gH,$
{ let $\Pi=\{\cH,\gG_0,\gG_1\}$ be}  an ordinary boundary triplet for
$A^*$, and { let $M(\cdot)$ be} the corresponding Weyl function. Let also
$A_0 := A^*\upharpoonright\ker(\gG_0)$ and  $\cD \in \cB(\R)$.

\item[\rm\;\;(i)] If $A^{ac}_0E_{A_0}(\cD)$ is a part { of}
${\wt A}^{ac} E_{\wt A}(\cD)$ for any extension $\wt A = {\wt
A}^*\in \Ext_A$ disjoint with $A_0,$
then $A^{ac}_0E_{A_0}(\cD)$ is a part { of} $\wt A^{ac}E_{\wt
A}(\cD)$ for any extension $\wt A = \wt A^*\in \Ext_A.$

\item[\rm\;\;(ii)] If $A^{ac}_0E_{A_0}(\cD)$ is unitarily
equivalent to  $\wt A^{ac} E_{\wt A}(\cD)$
for any extension $\wt A = \wt A^*\in \Ext_A$ disjoint with
$A_0$, then  $A^{ac}_0E_{A_0}(\cD)$ is unitarily
equivalent to the absolutely continuous part  $\wt A^{ac}E_{\wt
A}(\cD)$ of any  extension $\wt A = \wt A^*\in \Ext_A.$
\ec
\begin{proof}
By Proposition \ref{prop2.1} an extension  $\wt A\in
\Ext_A$ which is not disjoint with $A_0$ admits a representation
$\wt A_\gT$ with  $\gT= \gT^*\in \wt\cC(\cH) \setminus \cC(\cH).$
{ However}, $\gT$  admits a decomposition $\cH = \cH_{\rm op} \oplus
\cH_\infty,$ \ $\gT = \gT_{\rm op} \oplus \gT_{\infty}$ where
$\gT_{\rm op}$ is the graph of the operator $B_{\rm op}= B_{\rm
op}^*\in \cC(\cH_{\rm op})$ (cf. Section 2). Denoting by $\pi_{\rm
op}$ the orthogonal projection from $\cH$ onto $\cH_{\rm op}$ and
$M_{\rm op}(z) := \pi_{\rm op}M(z)\upharpoonright\cH_{\rm op}$, we
get $(\gT - M(z))^{-1} = (B_{\rm op} - M_{\rm op}(z))^{-1}\pi_{\rm
op}.$ Therefore formula \eqref{2.30} takes the form
\bed (A_\gT - z)^{-1} - (A_0 - z)^{-1} = \gga(z)(B_{\rm op} -
M_{\rm op}(z))^{-1}\pi_{\rm op}\gga(\overline{z})^*, \quad z \in
\C_\pm. \eed
Choose an operator $B_\infty = B_\infty^* \in
\cC(\cH_\infty)$ such that $(B_\infty - i)^{-1} \in
\gotS_1(\cH_\infty)$ and put $B = B_{\rm op} \oplus B_\infty.$ It
follows from Proposition \ref{prop2.9} that
\bed
(A_\gT - z)^{-1} - ({ A_B} - z)^{-1} \in \gotS_1(\gotH),
\eed
since  $(B_\infty - i)^{-1} \in \gotS_1(\cH_\infty).$
By Theorem \ref{V.1} the absolutely continuous parts $A^{ac}_\gT$
and $A^{ac}_B$ of  $A_\gT$ and $A_B$, respectively, are
unitarily equivalent.

(i) Since by assumption $A^{ac}_0E_{A_0}(\cD)$ is a part
of $A^{ac}_B E_{A_B}(\cD)$ and $A^{ac}_B$ is unitarily equivalent
to $A_\gT^{ac}$ we get that $A^{ac}_0E_{A_0}(\cD)$ is a
part of $A^{ac}_\gT E_{A_\gT}(\cD)$.

(ii) Since, by assumption, $A^{ac}_0E_{A_0}(\cD)$ is unitarily
equivalent to $A^{ac}_BE_{A_B}(\cD)$ and $A^{ac}_B$ is unitarily
equivalent to $A_\gT$, we get that $A^{ac}_0E_{A_0}(\cD)$ is
unitarily equivalent to $A^{ac}_\gT E_{A_\gT}(\cD)$.
\end{proof}

\subsection{Compact non-additive perturbations}

Here  we  generalize the Rosenblum-Kato theorem for the case of
compact perturbations. To this end we assume that the maximal
normal function
\bed
m^+(t) := \sup_{0<y\le 1}\|M(t+iy)\|
\eed
is finite for a.e. $t \in \R$. This is the case if and only if the
normal limits $M(t) := \wlim_{y\to+0}M(t + iy)$ exist and are
bounded operators for a.e. $t \in \R$. Indeed, let  $D = D^*$ be a
Hilbert-Schmidt operator such that $\ker(D) = \{0\}$ and
let $M^D(z) := DM(z)D$, $z \in \C_+$. Since the limit
$M^D(t) := \olim_{y\to+0}M^D(t + iy)$ exists and is a bounded
operator for a.e. $t \in \R$, see \cite{BirEnt67, Yaf92}, we find
that
\bed
\lim_{y\to+0}(M(t+iy)Df,Dg) = ({ M^D(t)}f,g), \quad f,g \in \cH,
\quad \text{for a.e.} \quad  t \in \R.
\eed
Hence the limit $\lim_{y\to+0}(M(t+iy)h,k)$ exists for a.e. $t \in
\R$ and $h,k \in \ran(D)$ which yields the existence of $M(t) :=
\wlim_{y\to+0}M(t+iy)$ for a.e. $t \in \R$. The converse statement
is obvious.

Now we are ready  to prove the main result of this section.
\bt\la{V.5}
Let $A$ be a densely defined, closed symmetric operator in $\gH$,
{ let
$\Pi=\{\cH,\gG_0,\gG_1\}$ be} an ordinary boundary triplet for $A^*$,
and  { let $M(\cdot)$ be} the corresponding Weyl function. Let $\wt A$ be a
self-adjoint extension of $A$ and $A_0 :=
A^*\upharpoonright\ker(\gG_0)$. If the maximal normal function
$m^+(t)$ is finite for a.e. $t \in \R$ and condition \eqref{5.0}
is satisfied, then the absolutely continuous parts $\wt A^{ac}$
and $A^{ac}_0$ of $\wt A$ and $A_0$, respectively,  are unitarily equivalent.
\et
\begin{proof}
We divide the proof into several steps.

(i) { First we assume that  the extensions $\wt A$ and $A_0$
are disjoint, that is}  $\wt A = A_B$  where $B = B^*\in
\cC(\cH).$ We { define} the operator $D\in
\gotS_{\infty}(\cH)$ in accordance with \eqref{4.2A},  $D :=
|B|^{-1/2},$ and investigate the function $M^D(z) := M^D(z):=
DM(z)D$, $z \in \C_+$. Let $ M^D(t) := DM(t)D.$ Since the (weak)
limit $M(t) := \wlim_{y\to+0}M(t+iy)$ exists for a.e. $t \in \R,$
by \cite[Lemma 6.1.4]{Yaf92}, the following limit exists
\be\label{4.9}
\olim_{y\to+0} \|{ M^D(}t+iy) - { M^D(t)}\| = 0 \qquad\text {for a.e.} \quad
t \in \R.
\ee
Let $\gd_a := \{t \in \R: \|M(t)\| \le a\}$. Since $D = D^*$ is a
{ a non-negative }  compact operator, it  admits the spectral
decomposition
\bed
D = \sum_{l \in \N} \mu_lQ_l
\eed
where $\{\mu_l\}^\infty_{l=1}$,  is the { decreasing }
sequence of eigenvalues of $D$, $\{Q_l\}_{l \in \N}$  the
corresponding sequence of eigenprojections, $\dim\{Q_l\}<\infty.$

Since  $\mu_l\to 0$ as $l\to \infty,$ there exists a number
$l_a\in \N$ such that   $\mu_{l_a}< 1/\sqrt {2a}$. We put $\cH_1
:= \bigoplus^\infty_{l=l_a+1}Q_l\cH$ and  $\cH_2 :=
\bigoplus^{l_a}_{l=1}Q_l\cH.$   Clearly,  $\cH = \cH_1 \oplus
\cH_2$ and  $\dim(\cH_2) <\infty.$ Moreover, the operator $D$ admits
the following decomposition $D = D_1 \oplus D_2$ where
\bed
D_1 := \sum^\infty_{l = l_a + 1}\mu_lQ_l \quad
\text{and}\quad  D_2 := \sum^{l_a}_{l=1}\mu_lQ_l.
\eed
Since $\mu_{l_a}< 1/\sqrt {2a},$ we have  $\|D_1\|< 1/\sqrt {2a}.$
Hence
\be\la{5.5}
\|D_1M(t)D_1\| < 1/2, \qquad t \in
\gd_a.
\ee
Denote by $P_1$ and $P_2$ the orthogonal projections from $\cH$
onto $\cH_1$ and $\cH_2$, respectively. Note that $P_1J = JP_1$
and $P_2J = JP_2$.

(ii) Our next aim is to show that the operator function $G(z) := J
- { M^D(z)}$ is invertible in $\C_+$  and { that $T(z) :=
G(z)^{-1}$  has the limits $T(t) := \slim_{y\to+0}T(t+iy)$ for
a.e. $t \in \gd_a$}. For this purpose we consider the
decompositions
\bed { M^D(z)} := \bigg(D_iM(z)D_j\bigg)_{i,j=1}^2 :=
\begin{pmatrix}
M^D_{11}(z) & M^D_{12}(z)\\
M^D_{21}(z) & M^D_{22}(z)
\end{pmatrix}:
\ba{c}
\cH_1\\
\oplus \\
\cH_2
\ea
\longrightarrow
\ba{c}
\cH_1\\
\oplus \\
\cH_2
\ea,
\eed
$z \in \C_+$, and
\bed
G(z) = J - { M^D(z)} =
\begin{pmatrix}
J_1 - M^D_{11}(z)  & -M^D_{12}(z)\\
-M^D_{21}(z) & J_2 - M^D_{22}(z)
\end{pmatrix}, \qquad z \in \C_+,
\eed
where $J_1 := JP_1$ and $J_2 := JP_2$.

{ ${\rm (ii)_1}$} Let us prove that  $\ker(J_1 - {
M^D_{11}}(z)) = \{0\}$ for $z \in \C_+$. Indeed, from $0= J_1g -
{ M^D_{11}}(z)g = J_1g - D_1M(z)D_1g $ one gets that $0 =
\im({ M^D_{11}}(z)g,g) = (\im(M(z)D_1g,D_1g)$. Hence $0 = D_1g
= Dg$ which yields  $g = 0$. Since $0\in \varrho(J_1)$ and ${
M^D_{11}}(\cdot)\in{\gotS}_\infty,$  we obtain that the operator
$J_1 - { M^D_{11}}(z) = J_1({ I_1} -  J_1{
M^D_{11}}(z))$ is boundedly invertible for every $z \in \C_+$.
Since ${ M^D_{11}}(z)$ is a $R_{\cH_1}$-function, we get  that
${ \Xi}(z):= (J_1 - { M^D_{11}(z)})^{-1}$, $z \in \C_+$,
is a $R_{\cH_1}$-function too.

{ ${\rm (ii)_2}$} We show that for a.e. $t \in \gd_a$, $a>0$, the limit
${ \Xi}(t) := \olim_{y\to+0}{ \Xi}(t + iy)$ exists in the operator norm
and the following representation holds
\be\la{5.6}
{ \Xi}(t) = (J_1 - { M^D_{11}}(t))^{-1}. \ee
First we note that $J_1 - { M^D_{11}}(z) = J_1(I_1 -
J_1{ M^D_{11}}(z)$. Using \eqref{5.5} we get $\|J_1{ M^D_{11}}(t)\| <1$ for
$t \in \gd_a$. Hence the inverse operator $(I_1 -
J_1{ M^D_{11}}(t))^{-1}$ exists for $t \in \gd_a$. Using $(J_1 -
{ M^D_{11}}(t))^{-1} = (I_1 - J_1{ M^D_{11}}(t))^{-1}J_1$ we find that the
inverse operator $(J_1 - { M^D_{11}}(t))^{-1}$ exist for $t \in \gd_a$.
Since ${ M^D_{11}}(z)$ has limits ${ M^D_{11}}(t)$ for a.e. $t \in
\R$ one gets that $J_1{ M^D_{11}}(t) = \olim_{y\to+0}J_1{ M^D_{11}}(t + iy)$
for a.e. $t \in \R$. Fix  any such $t_0\in \gd_a.$ Then due to
estimate \eqref{5.5} there exists $\eta =\eta(t_0)$ such that
$\sup_{y \in (0,\eta)}\|J_1{ M^D_{11}}(t_0 + iy)\| \le 1/2$. Therefore,
the family $\{\|(I_1 - J_1{ M^D_{11}}(t_0 + iy))^{-1}\|\}_{y \in
(0,\eta)}$ is uniformly bounded for any fixed $t_0 \in \gd_a.$
Using this fact and \eqref{4.9}   we can pass to the limit as $y
\to 0$ in the identity
\bead
\lefteqn{
(I_1  - J_1{ M^D_{11}}(t_0 + iy))^{-1} -
(I_1 - J_1{ M^D_{11}}(t_0))^{-1}  }\\
& &  \hspace{-5mm}
=(I_1  - J_1{ M^D_{11}}(t_0 + iy))^{-1}
(J_1{ M^D_{11}}(t_0 + iy)) -  J_1{ M^D_{11}}(t_0))(I_1 -
J_1{ M^D_{11}}(t_0))^{-1}.
\eead
We obtain  $\olim_{y\to+0}((I_1  - J_1{ M^D_{11}}(t + iy))^{-1}
= (I_1 - J_1{ M^D_{11}}(t))^{-1}$ for a.e. $t \in \gd_a$ which
yields the existence of ${ \Xi}(t) := \olim_{y\to+0}{ \Xi}(t + iy)$ and
proves representation \eqref{5.6}.

{ ${\rm (ii)_3}$} Next we  set
\bed
\gD(z) := { M^D_{22}}(z) +
{ M^D_{21}}(z)(J_1 - { M^D_{11}}(z))^{-1}{ M^D_{12}}(z), \quad
z \in \C_+.
\eed
and  show that the function $T_2(\cdot) := (J_2 -
\gD(\cdot))^{-1}$ is $R_{\cH_2}$-function.

Clearly,  $\gD(\cdot)$ is holomorphic in $\C_+$ and it acts in a
finite dimensional Hilbert space $\cH_2$. Since $\det(J_2 -
\gD(\cdot))$ is also holomorphic in $\C_+,$ the determinant
$\det(J_2 - \gD(\cdot))$ has only a discrete set of zeros in
$\C_+$. Hence the inverse operator $T_2(\cdot) := (J_2 -
\gD(\cdot))^{-1}$ exists for  $z\in \Omega \subset\C_+$
where $\C_+ \setminus\Omega$  is at most countable discrete set,
that is, $T_2(\cdot)$ is meromorphic in $\C_+$.

As we just  mentioned the inverse operator $(J_2 - \gD(z))^{-1}$
exists for   $z \in \Omega \subset \C_+$.  Choose any  $z \in
\Omega.$ Then, by the { Frobenius} formula,
\be\la{5.7}
T(z) := (J - { M^D(z)})^{-1} =
\begin{pmatrix}
T_1(z) & { \Xi}(z){ M^D_{12}}(z)T_2(z)\\
T_2(z){ M^D_{21}}(z){ \Xi}(z) & T_2(z)
\end{pmatrix}
\ee
where
\be\la{5.7A}
 T_1(z) := { \Xi}(z) +
{ \Xi}(z){ M^D_{12}}(z)T_2(z){ M^D_{21}}(z){ \Xi}(z).
       \ee
Hence
    \bed
T_2(z) = P_2T(z)\upharpoonright\cH_2, \qquad z\in \Omega.
     \eed
Since $T(\cdot)$ is a $R_{\cH}$-function, we get that $\im
{(T_2(z))}>0$ for $z\in \Omega.$ Since in addition $T_2(\cdot)$  is
meromorphic in $\C_+$, we conclude  that it is holomorphic. Thus,
 $T_2(\cdot) = (J_2 - \gD(\cdot))^{-1}$ is  $R_{\cH_2}$-function,
 too.

{ ${\rm (ii)_4}$} In this step we show that for any  $a > 0$ the
limit $T(t) := \olim_{y\to+0}T(t + iy)$ exists in the operator
norm for a.e. $t \in \gd_a$.
Since $T_2(\cdot)$ is the matrix $R_{\cH_2}$-function, the
limit $T_2(t) = \olim_{y\to+0}T_2(t+iy)$ exists for { a.e.} $t \in
\R$. Besides, \eqref{4.9} yields
\bed \lim_{y\to+0} \|{ M^D_{12}}(t+iy) - { M^D_{12}}(t)\|
= 0 \quad \text{and} \quad  \lim_{y\to+0} \|{ M^D_{21}}(t+iy)
- { M^D_{21}}(t)\| = 0 \eed
for { a.e.} $t \in \R$. { Combining}  these relations with
\eqref{5.6} and \eqref{5.7A} {  yields} the existence of the
limit $T_1(t) := \olim_{y\to+0}T_1(t+iy)$ for a.e $t \in \gd_a$.
Finally, combining all these relations with  the block-matrix
representation \eqref{5.7} we complete the proof of { (ii)}.

{ ${\rm (iii)}$} { Using the results of (ii) we are now going
  to complete the proof of the theorem}. We set
$\gd_n := \{t \in \R:\ m^+(t) \le n\}$ { and } note that
$\bigcup_{n \in \N}\gd_n$ differs from $\R$ by a set of Lebesgue
measure zero. By step { (ii)}  the limit $T(t) :=
\olim_{y\to+0}T(t + iy)$ exists for a.e. $t \in \bigcup_{n \in
\N}\gd_n$ in the operator norm. Hence the limit $T(t) :=
\olim_{y\to+0}T(t + iy)$ exist for { a.e.} $t \in \R$.
Combining this fact with  \eqref{4.9} we can pass to the limit in
the identity $(J - { M^D(}t+iy))T(t +iy) = I$ {  as $y\to
0.$} We get
 \be\label{4.13}
(J - { M^D(t)})T(t) = T(t)(J - { M^D(t)}) = I \qquad \text{ for a.e.} \quad
t \in \R
        \ee
The rest of the proof is similar to that of Theorem
\ref{V.1}. First we assume that $\wt A$ is disjoint with $A_0$,
hence, it admits a representation   $\wt A = A_B$ with $B\in
\cC(\cH)$. Therefore, setting  $M_B(\cdot) := (B-M(\cdot))^{-1}$
and assuming  without loss of generality that $0 \in \varrho(B)$
we arrive at the representation \eqref{5.3} with $D=|B|^{-1/2}$
for a.e. $t \in \R$. Moreover, \eqref{4.13} yields $\ran(T(t)) =
\cH$ for a.e. $t \in \R$. Therefore arguing as in \eqref{5.3a} and
\eqref{5.3b} we obtain
\bead
\lefteqn{
d_{M_B}(t) = \dim(\ran(\sqrt{\im({ M^D(t)})}\,)) =
\dim(\ran(\sqrt{\im(M(t))}D\,)) }\\
& & = \dim(\ran(\sqrt{\im(M(t))}\,)) = \dim(\ran(\im(M(t)))) =
d_M(t) \eead
for a.e. $t \in \R$. Applying Theorem \ref{IV.10}(ii) we complete the
proof.

{  Finally,  we apply  Corollary \ref{V.1A} to extend the
proof for extensions $\wt A$ not disjoint with $A_0.$}
\end{proof}
\br
{\em

{  The result as well as the proof remains valid} if $A$ is
{ non-densely} defined. In this case {  it suffices to
use} the boundary triplet technique for non-densely defined
operators developed in \cite{DM95,Ma92b}, cf. proof of Theorem
\ref{V.1}. However, the assumptions on the Weyl function are
indispensable.
}
\er

{ The following result
is immediate from Theorem \ref{IV.10}(ii) and Theorem \ref{V.5}.}
\bc\la{V.6}
Let the assumptions of Theorem \ref{V.5} be satisfied and let
\be
\cF := \{t \in \R: m^+(t) < \infty\}.
\ee
If condition \eqref{5.0} holds, then the parts $\wt A^{ac} E_{\wt
A^{ac}}(\cF)$ and $A^{ac}_0 E_{A^{ac}_0}(\cF)$ of $\wt A$ and
$A_0$, respectively, are unitarily equivalent.
\ec
\br\la{V.7}
{\em
{ Let us} define  the invariant maximal normal function

\be\label{4.12A}\hspace{-1.45mm}
\gotm^+(t) := \sup_{y\in(0,1]}
\left\|{\im(M(i))}^{-1/2}\left(M(t+iy) -
\re(M(i))\right){\im(M(i))}^{-1/2}
\right\|,
\ee
for $t \in \R$. For Weyl functions one easily proves that
$m^+(t)$ is finite if and only if $\gotm^+(t)$ is finite.

\item[\;\;(i)] The quantity $\gotm^+(t)$ has the advantage that it is
invariant: Let $A$ be a densely defined closed symmetric operator,
$\gP = \{\cH,\gG_0,\gG_1\}$  a boundary triplet for $A^*,$ and
$M(\cdot)$ the corresponding  Weyl function. Further, let $\wt\gP
= \{\wt\cH,\wt\gG_0,\wt\gG_1\}$ be another boundary triplet for
$A^*$ with the Weyl function $\wt{M}(\cdot)$ and let $A_0 :=
A^*\upharpoonright\ker(\gG_0) = A^*\upharpoonright\ker(\wt\gG_0)$.
In this case $M(\cdot)$ and $\wt M(\cdot)$ are related by
\eqref{2.11} However, $\wt \gotm^+(t) = \gotm^+(t)$ for $t \in
\R$, where $\gotm^+(t)$ is obtained { by} replacing in
\eqref{4.12A} $M(\cdot)$ by $\wt M(\cdot)$.

\item[\;\;(ii)] Further,  if the Weyl function $M(\cdot)$
satisfies $M(i) = i$, then $m^+(t) = \gotm^+(t)$ for $t \in \R$.

\item[\;\;(iii)] Let $\pi$ be an orthogonal projection onto a subspace
$\wh\cH$ { of $\cH$}. If $\gotm^+(t)$ is finite, then the
invariant maximal normal function $\wh \gotm^+(t)$,  obtained from
\eqref{4.12A} replacing $M(\cdot)$ by $\wh M(\cdot) := \pi
M(\cdot)\upharpoonright \wh\cH$, is also finite and satisfies $\wh
\gotm^+(t) \le \gotm^+(t)$ for $t \in \R$.
}
\er

\section{Direct sums of symmetric operators}

\subsection{Boundary triplets for direct sums}

Let ${ S_n}$ be a closed densely defined symmetric {
operators} in $\gotH_n,$ \ $n_+({ S_n})=n_-( S_n),$ and { let
$\gP_n = \{\cH_n,\gG_{0n},\gG_{1n}\}$  be} a boundary triplet for
$S_n^*,\ n\in \N.$ Let
\be\la{6.5a}
A := \bigoplus^\infty_{n=1}{ S_n}, \quad \dom(A) :=
\bigoplus^\infty_{n=1}\dom({ S_n}).
\ee
Clearly, $A$ is a closed densely defined symmetric operator  in
the Hilbert space  $\gotH := \bigoplus^\infty_{n=1}\gotH_n$ with
$n_\pm(A)= \infty.$ Consider the direct sum
$\gP:=\oplus_{n=1}^{\infty} \gP_n=: \{\cH,\gG_{0},\gG_{1}\}$ of
(ordinary) boundary triplets defined by
\be\la{6.5}
\cH := \bigoplus^\infty_{n=1}\cH_n, \quad \gG_0 :=
\bigoplus^\infty_{n=1}\gG_{0n} \quad \mbox{and} \quad \gG_1 :=
\bigoplus^\infty_{n=1}\gG_{1n}.
\ee
{ Clearly,}
\be\la{6.5b}
A^* = \bigoplus^\infty_{n=1}{ S}^*_n, \quad \dom(A^*) =
\bigoplus^\infty_{n=1}\dom({ S}^*_n).
\ee
{ We note that the Green's identity}
\bed ( S_n^*f_n, g_n) - (f_n,  S_n^*g_n) =
(\gG_{1n}f_n,\gG_{0n}g_n)_{\cH_n} - (\gG_{0n}f_n,
\gG_{1n}g_n)_{\cH_n}, \eed
{ $f_n, g_n \in \dom( S_n^*)$, holds for every $S^*_n$, $n \in
\N$. This yields the} Green's identity \eqref{3.1} for $A_*:=
A^*\upharpoonright \dom (\gG)$, $\dom (\gG) := \dom (\gG_0) \cap
\dom (\gG_1) \subseteq \dom (A^*)$, that is, for
$f=\oplus_{n=1}^{\infty}f_n$, $g = \oplus_{n=1}^\infty g_n \in
\dom (\gG)$ we have
\begin{equation}\label{3.1AA}
(A_*f, g) - (f, A_*g) = (\gG_1f, \gG_0 g)_{\cH} -
(\gG_0 f, \gG_1 g)_{\cH}, \qquad f,g\in\dom(\gG),
\end{equation}
where $A^*$ and $\gG_j$ are defined by \eqref{6.5b} and
\eqref{6.5}, respectively. However, the Green's identity
\eqref{3.1AA} cannot be extended to $\dom(A^*)$ in general, since
$\dom(\gG)$ is smaller than $\dom(A^*)$ generically. It might even
happen that $\gG_j$ are not bounded as mappings from $\dom(A^*)$
equipped with the graph norm into $\cH.$ Counterexamples for the
direct sum $\Pi=\oplus^\infty_{n=1} \Pi_n$,  which {
does not form} a boundary triplet, { firstly} appeared in
\cite{Koch79}).

In this section we show that it is always possible to modify the
boundary triplets $\gP_n$ in such a way that a  new sequence $\wt
\gP_n= \{\cH_n,\wt\gG_{0},\wt\gG_{1}\}$ of boundary triplets for
${ S^*_n}$ satisfies the following properties:  $\wt \gP =
{ \oplus^\infty_{n=1}} \wt \gP_n$ { forms} a boundary
triplet for $A^*$ and the following relations { hold}
\be\la{5.5A}
{ \wt S_{0n}} := { S^*_n}\upharpoonright
\ker(\wt\gG_{0n}) = { S^*_n}\upharpoonright \ker(\gG_{0n}) =:
{ S_{0n}}, \quad n\in \N.
\ee
Hence $\wt A_{0} := { \oplus^\infty_{n=1} \wt S_{0n}} = {
\oplus^\infty_{n=1}S_{0n}} =: A_0.$ { We note that the}
existence of a boundary triplet $\gP'= \{\cH, \gG'_{0},
\gG'_{1}\}$ for $A^*$ satisfying $\ker(\gG'_{0})= \dom(A_0) $ is
known (see \cite{GG91, DM91}). However, we emphasize that in
applications we need a special form \eqref{6.5} of a boundary
triplet for $A^*$ because it leads to the block-diagonal form of
the corresponding Weyl function (cf. Sections 5.2, 5.3 below).

We start with a simple technical lemma.
\bl\la{VI.1}
Let $S$ be a densely defined closed symmetric operator with equal
deficiency indices,  $\gP = \{\cH,\gG_0,\gG_1\}$  a boundary
triplet for $S^*$,  and $M(\cdot)$ the corresponding Weyl
function. Then there exists  a boundary triplet $\wt{\gP} =
\{\cH,\wt{\gG}_0,\wt{\gG}_1\}$ for $S^*$ such that
$\ker(\wt{\gG}_0) = \ker(\gG_0)$ and  the corresponding Weyl
function $\wt{M}(\cdot)$ satisfies  $\wt{M}(i) = i$.
\el
\begin{proof}
Let $M(i) = Q + iR^2$ where $Q := \re(M(i)),\ R :=
\sqrt{\im(M(i))}$. We set
\be\label{3.8}
\wt{\gG}_0 := R\gG_0
\quad \mbox{and} \quad
\wt{\gG}_1 := R^{-1}(\gG_1 - Q\gG_0).
\ee
A straightforward computation shows that $\wt{\gP} : =
\{\cH,\wt{\gG}_0,\wt{\gG}_1\}$ is a boundary triplet for $A^*$.
Clearly, $\ker(\wt{\gG}_0) = \ker(\gG_0)$. The Weyl function
$\wt{M}(\cdot)$ of $\wt{\gP}$ is given by $\wt{M}(\cdot) =
R^{-1}(M(\cdot) - Q)R^{-1}$ which yields $\wt{M}(i) = i$.
\end{proof}

If $S$ is a densely defined closed symmetric operator in $\gotH$,
then by the first v. Neumann formula the direct decomposition
$\dom(S^*) = \dom(S) \stackrel{.}{+} \gotN_i \stackrel{.}{+}
\gotN_{-i}$  holds where $\gotN_{\pm i} := \ker(S^*\mp i)$.
Equipping  $\dom(S^*)$ with the inner  product
\be\la{3.8b}
(f,g)_+ := (S^*f, S^*g) + (f,g), \quad f,g \in \dom(S^*),
\ee
one obtains a Hilbert space denoted by $\gotH_+$.
The first v. Neumann formula leads to the following
orthogonal decomposition
\bed
\gotH_+ = \dom(S) \oplus \gotN_i \oplus \gotN_{-i}.
\eed
\bl\la{VI.2}
Let $S$ be as in Lemma \ref{VI.1}, let $\gP = \{\cH,\gG_0,\gG_1\}$
be a (ordinary) boundary triplet for $S^*$, and $M(\cdot)$ the
corresponding  Weyl function.  If $M(i) = i$, then the operator
$\gG: \gotH_+ \longrightarrow \cH \oplus \cH$,
$\gG:=(\Gamma_0,\Gamma_1)^\top$  is a contraction. Moreover, $\gG$
isometrically maps $\gotN := \gotN_i\oplus \gotN_{-i}$ onto $\cH$.
\el
\begin{proof}
We show that
\be\la{6.4}
\|\gG(f + f_i + f_{-i})\|^2_{\cH \oplus \cH} = \|f_i + f_{-i}\|^2_+
\ee
where $f \stackrel{.}{+} f_i \stackrel{.}{+} f_{-i} \in \dom(S)
\stackrel{.}{+} \gotN_i \stackrel{.}{+} \gotN_{-i} = \dom(S^*)$.
Indeed, { since $\dom(S) = \ker(\gG_0) \cap \ker(\gG_1)$, we
find}
\bed
\|\gG(f + f_i + f_{-i})\|^2_{\cH \oplus \cH} =
\|\gG_0(f_i + f_{-i})\|^2_\cH + \|\gG_1(f_i + f_{-i})\|^2_\cH.
\eed
Clearly,
\be\la{6.4A} \|\gG_j(f_i + f_{-i})\|^2_\cH = \|\gG_jf_i\|^2  +
2\re((\gG_jf_i,\gG_jf_{-i})) +\  \|\gG_jf_{-i}\|^2, \quad j\in
\{0,1\}. \ee
Using $\gG_1f_i = M(i)\gG_0f_i = i\gG_0f_i$ and $\gG_1f_{-i} = M(-i)\gG_0f_{-i} = -i\gG_0f_{-i}$
we obtain
\be\la{6.4B}
\|\gG_1(f_i + f_{-i})\|^2_\cH =
(\gG_0f_i,\gG_0f_i) - 2\re((\gG_0f_i,\gG_0f_{-i})) + (\gG_0f_{-i},\gG_0f_{-i})
\ee
Taking a sum of \eqref{6.4A} and \eqref{6.4B} we get
     \be\la{6.4BC}
\|\gG_0(f_i + f_{-i})\|^2_\cH + \|\gG_1(f_i + f_{-i})\|^2_\cH =
2\|\gG_0f_i\|^2_\cH + 2\|\gG_0f_{-i}\|^2_\cH.
       \ee
Combining equalities $\gG_1f_{\pm i}= \pm i \gG_0f_{\pm i}$ with
Green's  identity \eqref{2.10A}  we obtain $\|\gG_0f_i\|_\cH =
\|f_i\|$ and $\|\gG_0f_{-i}\|_\cH = \|f_{-i}\|$. { Therefore
\eqref{6.4BC} takes the form}
\be\la{6.4C}
\|\gG_0(f_i + f_{-i})\|^2_\cH + \|\gG_1(f_i + f_{-i})\|^2_\cH =
2\|f_i\|^2 + 2\|f_{-i}\|^2.
\ee
A straightforward computation shows
$\|f_i + f_{-i}\|^2_+ = 2\|f_i\|^2 + 2\|f_{-i}\|^2$
which together with \eqref{6.4C} proves \eqref{6.4}.
Since $\|f_i + f_{-i}\|^2_+ \le \|f\|^2_+ + \|f_i + f_{-i}\|^2_+ =
\|f + f_i + f_{-i}\|^2_+$, we get from \eqref{6.4} that $\gG$ is a contraction.

Obviously, $\gG$ is an isometry from $\gotN$ into $\cH \oplus
\cH$. Since $\gP$ is a boundary triplet for $S^*$,
 $\ran(\gG) = \cH \oplus \cH.$ { Hence} $\gG$
is an isometry from $\gotN$ onto $\cH \oplus \cH$.
\end{proof}

Passing to  direct sum \eqref{6.5a}, we  equip  $\dom(A^*_n)$ and
$\dom(A^*)$ with the graph's norms and obtain the Hilbert spaces
$\gotH_{+n}$ and $\gotH_+$, respectively. Clearly, the
corresponding inner products $(f,g)_{+n}$
and $(f,g)_+$ are defined by \eqref{3.8b} with $S$ replaced by
${ S_n}$ and $A$, respectively. Obviously, $\gotH_+ =
\bigoplus^\infty_{n=1} \gotH_{+n}.$
\bt\la{VI.3}
Let $\{{ S_n}\}^\infty_{n=1}$ be a sequence of densely defined
closed symmetric operators, $\dom({ S_n})\subset \gotH_n$,
$n_+({ S_n})=n_-({ S_n})$, and let ${ S_{0n}} = {
S^*_{0n}}\in \Ext_{{ S_n}}$. Further, let $A$ and $A_0$ be
given by \eqref{6.5a} and
\be\la{6.8}
A_0 := \bigoplus^\infty_{n=1}{ S_{0n}},
\ee
respectively. Then there exist boundary triplets $\gP_n :=
\{\cH_n,\gG_{0n},\gG_{1n}\}$ for ${ S^*_n}$ such that ${
S_{0n}} = { S^*_n}\upharpoonright\ker(\gG_{0n})$, $n \in \N$,
and the direct sum ${ \gP = \oplus^\infty_{n=1} \gP_n}$
defined by \eqref{6.5} forms an ordinary boundary triplet for
$A^*$ satisfying $A_0 = A^*\upharpoonright\ker(\gG_0)$. Moreover,
the corresponding  Weyl function $M(\cdot)$ and the $\gga$-field
$\gga(\cdot)$ are { given by}
\be\la{6.7}
M(z) = \bigoplus^\infty_{n=1}M_n(z)
\quad \mbox{and} \quad
\gga(z) = \bigoplus^\infty_{n=1}\gga_n(z)
\ee
where $M_n(\cdot)$ and $\gga_n(\cdot)$ are the Weyl functions and
the $\gga$-field corresponding to $\gP_n, \ n\in \N$. In addition,
the condition { $M(i) = iI$}  { holds}.
\end{theorem}
\begin{proof}
For every ${ S_{0n}} = { S^*_{0n}}\in \Ext { S_n}$
there exists a boundary triplet $\gP_n =
\{\cH_n,\gG_{0n},\gG_{1n}\}$ for $S_n^*$ such that ${ S_{0n}}
:= A^*_n\upharpoonright\ker(\gG_{0n})$ (see \cite{DM91}). By Lemma
\ref{VI.1} we can assume without loss of generality that the
corresponding Weyl function $M_{n}(\cdot)$ satisfies $M_n(i) = i$.
By Lemma \ref{VI.2} the mapping
$\gG^n:=(\Gamma_{0n},\Gamma_{1n})^\top: \gotH_{+n} \longrightarrow
\cH_n \oplus \cH_n$, is contractive { for each} $n\in \N$.
Hence $\|\gG_j\| = \sup_{n}\|\gG_{jn}\|\le 1,  j\in \{0,1\}$,
where $\gG_0$ and $\gG_1$ are defined by \eqref{6.5}. It follows
that the mappings $\gG_0$ and $\gG_1$ are well-defined on
$\dom(\gG) = \dom(A^*) = \bigoplus^\infty_{n=1}\dom({
S^*_n})$. Thus, { the Green's} identity \eqref{3.1AA} holds
for all $f,g \in \dom(A^*).$

Further, we set $\gotN_{\pm in} := \ker({ S^*_n} \mp i)$,\  $\gotN_n
:= \gotN_{in} \stackrel{.}{+} \gotN_{-in}$, $\gotN_{\pm i} :=
\ker(A^* \mp i)$ and $\gotN := \gotN_i \stackrel{.}{+}
\gotN_{-i}.$  By Lemma \ref{VI.2}  the restriction
$\gG^n\upharpoonright\gotN_n$ is an isometry from $\gotN_n,$
regarded as a subspace of $\gotH_{+n},$ onto $\cH_n \oplus \cH_n$.
Since $\gotN$ regarded as a subspace of $\gotH_+$ admits the
representation $\gotN = \bigoplus^\infty_{n=1}\gotN_n$, the
restriction $\gG\upharpoonright\gotN$, $\gG :=
\bigoplus^\infty_{n=1}\gG^n$, isometrically maps $\gotN$ onto $\cH
\oplus \cH$. Hence $\ran(\gG) = \cH \oplus \cH$. Equalities
\eqref{6.7} are immediate from  Definition \ref{Weylfunc}.
\end{proof}
\br
{\rm
Kochubei \cite{Koch79} proved that
$\Pi=\oplus^{\infty}_{n=1}\Pi_n$ forms a boundary triplet whenever
any  pair $\{{ S_n},{ S_{0n}}\},$\ ${ S_{0n}} :=
S^*_n\upharpoonright\ker(\gG_{0n})$, ${n\in\N}$, is unitarily
equivalent to { $\{S_1,S_{01}\}$}.
}
\er

Recall, that for any non-negative symmetric operator { $A$}
the set of its non-negative self-adjoint extensions {
$\Ext_A(0,\infty)$} is non-empty (see \cite{AG81, Ka76}). The set
$\Ext_A(0,\infty)$ contains the Friedrichs (the biggest) extension
$A^F$ and the Krein (the smallest) extension $A^K$. These
extensions are uniquely determined by the following extremal
property in the class $\Ext_A(0,\infty):$
\bed
(A^F + x)^{-1}\le(\wt A+x)^{-1}\le(A^K + x)^{-1},\quad x>0,\quad
\wt A\in\Ext_A(0,\infty).
\eed
\begin{corollary}\label{cor5.5}
Assume conditions of Theorem \ref{VI.3}. Let $S_n\ge 0$,
$n\in{\N},$ and let $S^F_n$ and $S^K_n$ be the { Friedrichs} and
the { Krein}  extensions of $S_n$, respectively. Then
\begin{equation}\label{5.14}
A^F=\oplus^\infty_{n=1}S_n^F
\qquad \text{and}\qquad
A^K=\oplus^\infty_{n=1} S^K_n.
\end{equation}
\end{corollary}
\begin{proof}
Let us prove the second of relations \eqref{5.14}. The first one
is proved similarly.  By Theorem \ref{VI.3} there exists { a}
boundary triplet $\Pi_n=\{\cH_n,\Gamma_{0n},\Gamma_{1n}\}$ for
$S^*_n$ such that $S^K_n = S_{0n}$ and
$\Pi=\oplus^{\infty}_{n=1}\Pi_n$ is a boundary triplet for $A^*$.

Fix any $x_2\in\R_+$ and put $C_2:=\|M(-x_2)\|$. Then any
$h=\oplus^{\infty}_{n=1} h_n \in \cH$ can be decomposed by
$h=h^{(1)}\oplus h^{(2)}$ with $h^{(1)}\in\oplus^p_{n=1} \cH_n$ {
and} $h^{(2)}\in\oplus^{\infty}_{n=p+1}\cH_n$ {such that
$\|h^{(2)}\| < C^{-1/2}_2$}. Hence $|(M(-x_2)h^{(2)},h^{(2)})| <
1$. {  Due to the monotonicity of $M(\cdot)$ we get}
\bed
\bigg(M(-x)h^{(2)},h^{(2)}\bigg)>\bigg(M(-x_2)h^{(2)},h^{(2)}\bigg)>-1,\qquad
x\in (0, x_2).
\eed
Since $S_{0n} = S^K_n$, the Weyl function $M_n(\cdot)$ satisfies
\begin{equation}\label{5.16}
\lim_{{ x\downarrow 0}}\bigg(M_n(-x)g_n, g_n\bigg)=+\infty, \qquad
g_n\in\cH_n\setminus\{0\},
\end{equation}
{ cf. \cite[Proposition 4]{DM91}}. Since $M(\cdot)=
\oplus^\infty_{n=1} M_n(\cdot)$ { is block-diagonal}, cf.
\eqref{6.7}, we get from \eqref{5.16} that for any $N>0$ there
exists $x_1>0$ such that
\begin{equation}\label{5.17}
\bigg(M(-x)h^{(1)},h^{(1)}\bigg)=\sum^p_{n=1}\bigg(M_n(-x)h_n,h_n\bigg)>
N \quad \text{for}\quad x\in (0, x_1).
\end{equation}
Combining \eqref{5.16} with \eqref{5.17} and using the diagonal
form of $M(\cdot)$, we get
\bed (M(-x)h,h) = (M(-x)h^{(1)}, h^{(1)}) + (M(-x)h^{(2)},h^{(2)})
> N-1 \eed
for $0<x\le\min(x_1,x_2)$. Thus, $\lim_{{x\downarrow 0}}(M(-x)h,h)
= + \infty$ for $h\in \cH\setminus\{0\}.$  Applying
\cite[Proposition 4]{DM91} we {prove the second relation of}
\eqref{5.14}.
\end{proof}
\begin{remark}
{\em
Another proof can be obtained by using characterization of $A^F$
and $A^K$ by means of the respective quadratic forms.
}
\end{remark}

\subsection{Summands with arbitrary equal deficiency indices}

Here we apply Theorem \ref{V.5} to direct sums of symmetric
operators \eqref{6.5a}, allowing summands ${ S_n}$ to have arbitrary
(finite or infinite) equal deficiency indices. We start with a
simple general proposition.
\begin{proposition}\la{VI.5b}
Let $\{{ S_n}\}^\infty_{n=1}$ be a sequence of densely defined closed
symmetric operators, $\dom({ S_n})\subset \gotH_n$,
$n_+({ S_n})=n_-({ S_n})$, and let ${ S_{0n}} = { S^*_{0n}}\in \Ext_{{ S_n}}$.
Further, let $A$ and $A_0$ be given by \eqref{6.5a} and
\eqref{6.8}, respectively. If $\wt A$ is a self-adjoint extension
of $A$ such that condition \eqref{5.0} is satisfied, then
    \be\label{5.13}
\gs_{ac}(A_{0}) = \overline{\bigcup \gs_{ac}({ S_{0n}})} \subseteq
\gs(\wt A) \quad \mbox{and} \quad \gs_{ac}(\wt A) \subseteq
\overline{\bigcup \gs({ S_{0n}})} = \gs(A_{0}).
      \ee
\end{proposition}
\begin{proof}
By the Weyl theorem, condition \eqref{5.0} yields  $\gs_{\rm ess}(\wt
A) = \gs_{\rm ess}(A_0)$. Hence
\bed \overline{\bigcup \gs_{ac}({ S_{0n}})} = \gs_{ac}(A_0) \subseteq
\gs_{\rm ess}(A_0) = \gs_{\rm ess}(\wt A) \subseteq \gs(\wt A) \eed
and
\bed \gs_{ac}(\wt A) \subseteq \gs_{\rm ess}(\wt A) = \gs_{\rm ess}(A_0)
\subseteq \gs(A_0) = \overline{\bigcup \gs({ S_{0n}})}. \eed
\end{proof}

Our further considerations are substantially based on Theorem
\ref{VI.3}.
\bt\la{VI.4}
Let $\{{ S_n}\}^\infty_{n=1}$ be a sequence of densely defined closed
symmetric operators, $\dom({ S_n})\subset \gotH_n$,
$n_+({ S_n})=n_-({ S_n})$, and let ${ S_{0n}} = { S^*_{0n}}\in \Ext_{{ S_n}}$.
Further, let $\gP_n = \{\cH_n,\gG_{0n},\gG_{1n}\}$ be an ordinary
boundary triplet for ${ S^*_n}$  such that ${ S_{0n}} =
{ S^*_n}\upharpoonright\ker(\gG_{0n})$, $n \in \N$, and let
$M_n(\cdot)$ be the corresponding Weyl function. Moreover, let
$\gotm^+_n(t)$, $n \in \N$, be the { invariant maximal normal function} obtained from
\eqref{4.12A} by replacing $M(\cdot)$ by $M_n(\cdot)$. If
$\sup_{n\in\N}\gotm^+_n(t) < +\infty$ for a.e. $t \in \R$, then
for any  self-adjoint extension $\wt{A}$ of $A$ defined by
\eqref{6.5a}, which satisfies { the} condition \eqref{5.0}, the
absolutely continuous parts $\wt{A}^{ac}$ and $A^{ac}_0$ are
unitarily equivalent. In particular, { instead of}
\eqref{5.13} we have $\gs_{ac}(A_{0})= \gs_{ac}(\wt A)$.
\et
\begin{proof}
Let $\wt \gP_n = \{\cH_n, \wt\gG_{0n}, \wt\gG_{1n}\}$ be a
boundary triplet for ${ S^*}_n, \ n\in \N,$ defined according to
\eqref{3.8}, that is ${\wt\Gamma}_{0n}:= R_n\Gamma_{0n}$ and
${\wt\Gamma}_{1n}:=R^{-1}_n\bigl(\Gamma_{1n}- \re
(M_n(i))\Gamma_{0n}\bigr)$, where $R_n:=\sqrt{\im M_n(i))}.$ The
corresponding Weyl function ${\wt M}_n(\cdot)$ is
\bed
{\wt M}_n(z)=R^{-1}_n\bigl(M_n(z)- \re M_n(i)\bigr)R^{-1}_n, \quad
n\in \N.
\eed
Since  ${\wt M}_n(i)=i,\  n\in \N,$ by Theorem \ref{VI.3},
${\wt\Pi}=\oplus^{\infty}_{n=1}{\wt\Pi}_n=:\{\cH,{\wt\Gamma}_0,{\wt\Gamma}_1\}$
is a boundary triplet for $A^*=\oplus_{n=1}^{\infty} { S_n}^*$
satisfying
$A^*\upharpoonright\ker{\wt\Gamma}_0=A_0:=\oplus^{\infty}_{n=1}
{ S_{0n}}$.  By definition of $\gotm^+_n(\cdot)$ and due to
Remark \ref{V.7} one has $\gotm^+_n(t) = \wt m^+_n(t) := \sup_{y
\in (0,1]}\|\wt M_n(t+iy)\|$ for $t \in \R$, $n \in \N$. Since
$A_0 =\oplus_{n=1}^{\infty} { S_{0n}}$ we get  that $\wt
m^+(t) = \sup_n m^+_n(t)$, where $\wt m^+(t) := \sup_{y \in
(0,1]}\|\wt M(t + iy)\|$, $t \in \R$. Since, by assumption, the
maximal normal function $\wt m^+(t)$ is finite, we obtain from
Theorem \ref{V.5} that $\wt A^{ac}$ and $A^{ac}_0$ are unitarily
equivalent.
\end{proof}
\bc\la{VI.5}
Let the assumptions of Theorem \ref{VI.4} be satisfied and let
\be\la{6.12}
\cN := \{t \in \R:\  \sup_{n\in\N}\gotm^+_n(t) <
\infty\}.
\ee
If condition \eqref{5.0} holds, then the parts $\wt{A}^{ac}E_{\wt
A}(\cN)$ and $A^{ac}_0E_{A_0}(\cN)$ of the operators
$\wt{A}$ and ${A_0},$ respectively,  are unitarily equivalent.
\ec
Let $T$ and $T'$ be densely defined closed symmetric operators in
$\gotH$ and let $T_0$ and $T'_0$ be self-adjoint extensions of
$T$ and $T'$, respectively. It is  said that the pairs $\{T,T_0\}$
and $\{T',T'_0\}$ are unitarily equivalent if there exists  a
unitary operator $U$ in $\gotH$ such that $T' = UTU^{-1}$ and
$T'_0 = UT_0U^{-1}$.
\bc\la{VI.5a}
{ Assume  the conditions of Theorem \ref{VI.4}. Let also} the
pairs $\{{ S_n},{ S_{0n}}\}$, $n\in \N,$ be unitarily
equivalent to the pair { $\{S_1, S_{01}\}$}.  If the maximal
normal function $m^+_1(t)$ is finite for a.e. $t \in \R$ and
condition \eqref{5.0} is satisfied, then the absolutely continuous
parts $\wt{A}^{ac}$ and $A^{ac}_0$ are unitarily equivalent.
\ec
\begin{proof}
Since the symmetric operators ${ S_n}$ are unitarily
equivalent, we assume without loss of generality that $\cH_n =
\cH$ for each $n \in \N$. Let $U_n$ be a unitary operator such
that $A_1 = U_n{ S_n}U^{-1}_n$ and $A_{01}  = U_n{
S_{0n}}U^{-1}_n$. A straightforward computation shows that $\gP'_n
:= \{\cH,\gG'_{0n},\gG'_{1n}\}$, $\gG'_{0n} := \gG_{01}U_n$ and
$\gG'_{1n} := \gG_{1n}U_n$, defines a boundary triplet for ${
S^*_n}$. The Weyl function $M'_n(\cdot)$ corresponding to $\gP'_n$
is $M'_n(z) = M_1(z)$. Hence  $\gotm^+_n(\cdot)=
\gotm'^+_n(\cdot)$ and  $\gotm^+_1(t) = \gotm'^+_n(t)$ for $t \in
\R,$ where $\gotm^+_n(t)$ and $\gotm'^+_n(t)$ are the invariant
maximal normal functions corresponding  to the triplets $\gP_n$
and $\gP'_n$, respectively. By Remark \ref{V.7}(i), $\gotm^+_1(t)
= \gotm^+_n(t)$ for $t\in \R$ and $n \in \N$. Applying Theorem
\ref{VI.4} we complete the proof.
\end{proof}

\subsection{ Finite deficiency summands: $ac$-minimal extensions}

Here we improve the previous results assuming that
$n_{\pm}({ S_n})<\infty.$ First, we show that extensions
$A_0=\oplus^\infty_{n=1} { S_{0n}}(\in \Ext_A)$ of the form
\eqref{6.8} posses a certain spectral minimality property. To this
end we start with the following lemma.
\bl\la{V.9}
Let $H$ be  a bounded non-negative self-adjoint operator in a
separable Hilbert  space $\gotH$ and let $L$ be a bounded operator
in $\gotH$. Then

\item[\rm \;\;(i)]  $\dim(\overline{\ran(H)}) =
\dim(\overline{\ran(\sqrt{H})})$.

\item[\rm \;\;(ii)]
If $L^*L  \le H$, then
$\dim(\overline{\ran(L)}) \le \dim(\overline{\ran(H)})$.

\item[\rm \;\;(iii)] If $P$ is an orthogonal projection,
then $\dim(\overline{\ran(PHP)}) \le
\dim(\overline{\ran(H)})$.
\el
\begin{proof}
The assertion (i) is obvious.

(ii)  If $L^*L \le H$, then there is a contraction $C$ such that
$L = C\sqrt{H}.$ Hence $\dim(\overline{\ran(L)}) =
\dim(\overline{\ran(C\sqrt{H})}) \le
\dim(\overline{\ran(\sqrt{H})}) = \dim(\overline{\ran(H)})$.

(iii) Clearly, $\dim(\overline{\ran(PHP)}) \le
\dim(\overline{\ran(HP)}) \le \dim(\overline{\ran(H)})$.
\end{proof}
\bt\la{VI.6}
Let $\{{ S_n}\}^\infty_{n=1}$ be a sequence of densely defined
closed symmetric operators, $\dom({ S_n})\subset \gotH_n$,
with $n_+({ S_n})=n_-({ S_n})< \infty$, $n \in \N$ and let
${ S_{0n}} = { S^*_{0n}}\in \Ext_{ S_n}$. Let also $A$
and $A_0$ be given by \eqref{6.5a} and \eqref{6.8}, respectively.
Then { $A_0$} is $ac$-minimal, in particular,
$\sigma_{ac}(A_0) \; { \subseteq} \;\sigma_{ac}(\wt A).$
\et
\begin{proof}
By Theorem \ref{VI.3} there is a sequence of boundary triplets
$\gP_n := \{\cH_n,\gG_{0n},\gG_{1n}\}$, $n \in \N$, for ${
S^*_{n}}$ such that ${S_{0n}}  = {
S^*_n}\upharpoonright\ker(\gG_{0n})$, $n \in \N$, and the direct
sum $\gP =  \{\cH,\gG_0,\gG_1\} = \bigoplus^\infty_{n=1}\gP_n$ of
the form \eqref{6.5a} is a boundary triplet for $A^*$ satisfying
$A_0 = A^*\upharpoonright\ker(\gG_0)$. By Proposition
\ref{prop2.1}, any $\wt A = {\wt A}^*\in \Ext_A$ admits a
representation  $\wt A = A_\gT$ with  $\gT = \gT^*\in \wt
\cC(\cH).$   By Corollary \ref{V.1A}(i),  we can
{ assume that } $\wt A$ and
$A_0,$ { are disjoint, that is }
{ $\gT= B =B^*\in \cC(\cH)$}. Consider the generalized Weyl
function $M_B(\cdot) := (B - M(\cdot))^{-1}$, where $M(\cdot) =
\bigoplus^\infty_{n=1}M_n(\cdot)$, cf. \eqref{6.7}. Clearly,
\bed
\imag(M_B(z)) = M_B(z)^*\imag(M(z))M_B(z), \quad z \in \C_+.
\eed
Denote by $P_N$, $N \in \N$,  the orthogonal projection from $\cH$
onto the subspace $\cH_N := \bigoplus^N_{n=1}\cH_n$. Setting
$M^{P_N}_B(z):= P_NM_B(z)\upharpoonright\cH_N$,  { and taking
into account the block-diagonal form of $M(\cdot)$ and the
inequality $\imag(M(z))>0$ we obtain}
\bea\la{5.18}
\lefteqn{
\imag(M^{P_N}_B(z)) = \imag(P_NM_B(z)P_N) }\\
& & = P_NM_B(z)^*\imag(M(z))M_B(z)P_N \ge
M^{P_N}_B(z)^*\imag(M^{P_N}(z))M^{P_N}_B(z), \nonumber \eea
where $M^{P_N}(z) := P_NM(z)\upharpoonright\cH_N =
\bigoplus^N_{n=1}M_n(z)$. Since $P^N$ is a finite dimensional
projection the limits $M^{P_N}_B(t) := \slim_{y\to+0}M^{P_N}_B(t+iy)$
and $M^{P_N}(t) := \slim_{y\to+0}M^{P_N}(t+iy)$ exist for a.e. $t \in
\R$. From \eqref{5.18} we get
\be\la{5.19}
\imag(M^{P_N}_B(t)) \ge
M^{P_N}_B(t)^*\imag(M^{P_N}(t))M^{P_N}_B(t) \quad\text{ for
a.e.}\quad t \in \R.
\ee
Since $M_B(\cdot)$ is { a} generalized Weyl function, it is a
strict $R_\cH$-function, that is, $\ker(\imag(M_B(z)))=\{0\},\
z\in \C_+$. Therefore, $M^{P_N}_B(\cdot)$ is also strict. Hence
$0\in \varrho(M^{P_N}_B(z))$, $z\in \C_+$, and $G_N(\cdot) := -
(M^{P_N}_B(\cdot))^{-1}$ { is strict}. Since both $G_N(\cdot)$ and
$M^{P_N}_B(\cdot)$ are { matrix-valued} $R$-functions, the limits
$M^{P_N}_B(t+i0):= \lim_{y\to+0} M^{P_N}_B(t+iy)$ and $G_N(t+i0):=
\lim_{y\to+0} G_N(t+iy)$ exist for a.e. $t\in \R.$ Therefore,
passing to the limit in the identity $M^{P_N}_B(t+iy)G_N(t+iy) =
-I$ as $y\to 0,$  we get $M^{P_N}_B(t+i0)G_N(t+i0) = -I$ for a.e.
$t\in \R.$ Hence $M^{P_N}_B(t) := M^{P_N}_B(t+i0)$ is invertible
for a.e. $t \in \R$.

Further, combining  \eqref{5.19} with Lemma \ref{V.9}(ii) we get
     \bed
\dim\left(\overline{\ran\left(\sqrt{\imag
M^{P_N}(t)}M^{P_N}_B(t)\right)}\right) \le d_{M^{P_N}_B}(t) \quad
\text{for a.e.} \quad t \in \R.
        \eed
Since $M^{P_N}_B(t)$ is invertible for a.e. $t \in \R,$ we find
      \be\label{5.21}
 d_{M^{P_N}}(t) := \dim\left(\overline{\ran \left( \sqrt
{\imag M^{P_N}(t)}\right)}\right) \le d_{M^{P_N}_B}(t)\quad
\text{for a.e.} \quad t \in \R.
       \ee
Let ${ D}_N = P_N \oplus { D}_0$ where { ${D}_0 \in
\gotS_2(\cH^\perp_N)$ and satisfy} $\ker({ D}_0) = \ker({D}^*_0) =
\{0\}$. Then $\ker({ D}_N) = \ker({D}^*_N) = \{0\}$ and $P_N =
P_N{D}_N = { D}_NP_N$. By Lemma \ref{V.9}(iii),
$d_{M^{P_N}}(t) \le d_{M^{{ D}_N}_B}(t)$ for a.e. $t \in \R$.
Further, for any ${D}\in \gotS_2(\cH)$ and satisfying $\ker({D}) =
\ker({D}^*) = \{0\},$ $d_{M_B^{D}}(t) = d_{M^{{D}_N}_B}(t)$ for
a.e. $t \in \R.$ { Combining this equality with \eqref{5.21}}
we get  $d_{M^{P_N}}(t) \le d_{M_B^{ D}}(t)$ for a.e. $t \in \R$
and $N \in \N$. Since
\be\la{6.13A}
d_{M^{P_N}}(t) = \sum^N_{n=1}d_{M_n(t)} \quad \mbox{and} \quad
d_{M^D}(t) = \sum^\infty_{n=1}d_{M_n}(t)
\ee
for a.e. $t \in \R,$ we finally prove that $d_{M^{ D}}(t) \le
d_{M_B^{D}}(t)$ for a.e. $t \in \R$. { One completes the
proof by applying  Theorem} \ref{IV.10}(i).
\end{proof}
\bc\la{VI.7A}
Let  { the assumptions} of Theorem \ref{VI.6} be satisfied and
{ let $S_n \ge 0$, $n \in \N$}. { Further, let $A$ and $A$
be given by \eqref{6.5a} and \eqref{6.8}, respectively}. Then the
Friedrichs and { the } Krein extensions $A^F$ and $A^K$ of
$A$ are $ac$-minimal. In particular, $(A^F)^{ac}$ and $(A^K)^{ac}$
are unitarily equivalent.
\ec
\begin{proof}
{ Combining  Theorem \ref{VI.6}  and Corollary \ref{cor5.5}
yields the assertion.}
\end{proof}
\bc\la{VI.7}
Let { the assumptions} of Theorem \ref{VI.6} be satisfied and let
\be\la{6.13}
\cD := \{t \in \R: \sum_{n\in\N}d_{M_n}(t) = \infty\}.
\ee
If,  { in addition,} condition \eqref{5.0} holds, then the parts $\wt A^{ac}E_{\wt
A}(\cD)$ and $A^{ac}_0E_{A_0}(\cD)$ of the operators  $\wt A$ and
$A_0,$ respectively, are unitarily equivalent.
\ec
\begin{proof}
By \eqref{6.13A} and \eqref{6.13}, $d_{M^D}(t) = + \infty$ for
a.e. $t \in \cD$.  Applying Theorem \ref{VI.6} and
Theorem \ref{II.4a}(ii) we complete the proof.
\end{proof}
\bc\la{VI.10}
{ Let the assumptions of Theorem \ref{VI.6} be satisfied and
let $\cN$ and  $\cD$ be given  by \eqref{6.12} and \eqref{6.13},
respectively. If  condition \eqref{5.0} is valid, then} the {
parts} $\wt A^{ac}E_{\wt A}(\cD \cup \cN)$ and
$A^{ac}_0E_{A_0}(\cD\cup\cN)$ are unitarily equivalent.
\ec
\begin{proof}
By Corollary \ref{VI.5}, the parts $\wt A^{ac}E_{\wt A}(\cN)$ and
$A^{ac}_0E_{A_0}(\cN)$ are unitarily equivalent. Corollary
\ref{VI.7} yields the unitary equivalence of the parts
 $\wt A^{ac}E_{\wt A}(\cD)$ and $A^{ac}_0E_{A_0}(\cD)$. Hence the
 parts $\wt A^{ac}E_{\wt A}(\cD \cup \cN)$ and $A^{ac}_0E_{A_0}(\cD\cup\cN)$ are unitarily
equivalent { too}.
\end{proof}
\bc\la{VI.11}
{ Assume conditions  of Theorem \ref{VI.6}}.  { Then}
$\overline{\bigcup_{n\in\N}\gs_{ac}({ S_{0n}})} \subseteq
\gs_{ac}(\wt A).$ If, in addition, condition \eqref{5.0} is valid
and { the extensions $S_{0n}$, $n \in \N$, are purely
absolutely
  continuous, then}
\be\la{6.17}
\gs_{ac}(\wt A) = \overline{\bigcup_{n\in\N}\gs_{ac}({ S_{0n}})}.
\ee
\ec
\begin{proof}
The first statement { immediately follows} from  Theorem \ref{VI.6}.
Relation \eqref{6.17} is implied by  { Proposition} \ref{VI.5b}.
\end{proof}
\bc\la{VI.12}
{ Assume the conditions  of  Theorem \ref{VI.6}}.  { Let
also the pairs} $\{{ S_n},{ S_{0n}}\}$, $n \in \N$, be
pairwise unitarily equivalent. { If condition \eqref{5.0}
holds }, then  for any $\wt A\in \Ext_A$ the $ac$-parts $\wt
A^{ac}$ and $A^{ac}_0$
are unitarily equivalent.
\ec
\br
{\rm
(i)\ For { the} special case { $n_\pm(S_n) = 1$, $n \in
\N$}, Theorem \ref{VI.6} complements \cite[Corollary 5.4]{ABMN05}
where the inclusion { $\sigma_{ac}(A_0)\subseteq
\sigma_{ac}(\wt A)$} was proved. Moreover, Corollary  \ref{VI.12}
might be { regarded} as a substantial generalization of
\cite[Theorem 5.6(i)]{ABMN05} to the case $n_{\pm}({ S_n})>1$.
However,  in the case $n_{\pm}({ S_n})=1,$ Corollary
\ref{VI.12} is implied by \cite[Theorem 5.6(i)]{ABMN05} where the
unitary equivalence of ${\wt A}^{ac} = {\wt A}_B^{ac}$ and
$A_0^{ac}$ was proved  under { the weaker} assumption that
$B$ is purely singular. Indeed, by Proposition \ref{prop2.9}
condition \eqref{5.0} { with $\wt A = A_B$}  is equivalent to
discreteness of $B.$

(ii)\ The inequality   $N_{E^{ac}_{A_0}}(t) \le
N_{E^{ac}_{\wt{A}}}(t)$ in Theorem \ref{VI.6} might be strict even
for $t\in\sigma_{ac}(A_0)$. Indeed, assume that $(\alpha,\beta)$
is a gap for all { except for the operators}
$S_1,\ldots,S_N$. Set $S_1:=\oplus^N_{n=1} { S_n}$ and
$S_2:=\oplus^{\infty}_{ n=N+1}{ S_n}$.  Then
$n_{\pm}(S_2)=\infty$ and $(\alpha,\beta)$ is a gap for $S_2.$ By
\cite{Bra04} there exists ${\wt S}_2={\wt S}^*_2\in \Ext_{S_2}$
having $ac$-spectrum within $(\alpha,\beta)$ of arbitrary
multiplicity. Moreover, even for operators $A = {
\oplus^\infty_{n=1} S_n}$  satisfying {  assumptions} of
Corollary \ref{VI.12} with $n_{\pm}({ S_n})=1$ the inclusion
{ $\gs_{ac}(A_0)\subseteq \gs_{ac}(\wt A)$} might be strict
whenever condition \eqref{5.0} is violated, { cf. \cite{Bra04}
or \cite[{ Theorem 4.4}]{ABMN05} which  guarantees the
appearance of prescribed spectrum either within one gap or within
several gaps of $A_0$}.
}
\er

\section{Sturm-Liouville operators with operator \\potentials}

\subsection{Bounded operator potentials}

Let $\cH$ be a separable Hilbert space. As usual, $L^2({\mathbb
R}_+,\cH) := L^2({\mathbb R}_+)\otimes\cH$ stands for the Hilbert
space of (weakly) measurable vector-functions $f(\cdot):{\mathbb
R}_+\to\cH$ satisfying $\int_{{\mathbb
R}_+}\|f(t)\|^2_{\cH}dt<\infty$. Denote also by $W^{2,2}({\mathbb
R}_+,\cH) := W^{2,2}({\mathbb R}_+)\otimes\cH$ the Sobolev space of
vector-functions taking  values in $\cH$.

Let $T = T^* \ge 0$ be a bounded operator in $\cH$.  Denote by
$A:=A_{\min}$ the minimal operator generated on $\gH:=L^2({\mathbb
R}_+,\cH)$ by a differential expression
$\cA=-\frac{d^2}{dx^2}\otimes I_{\cH}+I_{L^2({\mathbb R}_+)}\otimes
T$. It is known (see \cite{GG91,Rof-Bek69}) that $A$ is
given by
\begin{eqnarray}\la{8.1}
(Af)(x)   =  -f''(x) + Tf(x), \quad f \in \dom(A) =
W^{2,2}_0({\mathbb R}_+,\cH),
\end{eqnarray}
where $W^{2,2}_0(\R_+,\cH) := \{f \in W^{2,2}(\R_+,\cH):\ f(0) =
f'(0) =0 \}$.

The operator $A$ is closed, symmetric and non-negative. It can be
proved similarly to \cite[Example 5.3]{BMN02} that $A$ is simple.
The adjoint operator $A^*$ is given by \cite[Theorem 3.4.1]{GG91}
\be\la{8.1B}
(A^*f)(x) = -f''(x) + Tf(x), \quad f \in \dom(A^*) =
W^{2,2}(\R_+,\cH).
\ee
By \cite[Theorem 1.3.1]{LioMag72} the trace operators
$\Gamma_0, \ \Gamma_1: \dom(A^*)\to \cH,$
\be\la{8.2} \gG_0 f = f(0) \quad \mbox{and} \quad \gG_1 f = f'(0),
\quad f\in \dom(A^*), \ee
are well defined Moreover, the deficiency subspace
$\mathfrak N_z(A)$ is
     \begin{equation}\label{8.3A}
\mathfrak N_z(A) = \{e^{ix\sqrt{z-T}}h: \ h \in \cH\}, \qquad z\in
\C_{\pm}.
         \end{equation}
\begin{lemma}\label{lem6.1}
A triplet $\gP = \{\cH,\gG_0,\gG_1\}$, with $\gG_0$ and $\gG_1$
defined by  \eqref{8.2}, forms
 a boundary triplet for $A^*$.
The corresponding Weyl function $M(\cdot)$ is
\be\la{8.2B} M(z) = i\sqrt{z-T} = i\int\sqrt{t+iy -
\gl}\;dE_T(\gl),  \quad z = t+iy \in \C_+.
 \ee
\end{lemma}
\begin{proof}
One obtains the Green formula integrating by parts. The
surjectivity of the mapping $\gG:=(\Gamma_0,\Gamma_1)^\top:
\dom(A^*) \rightarrow \cH \oplus \cH$ is immediate from
\eqref{8.2} and \cite[Theorem 1.3.2]{LioMag72}. Formula
\eqref{8.2B} is implied by \eqref{8.3A}.
\end{proof}
\bl\la{VII.1}
Let $T$ be a bounded non-negative self-adjoint operator in $\cH$
and let $A$ and $\Pi=\{\cH,\Gamma_0,\Gamma_1\}$ be defined by
\eqref{8.1} and \eqref{8.2}, respectively. Then

\item[\rm\;\;(i)] the invariant maximal normal function $\gotm^+(t)$
of the Weyl function $M(\cdot)$ is finite for all $t \in \R$ and
satisfies
\be\la{7.3a}
\gotm^+(t) \le 2(1 + t^2)^{1/4}, \quad t \in \R.
\ee

\item[\rm\;\;(ii)] The limit $M(t+i0) := \slim_{y\to+0}M(t+iy)$
exists, is bounded and equals
\be\la{7.3}
M(t +i0) = i\int_\R \sqrt{t-\gl} dE_T(\gl) \quad \text{for any}\ \
t \in \R.
\ee
\item[\rm\;\;(iii)]  $d_M(t) = \dim(\ran(E_T([0,t))))$ for any $t
\in \R$.
\el
\begin{proof} (i)
It is immediate from \eqref{8.2B} and definition \eqref{4.12A} of
$\gotm^+(\cdot)$ that
\bed \gotm^+(t) \le \sup_{y\in(0,1]}\sup_{\gl \ge 0}
\left|\frac{\sqrt{t + iy -\gl} -
\real(\sqrt{i-\gl})}{\imag(\sqrt{i - \gl})}\right|. \eed
Clearly, $\sqrt{i-\gl} = (1 + \gl^2)^{1/4}e^{i(\pi -\varphi)/2}$
where $ \varphi := \arccos\left(\tfrac{\gl}{\sqrt{1 +
\gl^2}}\right).$ Hence
\bed \left|\frac{\re(\sqrt{i-\gl})}{\im(\sqrt{i-\gl})}\right| =
\tan(\frac{\varphi}{2}) = \frac{1}{\gl + \sqrt{1 + \gl^2}} \le 1,
\quad \gl \ge 0. \eed
Furthermore, we have
\bed \left|\frac{\sqrt{t + iy -\gl}}{\im(\sqrt{i - \gl})}\right|
\le \sqrt{2}\sqrt{\frac{\sqrt{(\gl - t)^2 + y^2}}{\gl + \sqrt{1 +
\gl^2}}} \le 2^{3/4}(1 + t^2)^{1/4} \eed
for $\gl \ge 0$, $t \in \R$ and $y \in (0,1]$ which yields
\eqref{7.3a}.

(ii) From \eqref{8.2B} we find $M(t):= M(t +i0) := \slim_{y\to+0}
i\sqrt{t+iy - T} = i\sqrt{t-T},$ for  any $t \in \R,$ which proves
\eqref{7.3}. Clearly, $M(t)\in [\cH]$ since $T\in [\cH].$

(iii) It follows that $\im(M(t)) = \sqrt{t-T}E_T([0,t))$, which
yields $d_M(t) = \dim(\ran(\im(M(t))))= \dim(\ran(E_T([0,t))))$.
\end{proof}

With the operator $A=A_{\min}$ it is naturally  associated a
(closable) quadratic form ${\mathfrak t}'_F[f]:= (Af,f),\
\dom({\mathfrak t}')=\dom(A).$ Its closure ${\mathfrak t}_F$ is
given by
\be\la{8.1A} { \mathfrak t_F[f]} := \int_{\R_+}
\left\{\|f'(x)\|^2_\cH + \|\sqrt{T}f(x)\|^2_\cH\right\}dx, \ee
$f \in \dom({\mathfrak t}_F) = W^{1,2}_0(\R_+,\cH)$,
where  $W^{1,2}_0(\R_+,\cH) := \{W^{1,2}(\R_+,\cH): f(0)=0\}.$
By definition,  the Friedrichs extension $A^F$ of $A$ is a
self-adjoint operator associated with ${\mathfrak t}_F$. Clearly,
$A^F = A^*\upharpoonright (\dom(A^*)\cap\dom({\mathfrak t}_F)).$
\bt\label{prop7.2A}
Let $T\ge 0$, $T = T^*\in [\cH]$, and $t_0:=\inf\gs(T)$.
Let $A$ be defined by \eqref{8.1} and $\gP = \{\cH,\gG_0,\gG_1\}$
the boundary triplet for $A^*$ defined by \eqref{8.2}.
Then

\item[\rm \;\;(i)] { the} Friedrichs extension $A^F$ coincides
with $A_0$ { that is}
\bed \dom(A^F) = \dom(A^*)\cap\dom({\mathfrak t}_F)= \{f \in
W^{2,2}(\R_+,\cH):\ f(0) = 0\}= \dom(A_0), \eed
{ and } $A^F$ corresponds to the Dirichlet problem. Moreover,
$A^F$ is absolutely continuous,  $A^F = (A^F)^{ac}$, and
$\gs(A^F)=\gs_{ac}(A^F)= [t_0,\infty).$

\item[\rm \;\;(ii)] { the} Krein extension $A^K$ is given by
\begin{equation}\label{7.4C}
\dom(A^K)=\{f\in W^{2,2}({\mathbb R}_+,\cH):\ f'(0) +
\sqrt{T}f(0)=0\}.
\end{equation}
Moreover, $\ker(A^K) = \gH_0 := \overline{\gH'_0}$,  $\gH'_0:=
\{e^{-x\sqrt{T}}h:\ h\in \ran(T^{1/4})\}$ and the restriction
$A^K\upharpoonright\dom(A^K) \cap \gH^{\perp}_0$ is absolutely
continuous, that is $\gH^{\perp}_0=\gH^{ac}(A^{K})$ and $A^K=
{0}_{\gH_0}\oplus (A^K)^{ac}$.

\item[\rm\;\;(iii)] The extension $A_1 := A^*\upharpoonright
\ker(\gG_1)$, { coincides with $A^N$},  $\dom(A^N) := \{f \in
W^{2,2}(\R_+,\cH): f'(0) = 0\}$, { i.e. $A_1$ corresponds to
the Neumann boundary condition}. Moreover, $A^N$ is absolutely
continuous $(A^N)^{ac} = A^N$ and $\gs(A^N) = \gs_{ac}(A^N) =
[t_0,\infty)$.

\item[\rm \;\;(iv)]  The operators $A^F$, $(A^K)^{ac}$ and $A^N$ are
unitarily  equivalent.
\et
\begin{proof}
(i) Let $\gP = \{\cH,\gG_0,\gG_1\}$ be  the boundary triplet
defined { in } Lemma \ref{lem6.1}.  We show that $A^F = A_0
:= A^*\upharpoonright\ker(\gG_0)$. It follows from  \eqref{8.1B}
and \eqref{8.2} that $\dom(A_0) = \{f \in W^{2,2}(\R_+,\cH): f(0)
= 0\}$. Since $\dom(A_0)\subset W^{1,2}_0(\R_+,\cH)=
\dom({\mathfrak t}_F),$  we have $A_0 = A^F$ (see \cite[Section
8]{AG81} and \cite[Theorem 6.2.11]{Ka76}).

It follows from \eqref{7.3} and \cite[Theorem 4.3]{BMN02} that
$\gs_p(A_0) = \gs_{sc}(A_0) = \emptyset$. Hence $A_0$ is
absolutely continuous. Taking into account Lemma \ref{VII.1}(iii)
and Proposition \ref{III.8} we get $\gs(A_0) = \gs_{ac}(A_0) =
\cl_{ac}(\supp(d_M)) = [t_0,\infty)$ which proves (i).

(ii)\ By \cite[Proposition 5]{DM91}  $A^K$ is defined by $A^K=
A^*\upharpoonright \ker(\gG_1 - M(0)\gG_0).$ It follows from \eqref{8.2B}
that $M(0)=-\sqrt{T}$. Therefore, $A^K$ is defined by \eqref{7.4C}.

It follows from the extremal property of the Krein extension that
$\ker(A^K) =\ker(A^*)$. Clearly,   $f_h(x) := \exp(-x\sqrt{T})h\in
L^2(\R_+,\cH),$  $h\in \ran(T^{1/4}),$ since
  \begin{equation*}
\int^{\infty}_0\!\|\exp(-x\sqrt{T})h\|^2_{\cH}dx=
\int^{\|T\|}_0\!d\rho_h(t)\int^{\infty}_0
\!e^{-2x\sqrt{t}}dx=\int^{\|T\|}_0 \!\frac{1}{2\sqrt
t}d\rho_h(t)<\infty,
    \end{equation*}
where $\rho_h(t):=\bigl(E_T(t)h,h\bigr)$. Thus,
$\gH'_0\subset \ker(A^*).$ It is easily seen that  $\gH'_0$ is
dense in $\gH_0.$ To investigate the rest of the spectrum of
$A^K$ consider the Weyl function $M_K(\cdot)$ corresponding to
$A^K$. It follows from \eqref{8.2B} and Proposition \ref{prop3.1}
that
\bead
\lefteqn{
M_K(z)=M_{-\sqrt{T}}(z)=-\bigl(\sqrt{T}+M(z)\bigr)^{-1}}\\
& &
= -(\sqrt{T}+i\sqrt{z-T})^{-1}
=\frac{1}{z}(i\sqrt{z-T}-\sqrt{T}) = -\frac{2\sqrt{T}}{z}+\Phi(z).
\nonumber
\eead
where $\Phi(z):=\frac{1}{z}[i\sqrt{z-T}+\sqrt{T}]$.
It follows that for $t>0$
\begin{equation}\label{7.13}
\im M_K(t+i0) =  \im\Phi(t+i0)=t^{-1}\sqrt{t-T}E_T([0,t)).
\end{equation}
Hence, by \cite[Theorem4.3]{BMN02},  $\sigma_p(A^K)\cap(0,\infty)=
\sigma_{sc}(A^K)\cap(0,\infty)=\emptyset$. It follows from
\eqref{7.13} that $\im{(M_K(t+i0))}>0$ for $t>t_0$. By Proposition
\ref{III.8} $\sigma_{ac}(A^K)= [t_0,\infty)$. Further,  it follows
from \eqref{7.3} and \eqref{7.13} that for  any $t>t_0$
\bead
\lefteqn{
d_M(t)= \rank(\im(M(t)))} \\
& & = \rank(E_T([0,t))) = \rank(\im(M_K(t)))= d_{M_K}(t)
\eead
Combining this equality with $\sigma_{ac}(A^K)= \sigma_{ac}(A^F)=
[t_0,\infty),$ we conclude that $A^F$ and $(A^K)^{ac}$ are unitarily
equivalent.

(iii) By Proposition \ref{prop3.1} the Weyl function corresponding to
$A_1= A^*\upharpoonright \ker(\gG_1 - 0\gG_0)$ is
\bed
M_0(z) := (0 -M(z))^{-1} = i(z - T)^{-1/2} = i\int
\frac{1}{\sqrt{z -\gl}}dE_T(\gl), \quad z \in \C_+.
\eed
Since  $M_0(\cdot)$ is regular within $(-\infty, t_0),$  we have
$(-\infty, t_0)\subset \varrho(A_1)$. Further, let $\tau
> t_0$. We set $\cH_\tau := E_T([t_0,\tau))\cH$ and note that for any $h \in
\cH_\tau$ and $t > \tau$
\begin{equation}\label{7.14}
\bigl(M_0(t+i0)h,h\bigr) = i\bigl((t-T)^{-1/2}h,h\bigr)=
i\int_{t_0}^{\tau} \frac{1}{\sqrt{t -\gl}}d(E_T(\gl)h,h).
\end{equation}
Hence  for any $h\in \cH_\tau\setminus\{0\}$ and   $t > \tau$ 
\bed
0 < 
(t-t_0)^{-1/2}\|h\|^2 \le \imag(M_0(t+i0)h,h) = \int_{t_0}^{\tau}
(t -\gl)^{-1/2}d(E_T(\gl)h,h) <\infty.
\eed
By \cite[Proposition 4.2]{BMN02},  $\gs_{ac}(A_1) \supseteq
[\tau,\infty)$ for any $\tau > t_0,$ which yields $\gs_{ac}(A_1) =
[t_0,\infty)$. It remains to show that $A_1$ is purely
absolutely continuous. Since $M_0(t+i0)\not\in [\cH]$ we cannot
apply \cite[Theorem 4.3]{BMN02} directly. Fortunately, to
investigate the $ac$-spectrum of $A_1$ we can use \cite[Corollary
4.7]{BMN02}. For any   $t \in \R,$ \ $y > 0,$ and $h \in \cH$ we
set
\bed V_h(t+iy) := \im(M_0(t+iy)h,h) =
\int\im\left(\frac{1}{\sqrt{\gl-t-iy}}\right)d(E_T(\gl)h,h). \eed
Obviously, one has
\bed V_h(t+iy) \le \int\frac{1}{((\gl-t)^2 +
y^2)^{1/4}}d(E_T(\gl)h,h), \quad t \in \R, \quad y > 0, \quad h
\in \cH. \eed
Hence
\bed V_h(t+iy)^p \le \|h\|^{2(p-1)} \int\frac{1}{((\gl-t)^2 +
y^2)^{p/4}}d(E_T(\gl)h,h), \quad p \in (1,\infty). \eed
We show that  for $p \in (1,2)$ and $-\infty < a < b < \infty$
\bed C_p(h; a,b) := \sup_{y\in (0,1]}\int^b_a V_h(t+iy)^p\;dt <
\infty. \eed
Clearly,
\bead
\int^b_a V_h(t+iy)^p dt \le \|h\|^{2(p-1)}\;\int^{\|T\|}_0
d(E(\gl)h,h) \int^b_a \frac{1}{((\gl-t)^2 + y^2)^{p/4}}dt \nonumber  \\
= \|h\|^{2(p-1)}\;\int^{\|T\|}_0 d(E(\gl)h,h) \int^{b-\gl}_{a-\gl}
\frac{1}{(t^2 + y^2)^{p/4}}dt. \quad
\eead
Note, that for $p \in (1,2)$ and  $-\infty < a < b < \infty$
  \begin{equation*}
\int^{b-\gl}_{a-\gl} \frac{1}{(t^2 + y^2)^{p/4}}dt \le
\int^{b}_{a-\|T\|} \frac{1}{t^{p/2}}dt =: \varkappa_p(b,a-\|T\|) <
\infty,
     \end{equation*}
Hence  $C_p(h;a,b) \le \varkappa_p(b,a-\|T\|)\|h\|^{2p} < \infty$
for $p \in (1,2)$, $-\infty < a < b < \infty$ and $h \in \cH$. By
\cite[Corollary 4.7]{BMN02},  $A_1$ is purely absolutely
continuous on any bounded interval $(a,b)$. Hence $A_1$ is purely
absolutely continuous.

(iv) It follows from \eqref{7.3} and \eqref{7.13} that
$d_M(t)=d_{M_K}(t)= \rank(\sqrt{t-T})$ {for}  $t>t_0$.
Combining this equality with $\sigma_{ac}(A^K)= \sigma_{ac}(A^F)=
[t_0,\infty),$ we conclude that $A^F$ and $(A^K)^{ac}$ are unitarily
equivalent.

Passing to $A_1,$ we assume that $1 \le \dim(\ran(E_T([0,s))))=p_1
< \infty$ for some $s > 0$. Let $\gl_k$, $k \in \{1,\ldots,p\}, \
p\le p_1$, be the set of distinct eigenvalues within $[0,s)$.
Since $M_0(t+iy)E_T([0,t))$ is the $p\times p$ matrix-function,
the limit $M_0(t+i0)E_T([0,t))$ exists for $t \in [0,s) \setminus
\bigcup^p_{k=1}\{\gl_k\}.$   It follows from \eqref{7.14} that
\bed
\im(M_0(t)) = |T-t|^{-1/2}E_T([0,t)),\qquad t \in [0,s)
\setminus \bigcup^p_{k=1}\{\gl_k\}.
\eed
This yields
\bed
d_{M_0(t))} := \dim(\ran(\im(M_0(t)))) =
\dim(\ran(E_T([0,t)))) = d_M(t)
\eed
for a.e $t \in [0,s) \setminus \bigcup^p_{k=1}\{\gl_k\}$, that
is, for a.e. $t \in [0,s)$.

If $\dim(E_T([t_0,s))) = \infty$, then there exists  a point
$s_0\in (0, s),$ such that $\dim(E_T([0,s_0])) = \infty$ and
$\dim(E_T([0,s))) < \infty$ for $s \in [0,s_0)$. For any  $t \in
(s_0,s)$ choose  $\tau \in (s_0,t)$ and note that
$\dim(\ran(E_T([0,\tau)))) = \infty$. We set $\cH_\tau :=
E_T([0,\tau))\cH$ and $\cH_\infty := E_T([\tau,\infty))\cH$ as
well as $T_\tau := TE_T([0,\tau))$ and $T_\infty :=
TE_T([\tau,\infty))$. Further, we choose Hilbert-Schmidt
operators $D_\tau$ and ${D}_\infty$ in $\cH_\tau$ and
$\cH_\infty$, respectively, such that $\ker(D_\tau) =
\ker(D^*_\tau) = \ker(D_\infty) = \ker(
D^*_\infty) = \{0\}$. According to the decomposition
$\cH=\cH_\tau \oplus\cH_\infty$ we have $M_0 = M_\tau
\oplus M_\infty,$ ${ D} = { D}_\tau \oplus{ D} _\infty$ and
$d_{M_0^{ D}}(t) = d_{M_\tau^{{ D}_\tau}}(t) + d_{M_\infty^{
D_\infty}}(t)$ for a.e. $t \in [0,\infty).$
Hence $d_{M_0^{D}}(t) \ge d_{M_\tau^{{ D}_\tau}}(t)$ for a.e. $t
\in [0,\infty)$. Clearly, $M_\tau (t+iy) = i(t + iy -
T_\tau)^{-1/2}.$ If $t > \tau$, then $t\in \varrho(T_\tau)$ and
$M(t) := \slim_{y\to+0} M(t+i0)$ exists and
\bed
M_\tau(t) := \slim_{y\to 0}M_\tau(t+iy) = i(t -
T_\tau)^{-1/2}E_T([0,\tau)).
\eed
Hence $d_{M_\tau^{D_\tau}}(t) = \dim(\ran(E_T([0,\tau)))) =
\infty$ for $t > s_0$. Hence $d_{M_0^D}(t) = d_M(t) = \infty$ for a.e.
$t > s_0$ which yields $d_{M_0^D}(t) = d_M(t)$ for a.e. $t \in
[0,\infty)$. Using Theorem \ref{IV.10}(ii) we obtain that $A^{ac}_0$
and $A^{ac}_1$  are unitarily equivalent which shows $A_0$ and $A_1$
are unitarily equivalent.
\end{proof}

Next we describe the spectral properties of any self-adjoint
extension of $A.$ In particular, we show that the
Friedrichs extension $A^F$ of $A$ is $ac$-minimal, though $A$ does
not satisfy conditions of Theorem  \ref{VI.6}.
\bt\la{VII.3}
Let $T\ge 0$, {$T = T^*\in [\cH]$}, and $t_1:=\inf\gs_{\ess}(T)$.
Let also $A$ be the symmetric operator defined by \eqref{8.1} and
$\wt A = \wt A^* \in \Ext_A$.  Then

\item[\rm\;\;(i)] the absolutely continuous part $\wt A^{ac}
E_{\wt A}([t_1,\infty))$ of  $\wt A E_{\wt A}([t_1,\infty))$ is
unitarily equivalent to $A^FE_{A^F}([t_1,\infty)) =
(A^F)^{ac}E_{A^F}([t_1,\infty))$;

\item[\rm\;\;(ii)] the Friedrichs extension $A^F$ is
$ac$-minimal { and} $\sigma_{ac}(A^F)\subseteq
\sigma_{ac}(\wt A)$;

\item[\rm \;\;(iii)] the absolutely continuous part $\wt A^{ac}$
of $\wt A$ is unitarily equivalent to $A^F$  whenever either $(\wt
A - i)^{-1} - (A^F - i)^{-1} \in \gotS_\infty(\gotH)$ or $(\wt A -
i)^{-1} - (A^K - i)^{-1} \in \gotS_\infty(\gotH)$.
\et
\begin{proof}
By Corollary \ref{V.1A} it  suffices to assume that the
extension $\wt A = \wt A^*$  is disjoint with $A_0,$ that is, by
Proposition \ref{prop2.1}(ii) it admits a representation  $\wt A
= A_B$ with $B\in \cC(\cH).$

(i) Let $\Pi=\{\cH, \gG_0,\gG_1\}$ be a boundary triplet for $A^*$
defined by  \eqref{8.2}.  In accordance with Theorem \ref{IV.10}
we calculate $d_{M_B^K}(t)$ where $M_B(\cdot) :=
(B-M(\cdot))^{-1}$ is the generalized Weyl function of
the extension $A_B$. Clearly,
\be\label{7.15}
 \im(M_B(z)) = M_B(z)^*\im(M(z))M_B(z), \quad z \in
\C_+. \ee
Since  $\re(\sqrt{z - \gl})>0$ for $z = t +iy,\  y>0,$  it
follows from \eqref{8.2B} that
      \be\label{7.16}
  \im(M(z)) = \int_{[0,\infty)}\re(\sqrt{z -
\gl})\;dE_T(\gl) \ge \int_{[0,\tau)}\re(\sqrt{z -
\gl})\;dE_T(\gl),
      \ee
where $z = t +iy$. It is easily seen that
      \be \label{7.16A}
\re(\sqrt{z - \gl}) \ge \sqrt{t-\gl} \ge
\sqrt{t-\tau}, \quad \gl \in [0,\tau), \quad t > \tau.
     \ee
Combining \eqref{7.15} with \eqref{7.16} and \eqref{7.16A} we get
\bed  \im(M_B(t+iy)) \ge
\sqrt{t-\tau}M_B(t+iy)^*E_T([0,\tau))M_B(t+iy), \quad t > \tau
>0. \eed
Let $Q$ be a finite-dimensional orthogonal projection, $Q \le
E_T([0,\tau))$. Hence
\bed \im(M_B(t+iy)) \ge \sqrt{t-\tau}M_B(t+iy)^*QM_B(t+iy), \quad
t > \tau >0, \quad y > 0. \eed
Setting $\cH_1= \ran(Q)$, $\cH_2 := \ran(Q^\perp)$, and
choosing $K_2\in \mathfrak S_2(\cH_2)$ and  satisfying $\ker(K_2)
= \ker(K^*_2) = \{0\}$, we define a Hilbert-Schmidt operator $K :=
Q \oplus K_2\in \mathfrak S_2(\cH).$  Clearly,  $\ker(K) =
  \ker(K^*) = \{0\}$ and,
\bea\label{6.14A}
\lefteqn{
\im(K^*M_B(t+iy)K) \ge }\\
& &
\sqrt{t-\tau}K^*M_B(t+iy)^*QM_B(t+iy)K, \quad t > \tau >0.
\nonumber
\eea
Since  $M_B(\cdot)\in (R_\cH)$ and   $Q,\ K\in \mathfrak
S_2(\cH),$  the limits
\bead
K^*M_B(t)^*Q & := & \slim_{y\to+0}K^*M_B(t + iy)^*Q \quad \mbox{and} \quad \\
(QM_BK)(t)   & := & \slim_{y\to+0}QM_B(t+iy)K
\eead
exist for a.e. $t\in \R$ (see \cite{BirEnt67}). Therefore
passing to the limit as $y\to 0$ in \eqref{6.14A}, we arrive at
the inequality
\bed
\im(M_B^K(t)) \ge \sqrt{t-\tau}(K^*M_B(t)^*Q)(QM_BK(t)), \quad  t
> \tau >0, \quad y> 0.
\eed
It follows that
\be\label{7.18}
\dim(\ran\left((QM_BK)(t)\right)) \le
\dim(\ran\left(\im M_B^K(t)\right)) =d_{M^K_B}(t), \quad t > \tau.
\ee
We set $\wt M^Q_B(z) :=QM_B(z)Q \upharpoonright\cH_1.$
Since $\dim(\cH_1)<\infty$ the limit $\wt M^Q_B(t):=
\slim_{y\to+0}\wt M^Q_B(t+iy)$ exists  for a.e. $t\in \R.$
Since $(QM_BK)(t)\upharpoonright\cH_1 = \ran\left((\wt
M^Q_B)(t)\right)$, \eqref{7.18} yields the inequality
        \be\la{7.18a}
\dim(\ran\left(\wt M^Q_B(t)\right)) \le
\dim(\ran\left((QM_BK)(t)\right)) \le d_{M^K_B}(t)
         \ee
for a.e. $t \in [\tau,  \infty)$.

Since  $\dim(\cH_1) <\infty$ and $\ker(\wt M^Q_B(z)) =\{0\},
z\in \C,$ we easily get by  repeating the corresponding reasonings
of the proof of Theorem   \ref{VI.6} that $\ran\left(\wt
M^Q_B(t)\right) = \cH_1$ for a.e. $t \in \R$. Therefore
\eqref{7.18a} yields  $\dim(\cH_1) \le d_{M^K_B}(t)$ for a.e. $t
\in [\tau, \infty).$

If $\tau > t_1$, then $\dim(E_T([0,\tau))\cH) =\infty$ and
the dimension of a projection $Q \le E_T([0,\tau))$ can be
arbitrary. Thus, $d_{M_B^K}(t) = \infty$ for a.e. $t >
\tau$. Since $\tau > t_1$ is arbitrary we get
$d_{M^K_B}(t) = \infty$ for a.e. $t
> t_1$. By  Theorem \ref{IV.10}(ii)
the operator $\wt A^{ac}E_{\wt A}([t_1,\infty))$ is
unitarily equivalent to $A_0E_{A_0}([t_1,\infty))$.

(ii) If $\tau \in (t_0, t_1)$, then $\dim(E_T([0,\tau))\cH) =:
p(\tau)<\infty.$  Hence, $\dim(Q\cH) \le p(\tau)$ which shows that
$d_{M_B^K}(t) \ge p(\tau)$ for a.e. $t \in (\tau, t_1)$. Since
$\tau$ is arbitrary, we obtain $d_{M_B^K}(t) \ge p(\tau)$ for
a.e. $t \in [0,t_1)$. Using Theorem \ref{IV.10}(i) we prove (ii).

(iii) By Lemma \ref{VII.1} the invariant maximal normal
function $\gotm^+(t)$ is finite for $t \in \R$. By Theorem
\ref{V.5} $\wt A^{ac}$ and $(A^F)^{ac}$ are unitarily equivalent.
Similarly we prove that $\wt A^{ac}$ and $(A^K)^{ac}$ are
unitarily equivalent. To complete the proof it remains to
apply Theorem \ref{prop7.2A}(iv).
\end{proof}
\bc\la{VI.5Z}
{ Let the assumptions of Theorem \ref{VII.3} be satisfied.
If $\dim(\cH) = \infty$ and $t_0:=\inf\gs(T) =
\inf\gs_{\ess}(T)=:t_1$, then
the Friedrichs  extension $A^F$ is strictly $ac$-minimal.}
\ec
\br {\em
Let $\dim(E_T([t_0, t_1))\cH) = \infty$.
{ Then there are self-adjoint extensions $\wt A= \wt A^*\in
\Ext_A$ of $A$ such that $\gs_{ac}(\wt A) =  \gs(A^F) = \gs_{ac}(A^F)$ but
$\wt A$ is not unitarily equivalent to $A^F$. }
}\er

\subsection{Unbounded operator potentials}

In this subsection we  consider the differential expression
\eqref{8.1} with unbounded  $T = T^* \ge 0,$\ $T\in \cC(\cH),$
\be\label{8.10}
(\cA_T f)(x)  =  -\frac{d^2}{dx^2} f(x) + Tf(x).
\ee
The minimal operator $A :=A_{T, \min}$ is defined as the closure
of the operator $A'_T$ generated on $\mathfrak H:= L^2(\R_+, \cH)$
by  expression \eqref{8.10} on the domain
\be \label{8.11}
\cD'_0   :=  \left\{\sum_{1\le j \le k}\phi_j(x)h_j:\
\phi_j \in W_0^{2,2}(\R_+),\;
           h_j \in \dom(T), \;\  k \in \N \right\},
\ee
that is $A'_Tf=\cA_T f,\ \dom(A'_T)=\cD'_0.$  Clearly,  $A$ is
non-negative, since  $T\ge0$  and  $A_{T,min} :=
\overline{A'_T}$ coincides with $A$ defined by \eqref{8.1}
provided that $T$ is bounded.

Let $\cH_T$ be the Hilbert space which is obtained equipping the set
$\dom(T)$ with the graph norm of $T$.
Following \cite{LioMag72} we introduce the Hilbert spaces  $W^{k,2}_T({\R}_+;\cH) :=
W^{k,2}({\R}_+;\cH) \cap L^2({\R}_+,\cH_T)$, $k \in \N$, equipped with the
Hilbert norms
\bed \|f\|^2_{W^{k,2}_T}
=\int_{{\R}_+}\bigl(\|f^{(k)}(t)\|^2_{\cH} + \|f(t)\|^2_\cH +\|T
f(t)\|_{\cH}^2\bigr)dt.
\eed
Obviously,  we have $W^{2,2}_{0,T}(\R_+,\cH) := \{f\in
W^{2,2}_T({\R}_+;\cH):\  f(0)=f'(0) =0\} \subseteq
\dom(A_{T,\min})$.
\bl\la{VI.2.7}
{
{ Let $T=T^*$ be a non-negative  operator in $\cH$. Then}
$\dom(A_{T,\rm min})$ and $W^{2,2}_{0,T}(\R_+,\cH)$ {
coincide  algebraically and topologically.}
}
\el
\begin{proof}
{
Obviously, {  for any $f \in \cD'_0$ we have}
\bead
\lefteqn{
\left\|\cA_Tf\right\|^2_\gotH
= \int_{\R_+}\left\|f''(x)\right\|^2_\cH dx } \\
     & &
 + \int_{\R_+}\|Tf(x)\|_\cH ^2dx - 2\real\left\{\int_{\R_+}
\left(f''(x),Tf(x)\right)_\cH dx\right\}.
           \eead
Integrating by parts we find
         \bed
\int_{\R_+} \left(f''(x),Tf(x)\right)dx =
-\int_{\R_+}\left\|\sqrt{T}f'(x)\right\|^2_\cH dx.
        \eed
Combining these equalities  yields
      \begin{equation*}
\left\|\cA_Tf\right\|^2_\gotH =
\int_{\R_+}\left\|f''(x)\right\|dx +
  \int_{\R_+}\|Tf(x)\|^2dx +
2\int_{\R_+}\left\|\sqrt{T}f'(x)\right\|^2_\cH dx
             \end{equation*}
for any $f \in \cD'_0$. { Hence }
\bed
\|f\|^2_{W^{2,2}_T} \le \|\cA_Tf\|^2_\gotH + \|f\|^2, \quad f \in
\cD'_0.
\eed
Furthermore, { by the Schwartz inequality,}
\bed
2\left|\real\left\{\int_{\R_+} \left(f'(x),Tf(x)\right)_\cH dx\right\}\right| \le
\|f\|^2_{W^{2,2}_T}, \quad f \in \cD'_0.
\eed
which gives
\bed
\|\cA_Tf\|^2_\gotH + \|f\|^2 \le 2\|f\|^2_{W^{2,2}_T}, \quad f \in \cD'_0.
\eed
{ Thus, we arrive at the two-sided  estimate }
     \bed
\|f\|^2_{W^{2,2}_T} \le \left\|\cA_Tf\right\|^2_\gotH +
\|f\|^2_\gotH \le 2\|f\|^2_{W^{2,2}_T}\qquad f \in \cD'_0.
       \eed
{ Since  $\cD'_0$ is dense in
$W^{2,2}_{0,T}$ too,  we obtain that  $\dom(A_{T,\rm
  min})$ coincides with $W^{2,2}_{0,T}$ algebraically and topologically.}
}
\end{proof}

{  Since $A$ is  non-negative it admits the Friedrichs
extension $A^F$ and the Krein extension $A^K$.
We define the
extension $A^N$ as the self-adjoint operator associated with the
closed quadratic form ${\mathfrak t}_N$,
\be
{\mathfrak t}_N[f]  := \int^\infty_0 \left\{\|f'(x)\|^2_\cH +
\|{\sqrt T}f(x)\|^2_\cH\right\}dx = \|u\|^2_{W^{1,2}_{\sqrt T}} - \|u\|^2_{L^2(\R_+,\cH)} ,
\ee
$ \dom({\mathfrak t}_N)  :=   W^{1,2}_{\sqrt T}(\R_+,\cH)$.
The definition makes  sense for $T\in [\cH]$. In this case
$A^N=A_1$  with $A_1$ defined in Theorem \ref{prop7.2A}(iii). }

{ We also put $\gt_F := \gt_N\upharpoonright\dom(\gt_F)$,
$\dom(\gt_F) := \{f\in W^{1,2}_{\sqrt{T}}(\R_+;\cH):\ f(0)=0\}. $}
\bp\label{prop6.8}
Let $T=T^*\in\cC(\cH)$, $T\ge 0$, and { let} $A:=A_{T,\min}$
{ be } defined by \eqref{8.10}-\eqref{8.11}. Let also
$\cH_n:=\ran\bigl(E_T([n-1,n))\bigr)$, $T_n:= T E_T([n-1,n))$ and
let $S_n$ { be} the closed minimal symmetric operator defined
by \eqref{8.1} in $\gH_n:=L^2({\mathbb R}_+,\cH_n)$ with $T$
replaced by $T_n.$  Then
\item[\rm\;\;(i)]  { the following decompositions hold}
        \be\label{8.23}
 A = \bigoplus^\infty_{n=1} S_n, \quad A^F =
\bigoplus^\infty_{n=1} S^F_n, \quad A^K = \bigoplus^\infty_{n=1}
S^K_n, \quad A^N = \bigoplus^\infty_{n=1}S^N_n.
     \ee
\item[\rm\;\;(ii)] { The domain $\dom(A^F)$ equipped with the
graph norm   is a closed subspace of $W^{2,2}_T(\R_+,\cH),$}
\bed
\dom(A^F) = \{f \in W^{2,2}_T(\R_+,\cH):  f(0) = 0\}.
\eed

\item[\rm\;\;(iii)] { The domain $\dom(A^N)$ equipped with the
graph norm   is a closed subspace of $W^{2,2}_T(\R_+,\cH),$}\
$\dom(A^N) = \{f \in W^{2,2}_T(\R_+,\cH): f'(0)=0\}$.
\end{proposition}
\begin{proof}
{ (i)} We introduce the sets
\bed
\cD''_0   :=  \left\{\sum_{1\le j \le k}\phi_j(x)h_j:\
\phi_j \in W_0^{2,2}(\R_+),\;
           h_j \in \cH_n, \; k,n \in \N \right\}
\eed
and $\cD''_{0n}   :=  \left\{f \in \cD''_0: f(x) \in \cH_n\right\},
\quad n \in \N.$ Obviously, we have
$\cD''_0 = \bigoplus^\infty_{n=1}\cD''_{0n} \subseteq \cD'_0$.
Setting $A''_T := A'_T\upharpoonright\cD''_0$ we find
$\overline{A''_T} = \overline{A'_T} = A_{T,\rm min}$. Moreover,
setting $A''_n := A_n\upharpoonright\cD''_{0n}$, $n \in \N$, we have
$\overline{A''_n} = A_n$, $n \in \N$. Since
$A''_T = \bigoplus^\infty_{n=1}A''_n \subseteq A'_T,$
we obtain
\bed
A_{T,min} = \overline{A''_T} = \bigoplus^\infty_{n=1}\overline{A''_n}
= \bigoplus^\infty_{n=1}A_n \subseteq A_{T,min}
\eed
which proves { the first relation of \eqref{8.23}. The second
and the third relations are implied by  Corollary \ref{cor5.5}}.

{ To prove the last relation of \eqref{8.23} we set $S^N :=
\bigoplus^\infty_{n=1}S^N_n.$ Since $S^N_n= (S^N_n)^*\in
\Ext_{S_n}$ and $A = \oplus^\infty_{n=1} S_n,$  $S^N$ is a
self-adjoint extension of $A,\ S^N\in \Ext_A.$} Let
$f=\oplus^\infty_{n=1} f_n\in \dom(S^N).$  { Then integrating
by parts we obtain}
\bead \lefteqn{ (S^Nf,f) = \sum_1^{\infty}(S^N_nf_n, f_n) =
\sum_{n=1}^{\infty} \int^\infty_0 \left\{\|f'_n(x)\|^2_{\cH_n}  +
  \|{\sqrt T_n}f_n(x)\|^2_{\cH_n}\right\}dx} \\
& & = \int^\infty_0 \left\{\|f'(x)\|^2_{\cH} +
  \|{\sqrt T}f(x)\|^2_{\cH}\right\}dx
= {\mathfrak t}_N[f].
\eead
{ Since, by definition, $A^N$ is associated with the
quadratic form ${\mathfrak t}_N,$ the last equality yields
$S^N\subset A^N.$ Hence $S^N =  A^N,$ since both $S^N$ and  $A^N$
are self-adjoint extensions of $A$.}

{
(ii)\  Following the reasoning of Lemma \ref{VI.2.7} we find
       \be\la{8.23b}
\|f_n\|^2_{W^{2,2}_{T_n}} \le \|S^F_n f_n\|^2_{\gotH_n} +
\|f_n\|^2_{\gotH_n} \le 2\|f_n\|^2_{W^{2,2}_{T_n}}, \qquad n \in
\N,
         \ee
{ where}  $f_n \in \dom(S^F_n) = \{g_n \in
W^{2,2}(\R_+,\cH_n): g_n(0) =0\}$. { Using
 representation \eqref{8.23} for $A^F$ and setting $f^m :=
\oplus^m_{n=1} f_n$, $f_n \in \dom(F_n)$, we obtain from
\eqref{8.23b}}
\be\la{8.23a}
 \|f^m\|^2_{W^{2,2}_T} \le \|A^F f^m\|^2_{\gotH} +
\|f^m\|^2_{\gotH} \le 2\|f^m\|^2_{W^{2,2}_T}, \quad m \in \N.
\ee
{ Since the set  $\{f^m = \oplus^m_{n=1} f_n: \ f_n \in
\dom(S^F_n),\ m \in \N \}$, is  a core for $A^F,$ inequality
\eqref{8.23a} remains valid for} $f \in \dom(A^F)$. This shows
that $\dom(A^F) = \{f \in W^{2,2}_T(\R_+,\cH): f(0) = 0\}$. {
Moreover, due to \eqref{8.23a}  the graph norm of $A^F$ and  the
norm $\|\cdot\|_{W^{2,2}_T}$ restricted to $\dom(A^F)$ are
equivalent}.

(iii)\  { Similarly to  \eqref{8.23b} one gets}
     \bed
\|f_n\|^2_{W^{2,2}_{T_n}} \le \|S^N_n f_n\|^2_{\gotH_n} + \|f_n\|^2 \le 2\|f_n\|^2_{W^{2,2}_{T_n}}
         \eed
for $f_n \in \dom(S^N_n) = \{g_n \in W^{2,2}(\R_+,\cH_n): g'_n(0)
= 0\}$, $n \in \N$. { It remains to repeat the reasonings  of
(ii)}  }
\end{proof}

To { extend} Theorem \ref{prop7.2A} { to the case of
unbounded operators $T=T^*\ge0$} we first construct a boundary
triplet for $A^*$, using { Theorem \ref{VI.3} and
representation \eqref{8.23} for} $A.$
\bl\label{lem6.9}
{
Assume conditions  of Proposition} \ref{prop6.8}. Then there is a
sequence of boundary triplets $\wh\gP_n =
\{\cH_n,\wh\gG_{0n},\wh\gG_{1n}\}$ { for $S_n^*$} such that
$\gP := \oplus^\infty_{n=1}\wh\gP_n =: \{\cH,\wh
  \gG_0,\wh \gG_1\}$ { forms an ordinary  boundary triplet for} $A^*.$
{ Moreover,  $A^F = A^*\upharpoonright\ker(\wh \gG_0)$ and
the corresponding  Weyl function is}
      \be\la{8.7}
\wh M(z) = \frac{i\sqrt{z-T}
  +\im(\sqrt{i - T})}{\re(\sqrt{i-T})}.
\qquad z \in \C_+,
       \ee
 \el
\begin{proof}
For any $n \in \N$ we define a boundary triplet $\gP_n =
\{\cH_n,\gG_{0n},\gG_{1n}\}$ for { $S_n^*$} with
$\gG_{0n},\gG_{1n}$ defined by \eqref{8.2}. By Theorem
\ref{prop7.2A}(i) { $S^F_n = S_{0n} =
S_n^*\upharpoonright\ker(\gG_{0n})$} and by { Lemma
\ref{lem6.1}}
the corresponding Weyl function is $M_n(z) = i\sqrt{z-T_n}$.

Following Lemma \ref{VI.1}, cf. \eqref{3.8}, we  define a sequence
of regularized boundary triplets $\wh\gP_n =
\{\cH_n,\wh\gG_{0n},\wh\gG_{1n}\}$ { for $S_n^*$} by setting
$R_n:=(\re(\sqrt{i-T_n}))^{1/2},$ \ $Q_n:= -\im(\sqrt{i - T_n})$
and
\be\label{6.22A}
\wh\gG_{0n} := R_n\gG_{0n}, \quad \wh\gG_{1n} := R_n^{-1}(\gG_{1n}
- Q_n\gG_{0n}), \qquad n\in \N.
\ee
Hence { $S^F_n = S_{0n}$} and the corresponding  Weyl
function $\wh M_n(\cdot)$ is given by
\be \la{8.7B}
\wh M_n(z) = \frac{i\sqrt{z-T_n}
  +\im(\sqrt{i - T_n})}{\re(\sqrt{i-T_n})},
\qquad z \in \C_+, \qquad n \in \N.
\ee
By  Theorem \ref{VI.3} the direct sum $\wh\gP :=
\bigoplus_{n=1}^\infty\wh\gP_n  = \{\cH,\wh\gG_{0},\wh\gG_{1}\}$
forms a boundary triplet for $A^*$ and
the corresponding Weyl function is
\be\la{8.7C}
 \wh M(z) = \bigoplus_{n\in\N}\wh M_n(z) \qquad z \in \C_+.
\ee
Combining \eqref{8.7C} with  \eqref{8.7B} we arrive at
\eqref{8.7}.

{
Combining Theorem \ref{VI.3} (cf. \eqref{6.8})  with
Corollary \ref{cor5.5} we get
\be\la{6.25z}
A_0 = A^*\upharpoonright\ker(\wh \gG_0) =
\bigoplus^\infty_{n=1}S^*_n\upharpoonright\ker(\wh \gG_{0n}) =
\bigoplus^\infty_{n=1}S_{0n} =  \bigoplus^\infty_{n=1}S^F_{n}=A^F
\ee
which proves the second assertion.}
\end{proof}

{  Next we generalize Theorem \ref{prop7.2A} to the case of
unbounded operator potentials.}
\bt\la{VI.9}
Let $T = T^* \ge 0,$\  $t_0 := \inf\gs(T)$, and  $A := A_{T,min}$,
cf. \eqref{8.10}-\eqref{8.11}. { Let also  $\wh \gP =
\{\cH,\wh \gG_0,\wh \gG_1\}$ be the boundary triplet for $A^*$
defined by Lemma  \ref{lem6.9} and $\wh M(\cdot)$ the
corresponding  Weyl function (cf. \eqref{8.7}).}
{ Then}

\item[\rm \;\;(i)] The Friedrichs extension $A^F$ coincides with
$A_0 := A^*\upharpoonright\ker(\wh \gG_0).$
Moreover, $A^F$ is absolutely continuous, $A^F =
(A^F)^{ac}$, and $\gs(A^F) = \gs_{ac}(A^F) = [t_0,\infty)$.

\item[\rm \;\;(ii)] The Krein extension $A^K$ is given by
$A_{B^K}:=A^*\upharpoonright \ker(\gG_1 - B^K\gG_0)$, where
\be\la{6.32x}
B^K = \frac{1}{\sqrt{2}\sqrt{T} + \sqrt{T + \sqrt{1
+ T^2}}} \frac{1}{\sqrt{T + \sqrt{1 + T^2}}}.
\ee
Moreover, $\ker(A^K) = \gotH_0 := \overline{\gotH'_0}$, $\gotH'_0
:= \{e^{-x\sqrt{T}}h: h \in \ran(T^{1/4})\},$ the restriction $A^K
\upharpoonright\dom(A^K) \cap \gotH^\perp_0$ is absolutely
continuous, and $A^K = 0_{\gotH_0} \oplus (A^K)^{ac}$.

\item[\rm \;\;(iii)]
{ The  extension $A^N$ is given by $A^N =
A^*\upharpoonright \ker(\wh \gG_1 - B^N\wh \gG_0)$ where $B^N :=
\sqrt{T + \sqrt{1 + T^2}}$.} Moreover, $A^N$ is absolutely
continuous, $A^N = (A^N)^{ac}$ and $\gs(A^N) = \gs_{ac}(A^N) =
[t_0,\infty)$.

\item[\rm \;\;(iv)]
The operators $A^F$, $(A^K)^{ac}$ and and $A^N$ are unitarily
equivalent.
\et
\begin{proof}
(i) { This statement is implied by combining  Theorem
\ref{prop7.2A} with \eqref{6.25z}}.

(ii) Using  the polar decomposition $i-\gl = \sqrt{1 +
\gl^2}e^{i\theta(\gl)}$ with $\theta(\gl) = \pi - \arctan(1/\gl)$,
$\gl \ge 0$ we get
\be \label{6.23A} \re(\sqrt{i - T}) = \int^\infty_0 \sqrt[4]{1 +
\gl^2}\cos(\theta(\gl)/2)dE_T(\gl). \ee
Setting  $\varphi(\gl) = \arctan(1/\gl)$, $\gl \ge 0$ and noting
that $\cos(\varphi(\gl)) = {\gl}{(1 + \gl^2})^{-1/2},$ we find
$\cos(\theta(\gl)/2) = {2}^{-1/2}(1 + \gl^2)^{-1/4}(\gl + \sqrt{1
+ \gl^2})^{-1/2}.$ Substituting this expression in \eqref{6.23A}
{ yields}
       \be\la{6.30a}
\re(\sqrt{i - T}) = {2}^{-1/2}(T + \sqrt{1 + T^2}\
)^{-1/2}.
       \ee
Similarly, taking into account $\sin(\theta(\gl)/2) =
\cos(\varphi(\gl)/2)$ and $\cos(\varphi(\gl)/2) = {2}^{-1/2}(1 +
\gl^2)^{-1/4}(\gl + \sqrt{1 + \gl^2})^{1/2}$, we get
        \be\la{6.35}
\im(\sqrt{i- T}) = \int^\infty_0 \sqrt[4]{1 +
  \gl^2}\cos(\varphi(\gl)/2)dE_T(\gl) = \frac{1}{\sqrt{2}}\sqrt{T + \sqrt{1 + T^2}}.
         \ee
 { It follows from  \eqref{8.7} with account of \eqref{6.30a} and \eqref{6.35} that}
 $M(0):=\slim_{x\to+0}M(-x) =: B^K$ { where $B^K$ is} defined
by \eqref{6.32x}. Therefore, by \cite[Proposition 5(iv)]{DM91} the
Krein extension $A^K$ is given by $A_{B^K}:=A^*\upharpoonright
\ker(\gG_1 - B^K\gG_0).$ { The second statement follows from Proposition \ref{prop6.8} and
Theorem \ref{prop7.2A}(ii).}

(iii)   It is easily seen that in the boundary triplet $\wh \gP_n
= \{\cH_n,\wh \gG_{0n},\wh \gG_{1n}\}$ { defined by
\eqref{6.22A}} the  extension $A^N_n$ admits a representation
$A^N_n = A_{B_n}$ where  $B_n := \sqrt{T_n + \sqrt{1 + T^2_n}}$,
$n \in \N$. By Proposition \ref{prop6.8},  $A^N =
\oplus^\infty_{n=1}A^N_n = A_{B^N}$ where $B^N =
\oplus^\infty_{n=1}B_n$. The remaining part of (iii) follows from
the representation $A^N = \bigoplus^\infty_{n=1}A^N_n$ and Theorem
\ref{prop7.2A}(iii).

(iv) { The assertion follows from  Theorem \ref{prop7.2A}(iv)
and \eqref{8.23}.}
\end{proof}

{  Next we  generalize Theorem \ref{VII.3} to the case of
unbounded $T\ge0.$}
\bt\la{VI.10a}
Let $T = T^* \ge 0$ and $t_1 := \inf\gs_{\rm ess}(T)$.
Further, let $A := A_{T,\rm min}$, cf. \eqref{8.10}-\eqref{8.11},
and $\wt A = \wt A^* \in Ext_A$. Then

\item[\rm\;\;(i)] the absolutely continuous part
$\wt A^{ac} E_{\wt A}([t_1,\infty))$ is unitarily
equivalent to the part $A^FE_{A^F}([t_1,\infty)) =
(A^F)^{ac}E_{A^F}([t_1,\infty))$;

\item[\rm\;\;(ii)] the Friedrichs extension $A^F$ is
$ac$-minimal { and} $\sigma_{ac}(A^F) \subseteq
\sigma_{ac}(\wt A);$

\item[\rm \;\;(iii)] the $ac$-part $\wt A^{ac}$ of $\wt A$
is unitarily equivalent to $A^F$ if either $(\wt A - i)^{-1} -
(A^F - i)^{-1} \in \gotS_\infty(\gotH)$ or $(\wt A - i)^{-1} -
(A^K - i)^{-1} \in \gotS_\infty(\gotH)$.
\et
\begin{proof}
By Corollary \ref{V.1A} it  suffices to assume that the
extension $\wt A = \wt A^*$  is disjoint with $A_0,$ that is,
it admits a representation  $\wt A = A_B$ with $B\in \cC(\cH).$

(i) { We} consider the boundary triplet $\wh \gP = \{\cH,\wh
\gG_0,\wh \gG_1\}$ defined { in Lemma \ref{lem6.9}}. By
Proposition \ref{prop3.1} the generalized Weyl function
corresponding to the generalized boundary triplet $\wh \gP_B$
{ is defined by} $\wh M_B(z) = (B - \wh M(z))^{-1}$, $z \in
\C_+$, where $\wh M(z)$ is given by \eqref{8.7}. Clearly,
\be\la{6.32AB}
\im(\wh M_B(z)) = \wh M_B(z)^*\im(\wh M(z))\wh M_B(z), \qquad z
\in \C_+.
\ee
{ It follows from \eqref{8.7} that  $\bigl(\re(\sqrt{i -
T})\bigr)^{-1}\ge \sqrt 2.$ Therefore \eqref{6.30a} yields}
\be\la{6.32a}
 \im(\wh M(z)) \ge \sqrt{2}\im(M(z)), \quad z \in
\C_+, \quad \text{where} \quad M(z) = i\sqrt{z - T},
\ee
cf.  \eqref{8.2B}. Following the line of reasoning of the proof of
Theorem \ref{VII.3}(i) we obtain from \eqref{6.32a} that $d_{\wh
M^D}(t)=\infty$ for a.e. $t\in [t_1, \infty),$ where  $D=D^*\in
\gotS_2(\cH)$ and $\ker D=\{0\}.$ Moreover, it follows from
\eqref{6.32AB} that $d_{\wh M_B^D}(t)= d_{\wh M^D}(t)=\infty$ for
a.e. $t\in [t_1, \infty).$ One completes the proof by applying
Theorem \ref{IV.10}.

(ii) To prove (ii) we use  again  estimates \eqref{6.32a} and
follow the proof of Theorem \ref{VII.3}(ii). We complete the proof
by applying Theorem \ref{IV.10}.

(iii) The Weyl function $\wh M(\cdot)$
is given by  \eqref{8.7}. { Taking into account
\eqref{8.7C} one obtains  $\sup_{n\in\N}\gotm^+_n < \infty$,
where $\gotm^+_n$ is the maximal normal invariant function defined by
\eqref{4.12A}. Indeed, this follows from
\eqref{7.3a} because this estimate shows that $\gotm^+_n$ does not
depend on $n\in \N$. Applying Theorem \ref{V.5} and Remark \ref{V.7}
we complete the proof.}

{  To prove the second statement we note that the operator
$B^K$ defined by \eqref{6.32x} is  bounded}. { Therefore, by
Proposition  \ref{prop3.1} } a triplet $\wh \gP_{B^K} := \{\cH,\wh
\gG^{B^K}_0,\wh \gG^{B^K}_1\}$ with  $\wh \gG^{B^K}_1 := \wh
\gP_0$ and $\wh \gG^{B^K}_0 := B^K\wh \gG_0 - \wh \gG_1$, is a
boundary triplet for $A^*$ such that $A^K :=
A^*\upharpoonright\ker(\wh \gG^{B^K}_0)$. The corresponding Weyl
function is
\bed
{ \wh M_{B^K}(z) = (B^K - \wh M(z))^{-1}}, \quad z \in \C_+.
\eed
Inserting expression \eqref{6.32x} into this formula we get
\bed { \wh M_{B^K}(z)} = -\frac{1}{\sqrt{2}}\frac{1}{\sqrt{T} +
i\sqrt{z-T}}\frac{1}{\sqrt{T + \sqrt{1 + T^2}}}=
\frac{1}{z\sqrt{2}}\frac{\sqrt{T} - i\sqrt{z - T}}{\sqrt{T +
\sqrt{1 + T^2}}}.
 \eed
It follows that the limit ${ \wh M_{B^K}(t+i0)}$ exists for  any
$t\in \R \setminus \{0\}$ and
\bed
{ \wh M_{B^K}(t)} := \slim_{y\to*0}M_{B^K}(t+iy) =
-\frac{1}{t\sqrt{2}}\frac{\sqrt{T} - i\sqrt{t - T}}{\sqrt{T + \sqrt{1 + T^2}}}.
\eed
Clearly, ${ \wh M_{B^K}(t)}\in [\cH]$ for any $t\in \R \setminus \{0\}.$
{ By  Theorem \ref{V.5} the  $ac$-parts of $\wt A$} and $A^K$
are unitarily equivalent whenever $(\wt A - i)^{-1} - (A^K -
i)^{-1} \in \gotS_\infty(\gotH).$ This completes the {
proof.}
\end{proof}

{ Finally, we generalize Corollary \ref{VI.5Z} to unbounded
  operator potentials.}
  \bc\label{cor6.12}
{ Assume conditions of Theorem \ref{VI.10a}. If $\dim(\cH) =
\infty$ and $t_0:=\inf\gs(T) = \inf\gs_{\ess}(T)=:t_1$, then the
Friedrichs  extension $A^F$ and the Krein extension $A^K$ are
strictly $ac$-minimal.}
   \ec

\subsection{Application}

In this subsection we apply previous results of this section to
Schroedinger operator in the half-plane. To this end we denote by
$L = L_{\rm min}$ the minimal elliptic operator associated in
$L^2(\Omega), \ \Omega:=\R_+\times \R^n,$ with the differential
expression
        \be\label{7.40}
\mathcal L := -\Delta  + q(x) := -\bigl(\frac{\partial^2}{\partial
t^2} +\sum_{j=1}^n\frac{\partial^2}{\partial x_j^2}\bigr) + q(x),
\quad (t,x)\in \Omega,
       \ee
where $q= \bar q \in L^\infty(\R),$ \ $x := (x_1, \ldots, x_n)$
and $n\ge 1.$

Recall that $L_{\min}$ is the closure of $\mathcal L$,
defined on $C^\infty_0(\Omega).$
It is known that $\dom(L_{\min}) = H^{2}_0(\Omega).$ Clearly, $L$
is symmetric. The maximal operator $L_{\max}$ is then defined by
$L_{\max} = (L_{\min})^*.$ We emphasize that $H^{2}(\Omega)\subset
\dom(A_{\max})\subset H^{2}_{loc}(\Omega)$, while
$\dom(L_{\max})\not = H^{2}(\Omega)$.

Next we define the trace mappings  $\gamma_j\colon
C^{\infty}({\overline{\Omega}})\to C^{\infty}(\partial\Omega), \
j\in \{0,1\},$ by setting  $\gamma_0 u := u \upharpoonright
{\partial\Omega}$  and  $\gamma_1 u := \gamma_0
({\partial u}/{\partial n})$
where $n$ stands for the interior normal to $\partial\Omega$.
Denote  by $D_{\mathcal L}(\Omega)$ the domain $\dom(L_{\max})$
equipped with the graph norm. It is known (see \cite{LioMag72,
Gru08}) that $\gamma_j$ can be extended by continuity to the
operators mapping $D_{\mathcal L}(\Omega)$ continuously onto
$H^{-j-1/2}(\partial\Omega), \ j\in \{0,1\}.$

 Let us define the
following  extensions of $L_{\min}$ (realizations of $\mathcal
L$):
\begin{enumerate}

\item[(i)] $L^Df := \mathcal L[f],$ \  $f \in \dom (L^D) :=
\{\varphi  \in
H^2(\R_+\times \R):\  \gamma_0 \varphi = 0\}$;

\item[(ii)] $L^Nf := \mathcal L [f],$\  $f \in \dom (L^N) :=
\{\varphi \in
H^2(\R_+\times \R): \;
\gamma_1 \varphi=0\}$;

\item[(iii)] $L^Kf := \mathcal L[f],$\  $f \in \dom (L^K) :=
\{\varphi \in \dom(A_{\max}): \gamma_1 \varphi + \Lambda \gamma_0
\varphi = 0\}, \quad \text{where} \quad   \Lambda := \sqrt{
-\Delta_x + q(\cdot)}:\   H^{-1/2}(\partial\Omega)\to
H^{-3/2}(\partial\Omega).$
\end{enumerate}

To treat the operator $L_{\min}$ as the Sturm-Liouville operator
with (unbounded) operator potential we denote by $T$  the (closed)
minimal operator associated in $\cH := L^2({\mathbb R}^n)$ with
the Schr\"odinger expression
\begin{equation}\label{7.41}
-\Delta_x  + q(x) := - \sum_{j=1}^n\frac{\partial^2}{\partial
x_j^2} + q(x).
\end{equation}
Since $q= \bar q \in L^\infty(\R),$  the operator $T$ is
self-adjoint and semibounded. Moreover, $T\ge 0$ if $q(\cdot)\ge0.$
Let $A := A_{T,\rm min}$ be the minimal operator
associated with \eqref{8.10} where $T$ is defined by \eqref{7.41}.
\bp\label{prop6.12}
Let $q(\cdot) \in L^\infty(\R),$ \ $q(\cdot) \ge 0$,  and let $T$
be the minimal (self-adjoint) operator associated in $L^2(\R)$ with
the Schr\"odinger expression  \eqref{7.41}.  Let also $t_0:=\inf
\sigma(T)$ and $t_1:=\inf \sigma_{\ess}(T).$
Then:

\item[\rm \;\;(i)] the  operator $A_{T,\rm min}$ coincides with
$L=L_{\rm min}$ and   $\dom(A_{T,\rm min})= H^2_0(\Omega);$

\item[\rm\;\;(ii)] the Friedrichs extension $A^F$ coincides with
$L^D$, hence $L^D$ is absolutely continuous,
$\sigma(L^D)=\sigma_{ac}(L^D) = [t_0,\infty)$ and
$N_{L^D}(t)=\infty$ for a.e. $t\in [t_0,\infty);$

\item[\rm\;\;(iii)] the Krein extension $A^K$ coincides with
$L^K$, in particular, $L^K$ admits the decomposition $L^K =
0_{\cH_0}\oplus (L^K)^{ac}$, $\cH_0 := \ker(L^K)$,  and
$\sigma_{ac}(L^K)=[t_0,\infty)$;

\item[\rm\;\;(iv)] the  extension $A^N$ defined by  \eqref{8.23},
coincides with $L^N$, in particular, $L^N$ is absolutely
continuous and $ \sigma(A^N)  = \sigma_{ac}(A^N)=[t_0, \infty);$

\item[\rm\;\;(v)] the operators $L^D$,\  $L^N,$\  and $(L^K)$ are
  $ac$-minimal, in particular, $L^D$,\  $L^N,$\  and
$(L^K)^{ac}$ are pairwise  unitarily equivalent. If, in addition,
$t_0=t_1,$ then  the operators $L^D$,\  $L^N,$\  and $(L^K)$ are
strictly $ac$-minimal;

\item[\rm\;\;(vi)] if ${\wt L}$
is a  self-adjoint extension of $L$ and  $({\wt
L}-i)^{-1}-(L^D-i)^{-1}\in {\mathfrak S}_{\infty},$ then ${\wt
L}^{ac}$ and $L^D$ are unitarily equivalent. If ${\wt L}$
satisfies $({\wt L}-i)^{-1}-(L^K-i)^{-1}\in  {\mathfrak
S}_{\infty}$,  then ${\wt L}^{ac}$ and $L^D$ are unitarily
equivalent.
\end{proposition}
\begin{proof}
(i) We introduce the set
\bed \cD'_\infty := \left\{\sum_{1\le j \le k} \phi_j(x)h_j(\xi):
\;\phi_j \in C^\infty_0(\R_+), \  h_j \in C^\infty_0(\R),\  k \in
\N\right\} \eed
We note that $\cD'_\infty \subseteq \cD'_0$ and $\cD'_\infty
\subseteq C^\infty_0(\R_+\times\R)$. Moreover, $A_{T,\rm min}
\upharpoonright\cD'_\infty = L\upharpoonright\cD'_\infty$.  Since
$\cD'_\infty$ is a core for both $A_{T,\rm min}$ and $L_{\rm
min},$ we have $A_{T,\rm min} = L_{\rm min}$.

It is clear (after applying the Fourier transform) that $\dom(T)
=\dom(\Delta_x) = H^2(\R^n)$. Therefore, by Lemma \ref{VI.2.7},
$\dom(A_{T,\rm min})= W^{2,2}_{0,T}=H^2_0(\Omega).$

\noindent (ii) Since $\dom(T) = H^2(\R^n),$ we have
$W^{2,2}_{T}(\R_+;\cH) = H^2(\Omega).$ Therefore by Proposition
\ref{prop6.8} $A^F=L^D.$ The second assertion follows from Theorem
\ref{VI.9}(i).

\noindent (iii)\  It is proved in \cite[Section 9.7]{DM91} that
$L^K$ is the Krein extension of $L_{\min}.$  The rest of the
statements is implied by Theorem \ref{VI.9}(ii).

\noindent (iv) The equality  $A^N = L^N$ is immediate from
Proposition \ref{prop6.8}(iii).  The second statement follows from
Theorem \ref{VI.9}(iii).

\noindent (v) By Theorem \ref{VI.10a}(ii) the extension
$L^D(=A^F)$ is $ac$-minimal. By Theorem \ref{VI.9}(iv) $L^D$,\
$L^N (=A^N),$\  and $(L^K)^{ac} (= (A^K)^{ac})$ are pairwise
unitarily equivalent, hence $L^N,$\  and $L^K$ are $ac$-minimal
too. The last statement is immediate from Corollary \ref{cor6.12}.

\noindent (vi) The statement is immediate  from
(ii), (iii) and  Theorem \ref{VI.10a}(iii).
\end{proof}
\begin{remark}
{\em
Let  $T$ be the (closed) minimal non-negative operator associated
in $\cH := L^2({\mathbb R}^n)$ with general uniformly elliptic
operator
\begin{equation}
- \sum_{j,k=1}^n \frac{\partial}{\partial x_j} a_{jk}(x)
\frac{\partial}{\partial x_j} + q(x), \quad a_{jk} \in
C^{1}({\overline\Omega}),\ \ q \in C({\overline\Omega})\cap
L^\infty({\Omega}),
\end{equation}
where the coefficients $a_{jk}(\cdot)$ are bounded with their
$C^1$-derivatives, $q\ge 0.$ If the coefficients have some
additional "good" properties, then $\dom(T)=H^2({\R}^n)$
algebraically and topologically.  By Lemma \ref{VI.2.7},
{ $\dom(A_{T,\min})= W^{2,2}_{0,T}(\R_+,\cH) =
H^{2,2}_0(\Omega)$ } and Proposition  \ref{prop6.12} remains valid with
$T$ in place of the Schr\"odinger operator \eqref{7.41}.

Note also that the Dirichlet and the Neumann realizations $L^D$
and $L^N$ are always self-adjoint ((cf. \cite[Theorem
2.8.1]{LioMag72}, \cite{Gru08})).
}
\end{remark}
\begin{corollary}\label{cor6.13}
Assume the conditions of Proposition  \ref{prop6.12}. If, in
addition,
\begin{equation}\label{6.33}
\lim_{|x|\to\infty}\int_{|x-y|\le 1}|q(y)|dy = 0,
\end{equation}
then the  operators $L^D$,\  $L^N,$\  and $(L^K)$ are strictly
$ac$-minimal,
\bed
\gs(L^D) = \gs_{ac}(L^D) = \gs_{ac}(L^K) = \gs(L^N) =\gs_{ac}(L^N)
= [0,\infty),
\eed
and $N_{E_{L^D}}(t) = N_{E_{L^N}}(t)= N_{E^{ac}_{L^K}}(t)=\infty$ for a.e.
$t\in [0,\infty).$
\end{corollary}
\begin{proof}
By \cite[Section 60]{Gla66}  condition  \eqref{6.33} yields the
equality $\sigma_c(T)= \R_+,$ in particular $0\in \sigma_c(T)$ and
$t_1=0.$ Since $q\ge 0,$\  we have $0\le t_0 \le t_1=0,$ that is
$t_0 = t_1=0.$ It remains to apply Proposition \ref{prop6.12}
(ii)-(v).
\end{proof}
\begin{remark}
{\em
Condition  \eqref{6.33} is satisfied whenever
$\lim_{|x|\to\infty}q(x) = 0.$  Thus,  in this case the
conclusions  of Corollary  \ref{cor6.13}  are valid.
However, it might happen that  $\gs(L^F) = \gs_{ac}(L^K) =
\gs(L^N) = [t_0,\infty)$, $t_0 > 0$ though  $\inf q(x) = 0$.
}
\end{remark}

\begin{appendix}

\section{Appendix: Absolutely continuous closure}

Let us recall some basic facts of the $ac$-closure of a
Borel set of $\R$ introduced in\cite{BMN02}, see also \cite{GMZ08}.  
\bd[\cite{BMN02}]
Let  $\delta\in \mathcal \cB(\R)$. The set $\cl_{ac}(\delta)$
defined by
\bed
\cl_{ac}(\delta) := \{x \in \R: \; |(x-\varepsilon, x +
\varepsilon) \cap \delta| > 0 \;\; \forall \; \varepsilon
> 0\}.
\eed
is called  the absolutely continuous closure  of the { Borel
set $\delta\in \mathcal B(\R)$.}
\ed

Obviously, two Borel { sets $\delta_1, \delta_2\in  \mathcal
B(\R)$ have the same $ac$-closure if their} symmetric difference
$\delta_1 \bigtriangleup \delta_2$ has Lebesgue measure zero.
Moreover, the set $\cl_{ac}(\delta)$ is always closed and
$\cl_{ac}(\delta) \subseteq \overline{\delta}$. In particular, if
we have two measurable non-negative functions $\xi_1$ and $\xi_2$
which differ only on a set of Lebesgue measure zero, then
$\cl_{ac}(\supp(\xi_1)) = \cl_{ac}(\supp(\xi_2))$.
\bl\la{III.6}
If $\delta\in \mathcal \cB(\R)$,  then $|\delta \setminus
\cl_{ac}(\delta)| = 0$.
\el
\begin{proof}
Since $\cl_{ac}(\gd)$ is closed the set $\gD := \R \setminus
\cl_{ac}(\gd)$ is open. The open set $\gD$ is decomposed as
$\gD = \bigcup^L_{l=1}\gD_l, \ 1 \le L \le \infty$,
where $\gD_l = (a_l,b_l)$ are disjoint open intervals. We set
$\gD_l = \gd \cap \gD_l$, $l =1,2,\dots,L$. Obviously,
\bed
\gd \setminus \cl_{ac}(\gd) = \gd \cap \gD = \bigcup^L_{l=1}\gD_l.
\eed
We note that $\gD_l \cap \cl_{ac}(\gd) = \emptyset$, $l
=1,2,\ldots,L$. Hence for each $t \in
\gD_l$ there is a sufficiently small neighborhood $\cO_t$ such
that $|\cO_t \cap \gd| = 0$. If $\eta$ is sufficiently small, then
$[a_l+\eta,a_l-\eta] \subseteq (a_l,b_l)$ and
$\{\cO_t\}_{t\in \gD_l}$ performs a covering of
$[a_l+\eta,a_l-\eta]$. Since $[a_l+\eta,a_l-\eta]$ is compact we can
chosen a finite covering $\{\cO_{t_m}\}^M_{m=1}$ of
$[a_l+\eta,a_l-\eta]$. By $[a_l+\eta,a_l-\eta] \subseteq
\bigcup^M_{m=1}\cO_{t_m}$ we find $|[a_l+\eta,a_l-\eta] \cap \gd| =
0$ for each sufficiently small $\eta > 0$. Using that we get
\bead
\lefteqn{
|(a_l,b_l) \cap \gd| = |(a_l,a_l+\eta) \cap \gd| + |(b_l-\eta,b_l) \cap \gd| = }\\
& &
|(a_l,a_l+\eta) \cap \gd| + |(b_l-\eta,b_l) \cap \gd| \le 2\eta
\eead
for sufficiently small $\eta >0$. Hence $|\gD_l| = |(a_l,b_l) \cap
\gd| = 0$ which yields that $|\gd \setminus \cl_{ac}(\gd)| = 0$.
\end{proof}
\bl\la{III.7}
If $\{\gd_k\}_{k\in\N}$, $\gd_k \subseteq \R$, is a sequence of
Borel subsets, then
\be\la{3.16}
\cl_{ac}(\gd) = \overline{\bigcup_{k\in\N}\cl_{ac}(\gd_k)},
\quad \gd = \bigcup_{k\in\N}\gd_k.
\ee
\el
\begin{proof}
We set $\wh \gd_k = \gd_k \cap \cl_{ac}(\gd_k)$ and $\gD_k := \gd_k
\setminus \cl_{ac}(\gd_k)$. We have
$\gd = \wh \gd \cup \gD,$  $\mbox{where}\quad \wh \gd :=
\bigcup_{k\in\N}\wh \gd_k$  $ \mbox{and} \quad \gD :=
\bigcup_{k\in\N}\gD_k.$
By Lemma \ref{III.6}, $|\gD_k| = 0$, $k \in \N$, which yields
$|\gD| = 0$. Hence $\cl_{ac}(\gd) = \cl_{ac}(\wh \gd)$.
Similarly one gets $\cl_{ac}(\gd_k) = \cl_{ac}(\wh \gd_k)$, $k \in
\N$. Notice that $\wh \gd_k \subseteq \cl_{ac}(\wh \gd_k)$, $k \in
\N$. We have
\bed
\cl_{ac}(\wh \gd) \supseteq \bigcup_{k\in\N} \cl_{ac}(\wh \gd_k) \supseteq
\bigcup_{k\in\N}\wh \gd_k = \wh \gd.
\eed
Hence
\bed
\cl_{ac}(\wh \gd) =\overline{\cl_{ac}(\wh \gd)} \supseteq
\overline{\bigcup_{k\in\N} \cl_{ac}(\wh \gd_k)}
\supseteq \overline{\wh \gd} \supseteq \cl_{ac}(\wh \gd)
\eed
which yields $\cl_{ac}(\wh \gd) = \overline{\bigcup_{k\in\N}
\cl_{ac}(\wh \gd_k)}$. Since $\cl_{ac}(\wh \gd) = \cl_{ac}(\gd)$ and
$\cl_{ac}(\wh \gd_k) = \cl_{ac}(\gd_k)$, $k \in \N,$ we prove
\eqref{3.16}.
\end{proof}

\end{appendix}

\section*{Acknowledgment}
The first author thanks  Weierstrass Institute of Applied Analysis
and Stochastics in Berlin for financial support and hospitality.


\def\cprime{$'$} \def\cprime{$'$} \def\cprime{$'$} \def\cprime{$'$}
  \def\lfhook#1{\setbox0=\hbox{#1}{\ooalign{\hidewidth
  \lower1.5ex\hbox{'}\hidewidth\crcr\unhbox0}}} \def\cprime{$'$}
  \def\cprime{$'$} \def\cprime{$'$} \def\cprime{$'$} \def\cprime{$'$}
  \def\cprime{$'$} \def\lfhook#1{\setbox0=\hbox{#1}{\ooalign{\hidewidth
  \lower1.5ex\hbox{'}\hidewidth\crcr\unhbox0}}}

\end{document}